\documentclass{jocg}

\usepackage{fancyhdr}
\newcommand{\cleanjocgheaders}{%
	\fancyhf{}%
	 \fancyfoot[C]{\thepage}%
	\renewcommand{\headrulewidth}{0pt}
	\renewcommand{\footrulewidth}{0pt}%
}
\cleanjocgheaders
\fancypagestyle{plain}{\cleanjocgheaders}
\AtBeginDocument{%
	\pagestyle{fancy}%
	\cleanjocgheaders%
}
\makeatletter
\renewenvironment{abstract}
{%
	\par\bigskip
	\hrule
	\vskip 1ex
	\noindent\textsc{Abstract.}\ %
}{%
	\par\vskip 2ex
	\hrule
	\bigskip
}
\makeatother

\usepackage{amsmath,amsfonts}

\usepackage[mathlines]{lineno}

\title{%
	\MakeUppercase{Sharp approximate Carathéodory theorem
		\newline and application  to iterated Delaunay refinement
	}%
}

\author{%
	Raphaël~Tinarrage%
	\thanks{\affil{Institute of Science and Technology Austria}, \email{raphael.tinarrage@ist.ac.at}}%
	\orcid{0000-0002-1404-1095}
}

\usepackage{amsmath,amssymb,amsthm}
\usepackage{thmtools}
\usepackage{thm-restate}
\theoremstyle{plain}
\newtheorem{theorem}{Theorem}
\declaretheorem[name=Proposition]{proposition}
\declaretheorem[name=Lemma]{lemma}
\declaretheorem[name=Corollary]{corollary}
\newtheorem{remark}{Remark}

\usepackage{float}
\usepackage{longtable}
\usepackage{cleveref}
\usepackage{subcaption}
\usepackage{microtype}
\usepackage{placeins} 
\usepackage{bm} 
\usepackage{array} 

\usepackage{tocloft}

\newcommand{\A}{\mathcal A}
\newcommand{\E}{\mathbb E}
\newcommand{\N}{\mathbb N}
\newcommand{\PP}{\mathbb P}
\newcommand{\R}{\mathbb R}
\newcommand{\geomreal}[1] { |#1| } 
\newcommand{\Del}[1] { \mathrm{Del}(#1) } 
\renewcommand{\d} { \mathrm{d} } 

\DeclareMathOperator{\sub}{sub}
\DeclareMathOperator{\conv}{conv}
\DeclareMathOperator{\circum}{circ}
\DeclareMathOperator{\maxcirc}{max\,circ}
\DeclareMathOperator{\maxdiam}{max\,diam}
\DeclareMathOperator{\diam}{diam}
\DeclareMathOperator{\covrad}{cov\,rad}
\DeclareMathOperator{\diag}{diag}
\DeclareMathOperator{\tr}{tr}
\DeclareMathOperator{\im}{im}
\DeclareMathOperator{\dist}{dist}
\DeclareMathOperator{\aff}{aff}
\DeclareMathOperator{\vol}{vol}
\DeclareMathOperator{\Span}{span}
\DeclareMathOperator{\Sym}{Sym}

\begin{document}

\maketitle

\begin{abstract}
	We analyze the decrease of simplex diameters under iterated refinement of spherical Delaunay complexes.
	Unlike in ordinary subdivision, the refined Delaunay complex need not be a subdivision of the previous one, so mesh contraction is not automatic. 
	We derive explicit contraction bounds for several families of Steiner points, including Delaunay analogues of barycentric and edgewise subdivision. 
	The proof reduces the problem to sharp covering estimates for Euclidean simplices.
	These estimates are obtained through a strengthening of Maurey's empirical method via pivotal sampling and a dimension-dependent version of the approximate Carathéodory theorem.
	Theoretical results and numerical experiments show that Delaunay refinements achieve stronger contraction than their subdivision counterparts.
\end{abstract}

\vspace{.5cm}
\setcounter{tocdepth}{1}
\tableofcontents
\newpage

\section{Introduction}

\subsection{Motivations}
\label{subsec:motivations}

\emph{Subdivision} is a fundamental operation in computational geometry and topology.
In geometry, a mesh is refined in order to represent a domain more accurately and to improve the performance of numerical methods \cite{rivara1984mesh,maubach1995local,plaza2000local,Frey_2008,Brenner_2008,cheng2013delaunay}.
In topology, subdivision refines a triangulation without changing its underlying space, allowing one to relate different triangulations through common subdivisions, to compute topological invariants, and to approximate continuous maps by simplicial maps \cite{Kaczynski_2004,Zomorodian_2005,Kozlov_2008,Edelsbrunner_2009,Ellis_2019,Dey_2022}.

The last point is our main motivation.
Let $K,L$ be finite simplicial complexes, and let $f\colon \geomreal{K}\to\geomreal{L}$ be a continuous map between their geometric realizations.
\emph{Simplicial approximation} produces a simplicial map $g\colon K'\to L$, defined on a subdivision $K'$ of $K$, such that the realization $\geomreal{g}\colon\geomreal{K'}\to \geomreal{L}$ is homotopic to $f$.
Thus, simplicial approximation is a way of ``triangulating a map'': it replaces a continuous map by a finite combinatorial object.

The \emph{simplicial approximation theorem} says that such an approximation exists once $K$ has been subdivided finely enough, i.e., once the edge lengths are sufficiently small; see \cite[Theorem~2C.1]{hatcher_algebraic_2002}.  
Thus, we need a refinement procedure that produces triangulations with arbitrarily small simplices.  
The standard choice is \emph{barycentric subdivision}.
However, it is often inefficient in practice because the number of vertices grows rapidly under iteration.

The purpose of this paper is to study an alternative refinement procedure based on \emph{Delaunay complexes}.
Starting from a point set, we insert Steiner points and recompute the Delaunay complex.
This procedure is closer in spirit to mesh generation than to ordinary subdivision.
But it raises a difficulty: recomputing the Delaunay complex alters the previous simplices.
In particular, the simplex diameters need not decrease at each refinement step.
We prove that several natural Delaunay refinements do shrink after sufficiently many steps (\Cref{cor:iterated_refinement}).
More precisely, we study four families of Steiner points: the \emph{minicenters of the facets}, the \emph{barycenters of the facets}, the \emph{equal-weight barycenters of all faces} (as in barycentric subdivision), and the \emph{$k$-fold barycenters of the facets} for $k\geq2$ (as in edgewise subdivision).

We restrict our attention to a particular setting: refinements of \emph{Euclidean} and \emph{spherical} Delaunay complexes. 
In particular, our results do not give a general refinement procedure for arbitrary abstract simplicial complexes.
Still, this is the setting needed in our companion paper, where we use simplicial approximation to obtain triangulations of manifolds \cite{tinarrage:LIPIcs.SoCG.2026.93}.

The proofs are based on sharp Euclidean covering estimates for a simplex.
For each refinement scheme, we ask how far a point inside the simplex can be from the old vertices together with the newly inserted Steiner points.
We prove the best possible dimension-dependent bounds (\Cref{thm:steiner_covering}).
For the equal-weight barycenters, we employ \emph{pivotal sampling} to strengthen the argument underlying \emph{Maurey's empirical method} (\Cref{lem:pivotal_sampling}).
For the $k$-fold barycenters, we prove a sharp version of the \emph{approximate Carathéodory theorem} (\Cref{thm:steiner_covering_kfold}).

These results may help guide refinement strategies in computational settings such as mesh generation, where both contraction rate and mesh quality matter.
The final section illustrates this experimentally: by recomputing the Delaunay complex, the refinement can counteract the accumulation of poorly shaped simplices under iteration.

\subsection{Related work}

\subsubsection{Classical subdivision methods}
\label{subsec:classical_subdivision_methods}

The classical way to refine a simplicial complex is \emph{barycentric subdivision}.
It works in every dimension and gives a uniform decrease of the mesh size.
If $\sigma$ is a $d$-simplex and $\sub(\sigma)$ denotes its barycentric subdivision, then the maximal diameter of its simplices satisfies
\begin{equation}
	\label{eq:shrinking_classical_barycentric}
	\maxdiam(\sub(\sigma)) \leq \frac{d}{d+1} \diam(\sigma).
\end{equation}
Thus, after $n$ iterations of barycentric subdivision, the simplicial complex $\sub^n(\sigma)$ satisfies
\[
	\maxdiam(\sub^n(\sigma)) \leq \bigg(\frac{d}{d+1}\bigg)^n \diam(\sigma).
\]
In particular, all edge lengths can be made arbitrarily small.
The drawback is combinatorial: $\sub(\sigma)$ is a simplicial complex with $2^{d+1}-1$ vertices---one barycenter for each face of $\sigma$---and $(d+1)!$ facets.
Thus repeated barycentric subdivision quickly becomes expensive.

A more economical alternative is the \emph{Coxeter--Freudenthal--Kuhn subdivision} of simplices, also called \emph{edgewise subdivision} \cite{freudenthal1942simplizialzerlegungen,kuhn1960combinatorial,edelsbrunner2000edgewise,bey2000simplicial,goncalves2006simples}.
It depends on a parameter $k\geq2$ and creates a simplicial complex, still denoted $\sub(\sigma)$, with $\binom{d+k}{k}$ vertices.
More precisely, $k$-edgewise subdivision introduces all the $k$-fold barycenters of $\sigma=\{v_0,\dots,v_d\}$, by which we mean the equal-weight barycenters of all $k$-tuples $(v_{i_1},\dots,v_{i_k})$, repetitions allowed.
The simplices of $\sub(\sigma)$ are obtained from the Freudenthal--Kuhn triangulation of the unit cube, equivalently from the Coxeter hyperplane arrangement; see \cite[Appendix~A.1]{boissonnat_et_al:LIPIcs.SoCG.2021.17} for details.
This popular construction appears in several computational settings \cite{goncalves2007H2,weiss2011simplex,ChoudharyKachanovichWintraecken2020,BoissonnatKachanovichWintraeckenCFKCompact,boissonnat_et_al:LIPIcs.SoCG.2021.17,boissonnat2023tracing,brunck2023iterated,BoissonnatEtAlSimplicialSubdivisionCurved,manin2024algorithmic}.

For instance, $2$-edgewise subdivision introduces the midpoints of the edges of $\sigma$.
It subdivides a $2$-simplex into four homothetic triangles, and a $3$-simplex into eight (not all homothetic) tetrahedra; see \Cref{fig:edgewise_subdivision}.
In general, $k$-edgewise subdivision satisfies the estimate
\begin{equation}
	\label{eq:shrinking_classical_edgewise}
	\maxdiam(\sub(\sigma)) \leq \min\bigg\{1, \frac{\lfloor\frac{d+1}{2}\rfloor}{k}\bigg\} \diam(\sigma).
\end{equation}
This bound is sharp.
We have not found it in the literature and prove it in \Cref{subsec:additional_proofs_shrinking_classical_edgewise}.

When $\lfloor\frac{d+1}{2}\rfloor/k \geq 1$, $k$-edgewise subdivision does not guarantee contraction by an absolute constant smaller than $1$.
This obstruction already appears for tetrahedra with $k=2$: inserting the edge midpoints creates a central octahedron, which must be triangulated by choosing a diagonal.
The chosen diagonal can have length arbitrarily close to $\diam(\sigma)$.

A standard way to control the resulting diameters is to split the central octahedron along its shortest diagonal, as in \emph{red refinement} and in the related shortest-interior-edge partition \cite{zhang1995successive,kroger2008stability,korotov2022degenerating}.
This yields the contraction factor $1/\sqrt{2}$.
Related refinement schemes were studied in \cite{ong1994uniform,liu1996quality}.
However, in higher dimensions, the central cell is a hypersimplex rather than an octahedron, and there is no single diagonal to choose.
The literature instead focuses on the congruence classes of simplices produced by the subdivision \cite{bey2000simplicial,korotov2014red}.

In this article, we study a different refinement strategy based on Delaunay complexes.
The resulting contraction bounds, collected in \Cref{table:constant_alpha}, are significantly sharper.
For instance, the Delaunay analogue of $2$-edgewise subdivision attains the factor $\sqrt{(d-1)/(2d)}$ for $d\geq2$, improving on the barycentric bound $d/(d+1)$, on the $2$-edgewise bound $\lfloor\frac{d+1}{2}\rfloor/2$ and, in dimension $d=3$, on the red refinement bound $1/\sqrt{2}$.

\begin{figure}[!htbp]
	\centering
	\begin{subfigure}[t]{0.28\linewidth}
		\centering
		\includegraphics[width=.9\linewidth]{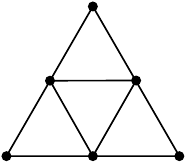}
		\subcaption{A triangle is subdivided into four homothetic triangles.}
	\end{subfigure}
	\hfill
	\begin{subfigure}[t]{0.69\linewidth}
		\centering
		\includegraphics[width=.99\linewidth]{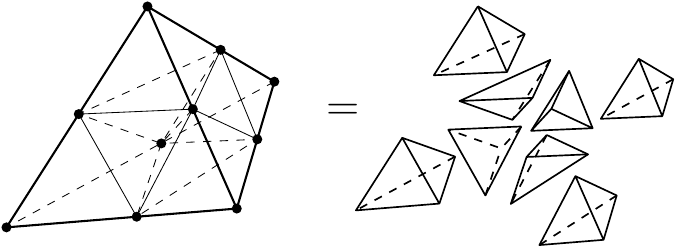}
		\subcaption{A tetrahedron is subdivided into eight tetrahedra: four corner tetrahedra and four tetrahedra triangulating the central octahedron.}
	\end{subfigure}
	\caption{
		$2$-Edgewise subdivision of a simplex uses the original vertices and the edge midpoints. 
		In dimensions $3$ and higher, the central hypersimplex must still be triangulated. 
		For a tetrahedron, this amounts to choosing a diagonal of the central octahedron.
		}
	\label{fig:edgewise_subdivision}
\end{figure}

\subsubsection{Delaunay refinement}

A central goal in mesh generation is to produce triangulations with well-shaped (non-flat) simplices, since this typically improves the performance of subsequent algorithms.
This is one reason why we work with \emph{Delaunay complexes}. 
Indeed, Delaunay triangulations enjoy several optimality properties: in $\R^2$ they maximize the minimum angle over all triangulations of a fixed point set \cite{sibson1978locally}; in $\R^d$ they minimize the maximal miniradius (enclosing-ball radius) of the simplices \cite{rajan1991optimality,Rajan1994}; and they also minimize a certain weighted sum of edge lengths \cite{musin1997properties}.

Delaunay triangulations come with a natural \emph{refinement} scheme: insert new points, called \emph{Steiner points}, and recompute the Delaunay triangulation. 
This idea is at the core of many mesh generation algorithms; see
\cite{cheng2013delaunay} for a modern treatment and \cite{shewchuk2002delaunay} for a detailed analysis in the planar case. 
The classical algorithms of Chew and Ruppert \cite{chew1989guaranteed,chew1993guaranteed,ruppert1995delaunay} start from an input domain and iteratively eliminate bad elements by inserting their circumcenters.
A large part of the literature refines this basic point-insertion strategy \cite{Edelsbrunner_2001,Edelsbrunner_2002,shewchuk2002delaunay,miller2005when,ungor2004offcenters,erten2009quality}. 

In three dimensions and higher, Delaunay refinement is complicated by the appearance of \emph{slivers}, which are poorly shaped tetrahedra despite having reasonable edge lengths.
This issue must be addressed in tetrahedral meshing algorithms \cite{cheng2000sliver,shewchuk1998tetrahedral,cheng2010piecewise}.
Delaunay refinement has also been extended to surfaces and domains with curved boundaries via \emph{restricted Delaunay triangulations} \cite{Boissonnat_2005,Boissonnat_2006,oudot2010meshing,rineau2007generic}, and implemented in \texttt{CGAL} \cite{jamin2015cgalmesh,cgal:eb-25b}.

Our objective is different from this classical mesh-generation setting.
Rather than refining only bad elements, we use Delaunay refinement as a global
``subdivision'' scheme to produce arbitrarily fine complexes.
We compare several natural choices of Steiner points and derive explicit contraction
factors for each of them.
The analysis reduces to sharp covering estimates for simplices, with different
Steiner points requiring different techniques.
In the $k$-fold edgewise case, this leads to a sharp form of the approximate
Carathéodory theorem.

\subsubsection{Approximate Carathéodory theorem}
\label{subsubsec:approximate_caratheodory}

\emph{Carathéodory's theorem} says that every point in the convex hull of a subset
$X\subset\R^d$ can be written as a convex combination of at most $d+1$ points of $X$.  
The \emph{approximate Carathéodory theorem} gives an approximation by barycenters.
If $X=\{x_1,\dots,x_N\}$ lies in the Euclidean unit ball $B^d\subset\R^d$ and if
$y\in \conv(X)$, then for every integer $k\geq 1$ there are indices $i_1,\dots,i_k\in\{1,\dots,N\}$ with repetitions allowed such that
\begin{equation}
	\label{eq:classical_approximate_caratheodory}
	\left\|
	y-\frac1k\sum_{\ell=1}^k x_{i_\ell}
	\right\|
	\leq \frac1{\sqrt{k}}.
\end{equation}
This estimate is usually proved via \emph{Maurey's empirical method}: one samples $k$ independent points $Y_1,\dots,Y_k\in X$ with expectation $y$, and then uses the
variance bound
\begin{equation}
	\label{eq:maurey_variance_bound}
	\E\left\|
	y-\frac1k\sum_{i=1}^k Y_i
	\right\|^2
	\leq \frac1k.
\end{equation}
We refer to Pisier~\cite{Pisier1981}, Carl~\cite{Carl1985}, and Carl--Pajor~\cite{CarlPajor1988} for early uses of this method, and to
Barman~\cite{Barman2018} for its algorithmic consequences. 
A modern exposition is given in \cite{Vershynin2026}.

The approximate Carathéodory theorem is used to replace convex combinations by sparse and equal-weight combinations.
It has applications to convex geometry, coreset construction, sparse approximation, and optimization \cite{Ivanov2021b,BlumHarPeledRaichel2016,Barman2018,KerdreuxColinDAspremont2023}. Related bounds and constructive proofs have been obtained using sampling without replacement \cite{KerdreuxColinDAspremont2023}, discrepancy theory \cite{ReisRothvoss2022}, greedy derandomization \cite{Ivanov2021}, mirror descent \cite{MirrokniLemeVladuWong2017}, and Frank--Wolfe methods \cite{CombettesPokutta2023}.

These works mainly target dimension-free estimates. 
In that sense, the bound $1/\sqrt{k}$ in \Cref{eq:classical_approximate_caratheodory} is the best possible. 
However, our goal is to obtain the optimal \textit{dimension-dependent} constant.
We prove in \Cref{thm:steiner_covering_kfold} that $1/\sqrt{k}$ can
be replaced by
\[
	R_{d,k}=\sqrt{\frac{s(d+1-s)}{dk^2}},
	\qquad
	s=\min\left\{k,\left\lfloor\frac{d+1}{2}\right\rfloor\right\}.
\]
This constant is optimal, with equality attained by regular simplices for $k\geq2$.
The proof combines an optimization argument with a delicate case analysis.
As an illustration, the parameter $k=2$, which corresponds to $2$-edgewise refinement, gives, for $d\geq2$,
\[
	R_{d,2}=\sqrt{\frac{d-1}{2d}} < \frac1{\sqrt{2}}.
\]

As a further consequence, the theorem settles some cases of a conjecture of
Bárány and Füredi on cylinder coverings; see \Cref{subsec:link_with_conjecture}.

One unsuccessful approach to proving this bound was to refine Maurey's probabilistic method by using different sampling distributions. 
A promising candidate was \emph{pivotal sampling}, introduced by Deville and Tillé~\cite{deville1998unequal}. 
Although this approach did not yield the result, it does give the sharp bound needed for the barycentric Delaunay refinement (\Cref{prop:steiner_covering_barycentric}). 
The proof is based on \Cref{lem:pivotal_sampling}, which follows from a variance estimate of Chauvet and Ruiz-Gazen~\cite{ChauvetRuizGazen2017} and can be viewed as a strengthening of \Cref{eq:maurey_variance_bound}.

\subsection{Overview}

We begin in \Cref{sec:delaunay} by introducing the objects of study and stating \Cref{thm:steiner_covering}, our main simplex-wise estimate.
It gives sharp covering bounds for a simplex after inserting the prescribed Steiner points. 
In \Cref{thm:shrinkingrefinements}, we transfer these Euclidean estimates to spherical Delaunay complexes.
This yields \Cref{cor:iterated_refinement}: under a mild density assumption, the covering radius decays geometrically, and consequently the maximal simplex diameter tends to zero.

The proof of \Cref{thm:steiner_covering} is split between \Cref{sec:sharp_bounds_centroid_minicenter_barycentric,sec:sharp_bounds_kfold}.
The first of these sections treats three types of Steiner points: minicenters of facets (\Cref{prop:steiner_covering_minicenter}), barycenters of facets (\Cref{prop:steiner_covering_centroid}), and equal-weight barycenters of faces (\Cref{prop:steiner_covering_barycentric}). 
The second section treats the family of $k$-fold edgewise refinements.
The estimate is formulated as an approximate Carathéodory theorem in \Cref{thm:steiner_covering_kfold}.

We conclude with experiments in \Cref{sec:experiments}, where we compare the different refinement schemes in terms of edge length, circumradius, and number of vertices. 
The experiments suggest that no single method is uniformly optimal: the best choice depends on the dimension.

Finally, \Cref{sec:notation} summarizes the notation used in the paper, and \Cref{sec:additional_proofs} contains additional proofs, including the edgewise subdivision bound stated in \Cref{eq:shrinking_classical_edgewise}.

\section{Successive refinements of spherical Delaunay triangulations}
\label{sec:delaunay}

We begin by defining the spherical Delaunay refinement schemes studied in this paper and establishing their contraction properties. 
The argument reduces the spherical problem to Euclidean covering estimates for simplices, which are proved in the following sections.

\subsection{Euclidean and spherical Delaunay triangulations}
\label{subsec:delaunay_triangulations}

Let $X \subset \R^d$ be finite.
A subset of $d+1$ points has the \emph{empty circumsphere property} if its circumscribing open ball is empty of points of $X$.
These subsets form the facets of an abstract simplicial complex $\Del{X}$, called the \emph{Delaunay complex}.
Under the genericity assumption that no $d+2$ points lie on the same sphere, $\Del{X}$ is naturally embedded in $\R^d$~\cite{boissonnat2018geometric}.

These definitions extend to the spherical setting: given a subset $X$ of the unit sphere $S^d\subset\R^{d+1}$, the empty circumsphere property is understood with respect to the geodesic distance on $S^d$; the resulting complex is called the \emph{spherical Delaunay complex} and is still denoted $\Del{X}$.
If no $d+2$ points of $X$ lie on the same geodesic sphere, then it is naturally embedded in $\R^{d+1}$.
All point sets considered below are assumed to satisfy this condition.

A natural map $\geomreal{\Del{X}}\to S^d$ is given by the scaling $x\mapsto x/\|x\|$, which is a homeomorphism provided that the origin lies in the interior of the convex hull of $X$.
A finite set $X$ satisfying this assumption will be referred to as \emph{admissible}, and $\Del{X}$ as an \emph{admissible triangulation}.
The map $x\mapsto x/\|x\|$ will be called the \emph{radial projection}.

It is well known that the spherical Delaunay complex coincides with the boundary of the convex hull of its vertices, thus reducing the computation of $\Del{X}$ to $\conv(X)$ \cite{b-gtfga-80}.
In our experiments, we used \texttt{Qhull}, a popular software package for computing convex hulls \cite{barber1996quickhull}.

\subsection{Global Delaunay refinements}
\label{subsec:global_delaunay_refinements}

Given a spherical Delaunay triangulation $\Del{X_0}$ with vertex set $X_0\subset S^d$, one obtains a new triangulation by choosing additional points $Y_0\subset S^d$ and building the Delaunay complex on $X_1 = X_0 \cup Y_0$.
The new points are called \emph{Steiner points}, and this procedure is known as a \emph{Delaunay refinement}.
This construction can be iterated, yielding a sequence of Steiner points $Y_0$, $Y_1$, $Y_2$, $\dots$ and complexes $\Del{X_0}$, $\Del{X_1}$, $\Del{X_2}$, $\dots$

Whether the complexes become progressively finer depends on the choice of Steiner points. 
We therefore consider several families of Delaunay refinements, illustrated in \Cref{fig:delaunay_refinements}.
\begin{description}
	\item[Minicenter refinement:]
		We first consider a variant inspired by classical Delaunay refinement. For each facet (i.e., maximal simplex) of $\Del{X_i}$, we insert its \textit{minicenter}, that is, the center of its minimum enclosing ball.
		We let $Y_i$ be the set of all such points.
		We favor minicenters over circumcenters because they always lie inside the simplex that defines them, which better reflects the local nature of classical subdivision.
	
	\item[Centroid refinement:]
		We may instead insert the \textit{barycenters} of the facets of $\Del{X_i}$, standing as a simpler alternative to minicenter refinement.
		Explicitly, we insert
		\[
			Y_i= \left\{ \frac{1}{d+1}\sum_{\ell=0}^d v_\ell \;\middle|\; \sigma=[v_0,\dots,v_d]\in\Del{X_i}\right\}.
		\]
	
	\item[Barycentric refinement:]
		Inspired by barycentric subdivision, we let the Steiner points be the equal-weight barycenters of all positive-dimensional simplices of $\Del{X_i}$.
		Explicitly,
		\[
			Y_i= \left\{ \frac{1}{k+1}\sum_{\ell=0}^k v_\ell \;\middle|\; \sigma=[v_0,\dots,v_k]\in\Del{X_i},\; 1\leq k\leq d\right\}.
		\]
	
	\item[$k$-Edgewise refinement:]
		Finally, in analogy with edgewise subdivision, we consider inserting the $k$-fold barycenters of the facets of $\Del{X_i}$, where $k\geq2$ is a parameter.
		For a facet $\sigma=[v_0,\dots,v_d]$, these are the equal-weight barycenters of all $k$-tuples of vertices of $\sigma$ with repetitions allowed.
		Explicitly, the Steiner points are
		\[
			Y_i= \left\{ \frac{1}{k}\sum_{\ell=1}^k v_{i_\ell} \;\middle|\; \sigma=[v_0,\dots,v_d]\in\Del{X_i},\ i_1,\dots,i_k\in\{0,\dots,d\}\right\} \setminus X_i.
		\]
\end{description}

In all cases, the Steiner points (minicenters or barycenters) are computed in the ambient space $\R^{d+1}$ and projected onto the sphere.
For minicenters and edge midpoints, this agrees with computing them intrinsically on the sphere; see Fiedler’s book \cite[Section 5.4]{fiedler2011matrices}.

\begin{figure}[!htbp]
	\begin{subfigure}[t]{1\linewidth}
		\centering
		\begin{minipage}{0.19\linewidth}
			\centering
			\includegraphics[width=.99\linewidth]{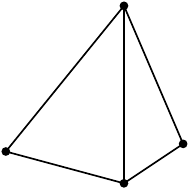}
			Initial complex
		\end{minipage}
		\begin{minipage}{0.19\linewidth}
			\centering
			\includegraphics[width=.99\linewidth]{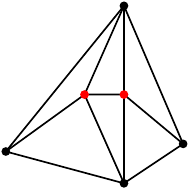}
			Minicenter
		\end{minipage}
		\begin{minipage}{0.19\linewidth}
			\centering
			\includegraphics[width=.99\linewidth]{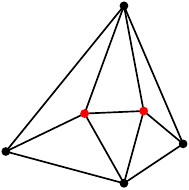}
			Centroid
		\end{minipage}
		\begin{minipage}{0.19\linewidth}
			\centering
			\includegraphics[width=.99\linewidth]{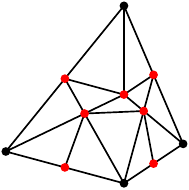}
			Barycentric
		\end{minipage}
		\begin{minipage}{0.19\linewidth}
			\centering
			\includegraphics[width=.99\linewidth]{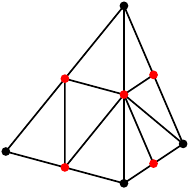}
			$2$-Edgewise
		\end{minipage}
		\subcaption{%
			Local picture: A Delaunay complex on a set of four points in $\R^2$ and the four proposed refinements.
			}
		\label{subfig:delaunay_refinements:local}
	\end{subfigure}
	\begin{subfigure}[t]{1\linewidth}
		\centering
		\begin{minipage}{0.19\linewidth}
			\centering
			\includegraphics[width=.99\linewidth]{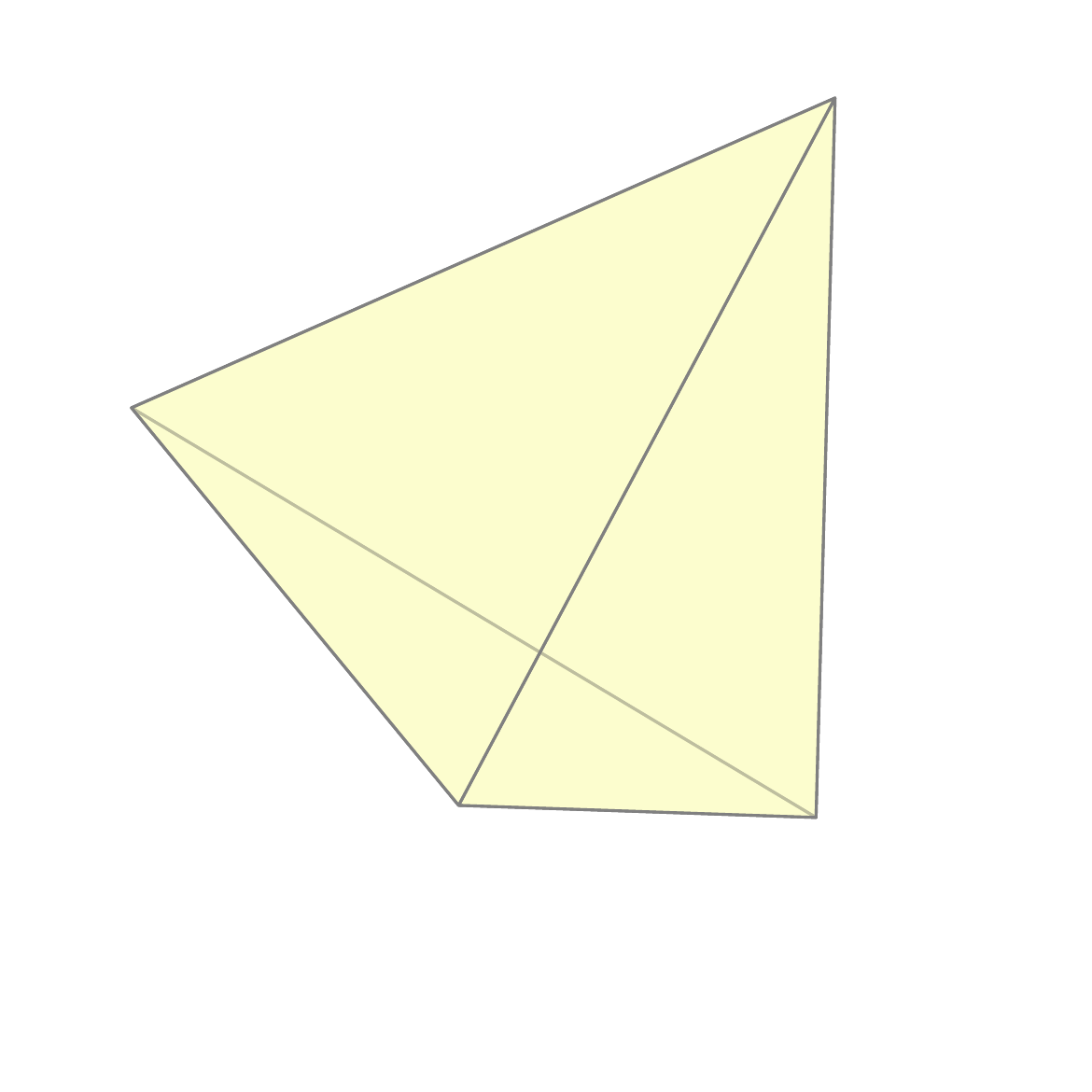}
			Initial complex
		\end{minipage}
		\begin{minipage}{0.19\linewidth}
			\centering
			\includegraphics[width=.99\linewidth]{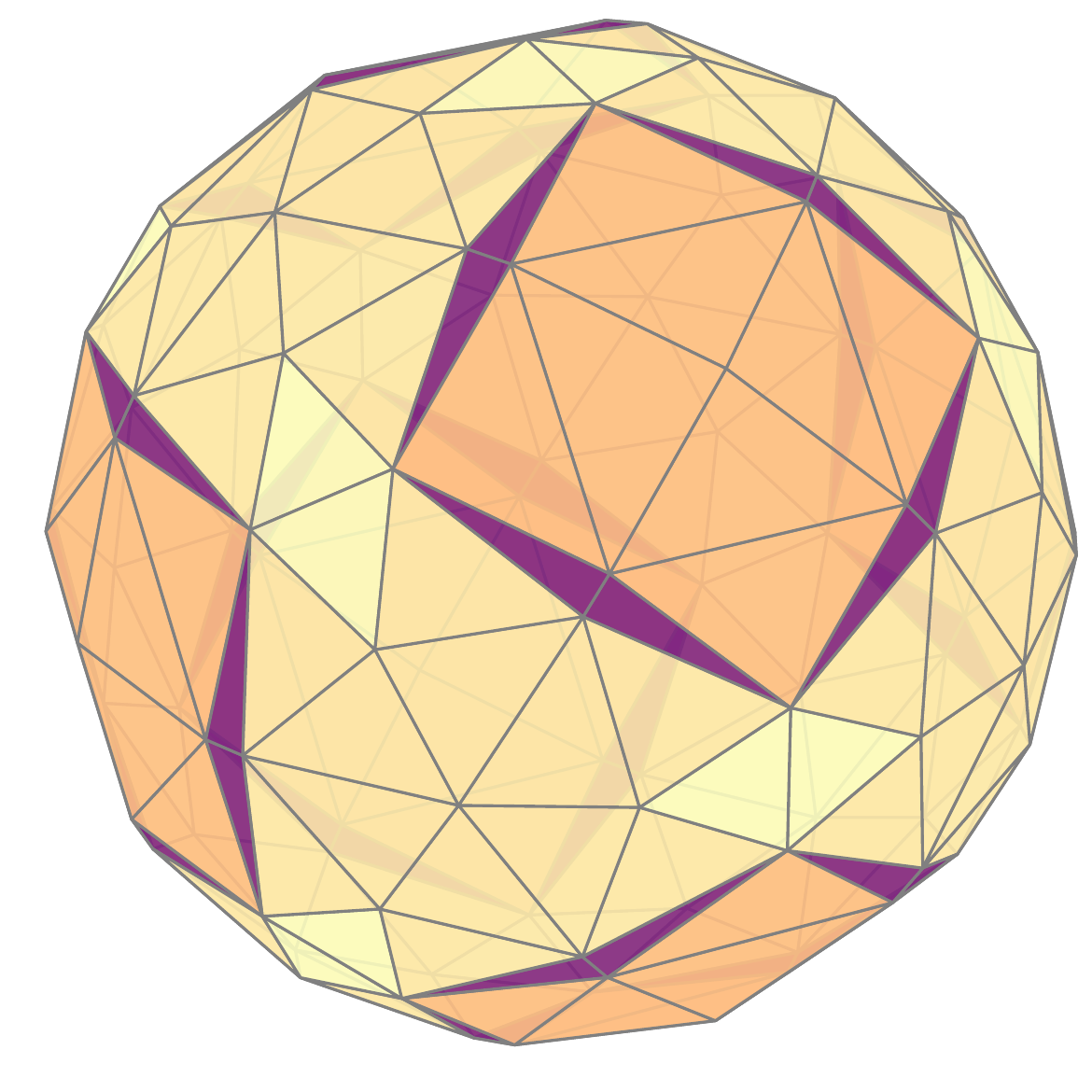}
			Minicenter
		\end{minipage}
		\begin{minipage}{0.19\linewidth}
			\centering
			\includegraphics[width=.99\linewidth]{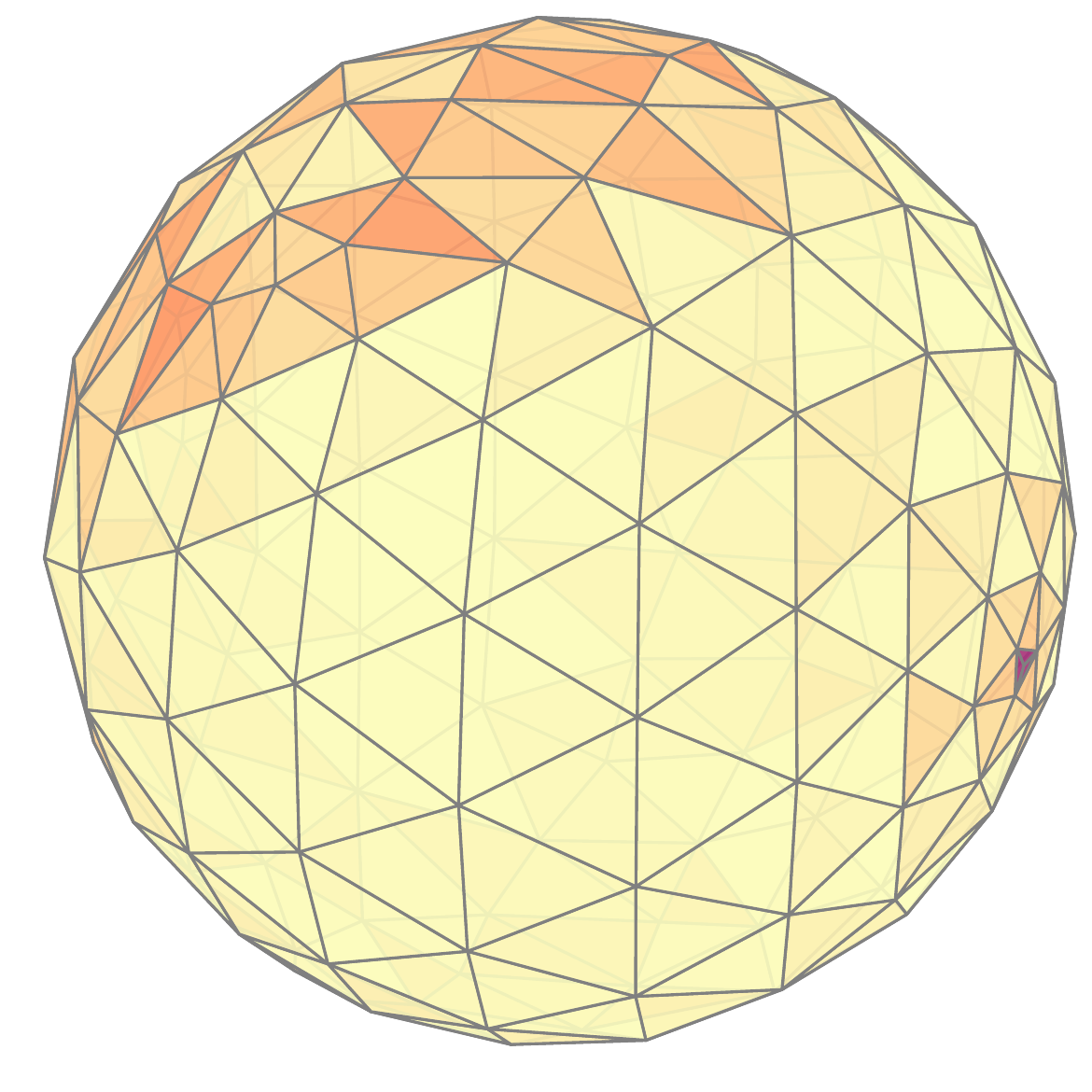}
			Centroid
		\end{minipage}
		\begin{minipage}{0.19\linewidth}
			\centering
			\includegraphics[width=.99\linewidth]{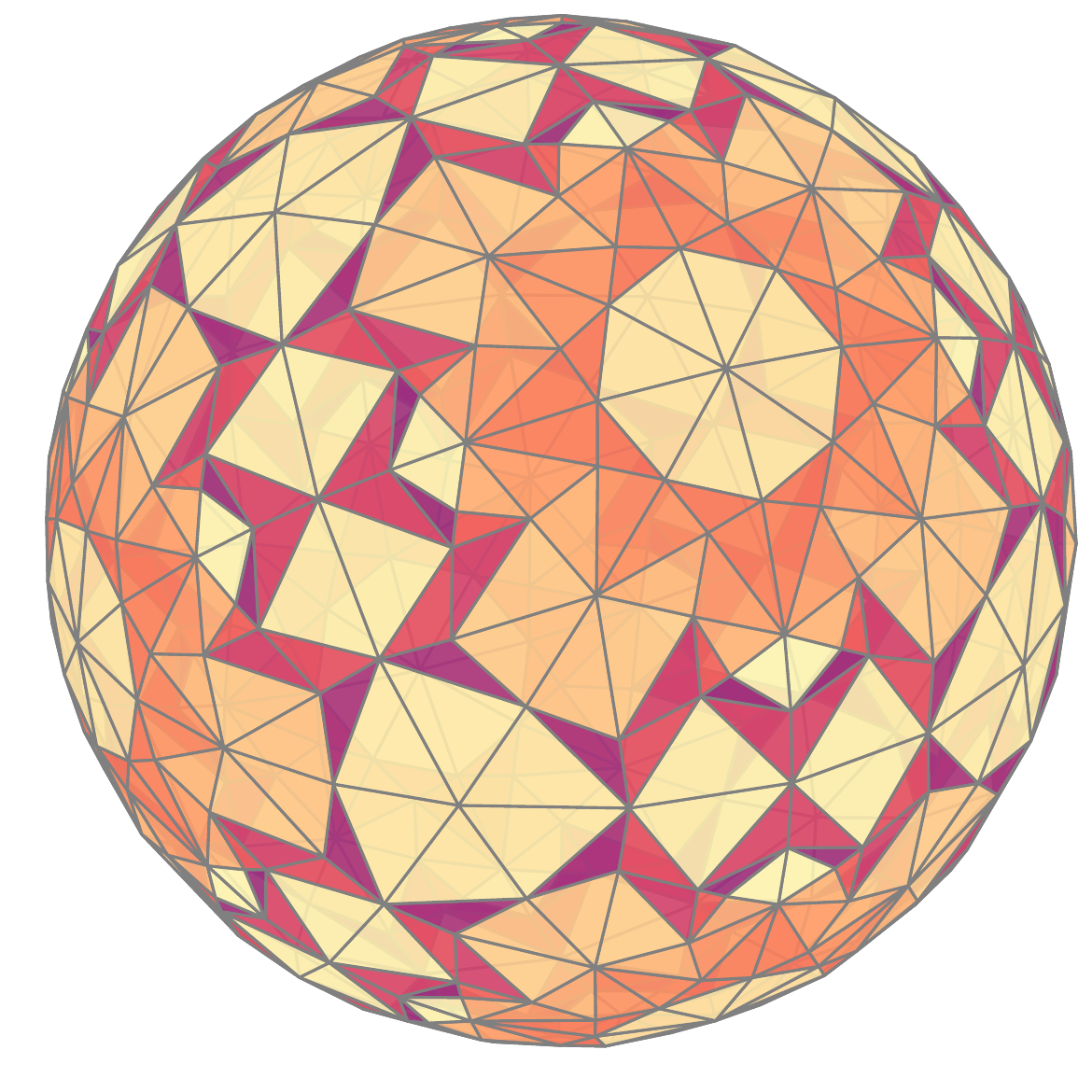}
			Barycentric
		\end{minipage}
		\begin{minipage}{0.19\linewidth}
			\centering
			\includegraphics[width=.99\linewidth]{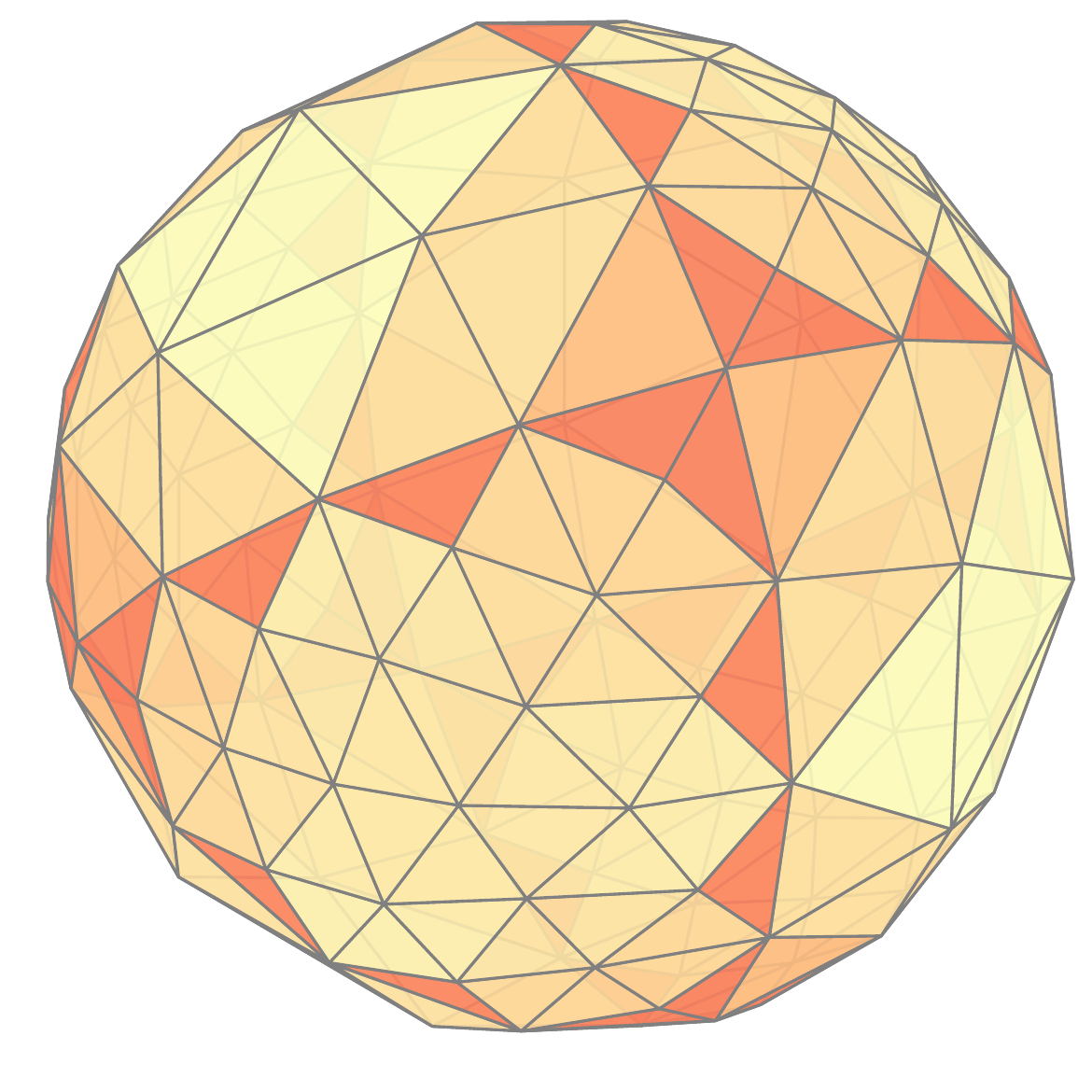}
			$2$-Edgewise
		\end{minipage}
		\subcaption{%
			Global picture: Spherical Delaunay complexes after three or four refinements of an initial triangulation of $S^2$ (the boundary of the standard simplex).
			The colors indicate the simplex ratio inradius/circumradius.
		}
	\end{subfigure}
	\caption{
		We employ Delaunay refinement as a ``subdivision'' scheme for triangulations of $S^d$.
	}
	\label{fig:delaunay_refinements}
\end{figure}

\subsection{Contraction under refinement}
\label{subsec:shrinking_property}

Note that, strictly speaking, $\Del{X_{i+1}}$ is not a subdivision of $\Del{X_i}$: when realized on the sphere, a simplex of $\Del{X_{i+1}}$ need not be contained in a single simplex of $\Del{X_{i}}$.
This can be seen in \Cref{subfig:delaunay_refinements:local}, where only the minicenter and edgewise refinements happen to be subdivisions.
As a consequence, it is not immediate that simplices become smaller under refinement.
In particular, adding a point to a Delaunay triangulation can increase the maximal edge length, as seen in \Cref{fig:delaunay_length_increase} (drawn in the plane for clarity).
Instead, we show that the \emph{maximal spherical circumradius} of $\Del{X_i}$, denoted $\maxcirc(\Del{X_i})$, decreases monotonically to zero.
Since the maximal diameter of the simplices of $\Del{X_i}$ is at most twice its maximal circumradius, it follows that the maximal diameter also tends to zero.

\begin{figure}[!htbp]
	\centering
	\includegraphics[width=.55\linewidth]{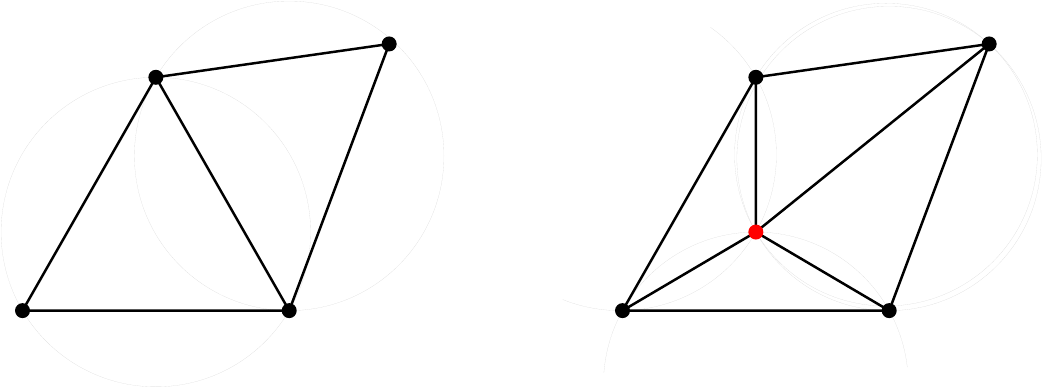}
	\caption{Adding a point to a Delaunay complex may increase the maximal edge length.
	}
	\label{fig:delaunay_length_increase}
\end{figure}

A convenient quantity for studying a Delaunay complex on a subset $X\subset S^d$ is its \textit{covering radius} (also known as sampling radius) \cite{boissonnat2014delaunay,boissonnat2018geometric}.
On the sphere, it is defined by
\begin{equation}
	\label{eq:def_covering_radius}
	\covrad(X) = \sup_{y \in S^d} \inf_{x \in X \vphantom{S^d}} \d(x,y),	
\end{equation}
where $\d(x,y)$ is the geodesic (great-circle) distance on the sphere.
We note that the set $X$ is admissible as long as $\covrad(X)<\pi/2$ (i.e., $X$ is not contained in a closed hemisphere).

The following standard observation relates this quantity to the maximal circumradius of the spherical Delaunay complex.

\begin{lemma}
	\label{lem:equality_covering_circumradius}
	For every admissible subset $X\subset S^d$, $\maxcirc(\Del{X}) = \covrad(X)$.
\end{lemma}

Thus, to prove that the maximal circumradius decreases under refinement, it suffices to control the covering radius. 
We proceed in two steps.
First, we prove a Euclidean estimate on each simplex before radial projection. 
The proof is technical and is split between \Cref{sec:sharp_bounds_centroid_minicenter_barycentric,sec:sharp_bounds_kfold}.
In \Cref{sec:sharp_bounds_kfold}, we also explain how it relates to the approximate Carathéodory theorem.

\begin{theorem}[Covering theorem]
	\label{thm:steiner_covering}
	Consider a simplex $\sigma=\{v_0,\dots,v_{d}\}\subset\R^{d}$ contained in the closed unit ball, where $d\geq2$.
	Let $Y_\mathrm{Euc}\subset\R^d$ denote the Euclidean Steiner points associated with $\sigma$ (the minicenter, the centroid, the equal-weight barycenters, or the $k$-fold barycenters) before projection onto the sphere.
	For every $y\in\conv(\sigma)$, there exists $x\in \sigma\cup Y_\mathrm{Euc}$ such that
	\[
	\|x-y\| \leq \alpha,
	\]
	where $\alpha$ depends on the chosen refinement, as given in \Cref{table:constant_alpha}.
	These constants are optimal.
	More generally, if $\sigma$ lies in a ball of radius $\delta\geq0$, then the right-hand side is multiplied by $\delta$.
\end{theorem}

\begin{proof}
	The four cases follow, respectively, from \Cref{prop:steiner_covering_minicenter,prop:steiner_covering_centroid,prop:steiner_covering_barycentric} and \Cref{thm:steiner_covering_kfold}.
\end{proof}

\begin{table}[H]
	\centering
	\renewcommand{\arraystretch}{1.5}
	\setlength{\tabcolsep}{4pt}
	\begin{tabular}{|c|c|c|c|c|}
		\hline
		Refinement & Minicenter & Centroid & Barycentric & $k$-Edgewise \\ \hline
		$\displaystyle \alpha$
		\rule[-3.5em]{0pt}{7.5em}
		& 
		$\displaystyle \frac{1}{\sqrt{2}}$
		&
		\renewcommand{\arraystretch}{2.5}
		$\left\{
		\begin{array}{ll}
			\displaystyle\frac23 & \text{if } d=2,\\[.1mm]
			\displaystyle\frac{d^2-1}{d^2+1} & \text{if } d\geq3
		\end{array}
		\right.$
		& 
		$\displaystyle \frac{1}{2}$
		&
		\renewcommand{\arraystretch}{2.5}
		\begin{tabular}{c}
			$\displaystyle \sqrt{\frac{s(d+1-s)}{dk^2}}$, \\[.1mm]
			$\displaystyle s=\min\left\{k,\left\lfloor\frac{d+1}{2}\right\rfloor\right\}$
		\end{tabular}
		\\\hline
	\end{tabular}
	\caption{Optimal Euclidean covering constant $\alpha$ for the refinement schemes considered in \Cref{thm:steiner_covering}.}
	\label{table:constant_alpha}
\end{table}

We then apply this estimate facet by facet to $\Del{X}$, while controlling the distortion caused by radial projection. 
This gives the following global spherical estimate:

\begin{theorem}[Contraction theorem, proof p.~\pageref{proof:shrinkingrefinements}]
	\label{thm:shrinkingrefinements}
	Let $X\subset S^d$ be finite, and let $Y$ denote the (spherical) Steiner points associated with $\Del{X}$.
	Assume that $\covrad(X)\leq \pi/3$. 
	Then
	\[
	\covrad(X\cup Y)\leq \tilde\alpha\, \covrad(X),
	\]
	where $\tilde\alpha = \alpha/\cos(\covrad(X))$ and $\alpha$ depends on the chosen refinement, as given in \Cref{table:constant_alpha}.
\end{theorem}

\begin{remark}
	The factor $1/\cos(\covrad(X))$ accounts for the distortion of lengths caused by radial projection onto the sphere.
	It tends to $1$ as the covering radius tends to zero.
	In the Euclidean setting, no such distortion occurs, and the same estimate holds with $\tilde\alpha=\alpha$.
\end{remark}

Iterating this estimate directly yields our main result for Delaunay refinements.

\begin{corollary}[Iterated Delaunay refinement]
	\label{cor:iterated_refinement}
	Let the initial sample $X_0\subset S^d$ be sufficiently dense so that $\covrad(X_0)\leq\pi/3$.
	Then the $n^\mathrm{th}$ iteration of Delaunay refinement satisfies
	\[
		\covrad(X_n) \leq \tilde\alpha^n \covrad(X_0),
	\]
	where $\tilde\alpha = \alpha/\cos(\covrad(X_0))$ and $\alpha$ depends on the chosen refinement, as given in \Cref{table:constant_alpha}.
	In particular, if $\tilde\alpha < 1$, then the maximal diameter of simplices of $\Del{X_n}$ tends to zero.
\end{corollary}

\begin{remark}
	This result highlights a notable distinction between Delaunay refinement and standard subdivision.
	Under Delaunay barycentric refinement, the maximal diameter decreases asymptotically by a factor $1/2$, whereas the best general bound for barycentric subdivision is $d/(d+1)$; see \Cref{eq:shrinking_classical_barycentric}.
	Centroid refinement achieves the same factor in dimension $d=2$, while introducing only one new vertex per facet.
	Similarly, the Delaunay analogue of $2$-edgewise subdivision has contraction factor $\sqrt{(d-1)/(2d)}$, improving on the classical edgewise bound in \Cref{eq:shrinking_classical_edgewise} and, in dimension $3$, on the red refinement bound $1/\sqrt{2}$.
\end{remark}

\subsection{Proof of the contraction theorem}

The proof uses two further auxiliary lemmas, stated after the proof.

\begin{proof}[Proof of \Cref{thm:shrinkingrefinements}]
	\label{proof:shrinkingrefinements}
	To prove that $X\cup Y$ is sufficiently dense, we proceed facet by facet in $\Del{X}$.
	We begin by estimating the covering radius in the Euclidean facets (i.e., in the convex hulls of the vertices) using \Cref{thm:steiner_covering}.
	We then transfer this estimate to the sphere by means of \Cref{lem:spherical_distortion_distances} and use \Cref{lem:spherical_convexity} to express the bound in terms of geodesic distance.
	
	Given a facet $\sigma\in\Del{X}$, let $Y_\mathrm{Euc}^\sigma$ denote the Euclidean Steiner points associated with $\sigma$, before projection onto the sphere.
	Also let $\delta$ denote the Euclidean circumradius of $\sigma$, computed in the $d$-dimensional affine subspace it spans.
	By \Cref{thm:steiner_covering}, for every point $y$ in the convex hull of $\sigma$, there exists a point $x\in \sigma\cup Y_\mathrm{Euc}^\sigma$ such that
	\[
	\|x-y\| \leq \alpha\, \delta.
	\]
	The Euclidean circumradius $\delta$ is related to the spherical circumradius $\Delta$ by
	\[
	\delta = \sin(\Delta).
	\]
	Consequently, \Cref{lem:spherical_distortion_distances}, stated below, yields
	\[
	\left\|\frac{x}{\|x\|}-\frac{y}{\|y\|} \right\|
	\leq
	\frac{1}{\cos(\Delta)} \|x-y\|
	\leq
	\frac{1}{\cos(\Delta)} \alpha\sin(\Delta).
	\]
	By \Cref{lem:equality_covering_circumradius}, the spherical circumradius of $\sigma$ is at most the covering radius of $X$:
	\[
	\Delta \leq \covrad(X).
	\]
	Since the function $t\mapsto\sin(t)/\cos(t)$ is increasing on $[0,\pi/3]$, we deduce that
	\[
	\left\|\frac{x}{\|x\|}-\frac{y}{\|y\|} \right\|
	\leq
	\frac{1}{\cos(\covrad(X))} \alpha\sin(\covrad(X)).
	\]
	We turn to the spherical Steiner points $Y$.
	These are the projection of the Euclidean points:
	\[
		Y = \left\{\frac{x}{\|x\|} ~\bigg|~ x\in Y_\mathrm{Euc}^\sigma \text{ for some facet } \sigma\in\Del{X} \right\}.
	\]
	Since every point of $S^d$ lies in the radial projection of some facet of $\Del{X}$, we can apply the inequality above facet by facet.
	We get that for all $\tilde y\in S^d$, there exists $\tilde x\in X\cup Y$ such that
	\[
	\left\| \tilde x-\tilde y\right\|
	\leq
	\frac{\alpha}{\cos(\covrad(X))} \sin(\covrad(X)).
	\]
	To conclude, we apply \Cref{lem:spherical_convexity} with $c=\alpha/\cos(\covrad(X))$ and $t=\sin(\covrad(X))$: the geodesic distance is bounded above by
	\[
	\d(\tilde x,\tilde y)
	\leq
	\frac{\alpha}{\cos(\covrad(X))} \covrad(X).
	\]
	The lemma holds provided that the following two conditions are satisfied: 
	\[
	\frac{\alpha}{\cos(\covrad(X))}\leq2,
	\qquad
	\sin(\covrad(X))\leq1.
	\] 
	The first condition is true since $\alpha\leq 1$ by \Cref{table:constant_alpha} and $\covrad(X)\leq \pi/3$ by assumption.
	The second condition is immediate.
\end{proof}

\begin{lemma}
	\label{lem:spherical_distortion_distances}
	For $\sigma=\{v_0,\dots,v_d\}\subset S^d$, denote by $\Delta$ its spherical circumradius and by $\conv(\sigma)\subset\R^{d+1}$ the Euclidean simplex it spans.
	Assume $\Delta<\pi/2$.
	For every $x,y\in\conv(\sigma)$,
	\[
	\left\|\frac{x}{\|x\|}-\frac{y}{\|y\|} \right\|
	\leq
	\frac{1}{\cos(\Delta)} \|x-y\|.
	\]
\end{lemma}

\begin{proof}
	Let $u\in S^d$ be the spherical circumcenter of $\sigma$.
	The great-circle distance from $u$ to every vertex $v_i$ is $\Delta$.
	In terms of angles, it means that 
	\[
	\langle u, v_i \rangle = \cos(\Delta).
	\]
	In particular, the vertices all lie in the hyperplane
	\[
	H = \{v \in \R^{d+1} \mid \langle u, v \rangle = \cos(\Delta)\}.
	\]
	Consequently, $\conv(\sigma)\subset H$.
	By Cauchy–Schwarz, it follows that every $p\in \conv(\sigma)$ satisfies
	\[
	\|p\|\geq \cos(\Delta).
	\]	
	Next, let $\pi\colon \R^{d+1}\setminus\{0\} \to S^d$ be the projection onto the sphere.
	Its operator norm at $p$ is
	\[
	\|D \pi_p\|_\mathrm{op} = \frac{1}{\|p\|}.
	\]
	We deduce that $\pi$ is $1/\cos(\Delta)$-Lipschitz on $H$, and the result follows.
\end{proof}

\begin{lemma}\label{lem:spherical_convexity}
	Let $h\colon t \mapsto 2\arcsin(t/2)$ be the map that converts Euclidean distance into spherical distance.
	For all $c \in [0,2]$ and $t\in[0,1]$, $h(ct)\leq c\arcsin(t)$.
\end{lemma}

\begin{proof}
	This follows directly from the convexity of $h$ on the interval $[0, 2t]$:
	\begin{align*}
		h(ct) &= h\left(\left(1-\frac{c}{2}\right)\cdot0 + \left(\frac{c}{2}\right)\cdot 2t\right)
		\leq \left(1-\frac{c}{2}\right)\cdot h(0) + \left(\frac{c}{2}\right)\cdot h\left(2t\right),
	\end{align*}
	and the right-hand side equals $c\arcsin(t)$.
\end{proof}

\section{Covering bounds for minicenter, centroid, and barycentric refinements}
\label{sec:sharp_bounds_centroid_minicenter_barycentric}

In this section, we prove \Cref{thm:steiner_covering} for the minicenter, centroid, and barycentric refinements.
We show that, after adding the Steiner points to a simplex $\sigma\subset B^d$, every point of $\conv(\sigma)$ lies at distance at most $\alpha$ from either a vertex of $\sigma$ or a Steiner point.
Equivalently, the covering radius of the old and new points relative to $\conv(\sigma)$ is at most $\alpha$.

The proofs of \Cref{prop:steiner_covering_minicenter,prop:steiner_covering_centroid,prop:steiner_covering_barycentric} follow a similar mechanism: we study different cases, based on the values of the largest barycentric coordinates.
For the barycentric refinement, we also use a probabilistic argument, stated in \Cref{lem:pivotal_sampling}.
The $k$-fold edgewise refinement requires a more involved optimization argument and is treated in \Cref{sec:sharp_bounds_kfold}.
Throughout the proofs, we assume $d\geq2$; the one-dimensional case (simplices on the real line) is immediate.

In each case, the constant $\alpha$ appearing in \Cref{thm:steiner_covering} is sharp.
The extremal configurations for the centroid, minicenter, and barycentric refinements are described in \Cref{rem:extremal_configuration_minicenter,rem:extremal_configuration_centroid,rem:extremal_configuration_barycenter}.
In particular, the regular simplex is not always extremal.

\subsection{Covering bound for minicenter refinement}

\begin{proposition}
	\label{prop:steiner_covering_minicenter}
	Consider a simplex in the unit ball $\sigma=\{v_0,\dots,v_{d}\}\subset B^d\subset\R^{d}$ and let $w$ be its minicenter. 
	For every $y\in\conv(\sigma)$, there exists $x\in \sigma\cup \{w\}$ such that $\|x-y\| \leq 1/\sqrt{2}$.
\end{proposition}

\begin{proof}
	Denote $r_*=\|y-w\|$ and $r_i = \|y-v_i\|$ for all $i\in\{0,\dots,d\}$.
	We show that at least one value among $r_*,r_0,\dots,r_d$ is less than or equal to $1/\sqrt{2}$.
	
	By assumption, the ball of radius $1$ centered at the origin contains $\sigma$. Hence, by minimality of the minicenter, every vertex satisfies $\|v_i-w\|\leq 1$.
	
	Write $y=\sum_{i=0}^d \lambda_i v_i$ in barycentric coordinates.
	Using the decomposition $y-v_i=(y-w)-(v_i-w)$, we get
	\begin{align*}
		\sum_{i=0}^d \lambda_i \|y-v_i\|^2
		&=
		\sum_{i=0}^d \lambda_i
		\bigg(
		\|y-w\|^2+\|v_i-w\|^2
		-2\langle y-w,v_i-w\rangle\bigg) \\
		&=
		r_*^2
		+\sum_{i=0}^d \lambda_i \|v_i-w\|^2
		-2\bigg\langle y-w,\sum_{i=0}^d \lambda_i(v_i-w)\bigg\rangle .
	\end{align*}
	Since
	\[
	\sum_{i=0}^d \lambda_i(v_i-w)=y-w,
	\]
	it follows that
	\[
	\sum_{i=0}^d \lambda_i \|y-v_i\|^2
	=
	\sum_{i=0}^d \lambda_i \|v_i-w\|^2
	-r_*^2
	\leq
	1-r_*^2.
	\]
	Equivalently,
	\[
	r_*^2+\sum_{i=0}^{d}\lambda_ir_i^2 \leq 1.
	\]
	Let $t=\min\{r_*,r_0,\dots,r_d\}$.
	From the above inequality we deduce that $2 t^2 \leq 1$, as claimed.
\end{proof}

\begin{remark}
	\label{rem:extremal_configuration_minicenter}
	In dimension $d\geq2$, the constant $1/\sqrt{2}$ of \Cref{prop:steiner_covering_minicenter} is sharp, and is not attained by the regular simplex.
	Indeed, let $(e_1,\dots,e_d)$ be the standard basis of $\R^d$, and let 
	\[
	\sigma=\{-e_1,e_1,e_2,\dots,e_d\}\subset \R^d.
	\] 
	The pair $\{-e_1,e_1\}$ forces the minicenter to be the origin.
	Now, the point $y=(e_1+e_2)/2$ is at distance $1/\sqrt{2}$ from $e_1$, $e_2$, and the origin, and is farther away from the other vertices.
	
	For comparison, the regular simplex on the unit sphere is not extremal. 
	Its minicenter is $0$, and the maximum of the distance function to $\sigma\cup\{0\}$ is $\sqrt{d/(2(d+1))}$.
	This value is attained on the edges, for instance at the point $y=((d+2)v_0+dv_1)/(2(d+1))$.
\end{remark}

\subsection{Covering bound for centroid refinement}

\begin{proposition}
	\label{prop:steiner_covering_centroid}
	Consider a simplex in the unit ball $\sigma=\{v_0,\dots,v_{d}\}\subset B^d\subset\R^{d}$ and let $c$ be its centroid. 
	For every $y\in\conv(\sigma)$, there exists $x\in \sigma\cup \{c\}$ such that $\|x-y\| \leq \alpha_d$, where
	\[
	\alpha_d=
	\left\{
	\begin{array}{ll}
		\frac23 &\text{ if }  d=2,\\[2mm]
		\frac{d^2-1}{d^2+1} &\text{ if } d\geq3.
	\end{array}
	\right.
	\]
\end{proposition}

\begin{proof}
	We distinguish the cases $d=2$ and $d\geq3$.
	\paragraph{Case 1: $d=2$.}
	In this case, $\sigma = \{v_0,v_1,v_2\}$ is a triangle.
	Denote its centroid by~$c$.
	We show that the distance from any point $y\in\conv(\sigma)$ to $\{c,v_0,v_1,v_2\}$ is at most $2/3$.
	
	Let $(\lambda_i)_{i=0}^2$ be the barycentric coordinates of $y$.
	Without loss of generality, assume that the coordinates are ordered: $\lambda_0\geq\lambda_1\geq\lambda_2$. 
	We consider whether $\lambda_0\geq2/3$ or $\lambda_0<2/3$.
	
	If $\lambda_0\geq2/3$, then $\lambda_1+\lambda_2\leq1/3$, and $y$ is close to a vertex:
	\begin{align*}		
		\|y-v_0\| = \left\|\sum_{i=0}^2 \lambda_iv_i -v_0  \right\|
		&\leq \left\|\left(\lambda_0-1\right)v_0\right\| + \|\lambda_1v_1\|+\|\lambda_2v_2\|\\
		&\leq \frac13 + \frac13 = \frac23.
	\end{align*}
	On the other hand, if $\lambda_0<2/3$, then $y$ must be close to $c$.
	Indeed, 
	\begin{equation*}		
		\|y-c\| = \left\|\sum_{i=0}^2 \lambda_iv_i -\sum_{i=0}^2 \frac13 v_i\right\|
		\leq \sum_{i=0}^2\left\|\left(\lambda_i-\frac13\right)v_i\right\|
		\leq \sum_{i=0}^2\left|\lambda_i-\frac13\right|.
	\end{equation*}
	Note that $\lambda_0\geq1/3$ and $\lambda_2\leq1/3$ since $\lambda_0\geq\lambda_1\geq\lambda_2$.
	If $\lambda_1\leq 1/3$, then
	\begin{equation*}
		\sum_{i=0}^2
		\left|\lambda_i-\frac13\right|
		=
		\left(\lambda_0-\frac13\right)
		+\left(\frac13-\lambda_1\right)
		+\left(\frac13-\lambda_2\right)
		=
		2\left(\lambda_0-\frac13\right)
		<\frac23.
	\end{equation*}
	If instead $\lambda_1>1/3$, then
	\begin{equation*}
		\sum_{i=0}^2
		\left|\lambda_i-\frac13\right|
		=
		\left(\lambda_0-\frac13\right)
		+\left(\lambda_1-\frac13\right)
		+\left(\frac13-\lambda_2\right)
		=
		2\left(\frac13-\lambda_2\right)
		\leq\frac23.
	\end{equation*}
	In either case, $\|y-c\|\leq2/3$.
	Thus, $y$ is always within distance $2/3$ of a vertex or $c$.

	\paragraph{Case 2: $d\geq3$.}
	We now prove the result in higher dimensions.
	Define the quantities
	\[
	\alpha_d=\frac{d^2-1}{d^2+1},\qquad	
	\beta=\frac{d-1}{d+1},
	\qquad
	t=\frac{2d}{d^2+1}.
	\]
	Note that $\alpha_d^2+t^2=1$ and $\alpha_d=\beta(1+t)$.
	We still denote by $(\lambda_i)_{i=0}^d$ the barycentric coordinates of $y$, which we assume are given in non-increasing order.
	We distinguish two cases: $\lambda_0\leq1/2$ (Subcase 2a) and $\lambda_0>1/2$ (Subcase 2b).
	In the first case, $y$ must be close to a vertex $v_i$ or to the centroid $c$; in the second case, it must be close to $v_0$ or to $c$.

	\subparagraph{Subcase 2a: $\lambda_0\leq1/2$.}
	
	First, we show that
	\begin{equation}
		\label{eq:proof_centroid_case2a_1}
		\|y-c\|
		\leq
		\beta(1+\|y\|).
	\end{equation}
	To prove this inequality, define the new barycentric coordinates
	\[
		\mu_i = \frac{1-2\lambda_i}{d-1},\quad i\in\{0,\dots,d\}.
	\]
	They are nonnegative and sum to 1.
	Consider the point $z=\sum_{i=0}^{d}\mu_i v_i$. It satisfies
	\[
		z=\frac{d+1}{d-1}c-\frac{2}{d-1}y,\qquad
		y-c = \frac{d-1}{d+1}y - \frac{d-1}{d+1}z.
	\]
	We deduce \Cref{eq:proof_centroid_case2a_1}:
	\begin{equation*}
		\|y-c\|
		=
		\beta\|y-z\|
		\leq
		\beta(1+\|y\|).
	\end{equation*}
	
	On the other hand, the identity
	\[
	\sum_{i=0}^d\lambda_i\|v_i-y\|^2
	=
	\sum_{i=0}^d\lambda_i\|v_i\|^2-\|y\|^2
	\leq
	1-\|y\|^2
	\]
	shows that 
	\begin{equation}
		\label{eq:proof_centroid_case2a_2}
		\min_{0\leq i\leq d}\|v_i-y\|
		\leq
		\sqrt{1-\|y\|^2}.
	\end{equation}
	
	To conclude, observe that if $\|y\|\leq t$, then \Cref{eq:proof_centroid_case2a_1} yields
	\[
		\|y-c\|
		\leq
		\beta(1+t)
		=
		\alpha_d,
	\]
	and if $\|y\|\geq t$, then \Cref{eq:proof_centroid_case2a_2} yields
		\[
		\min_{0\leq i\leq d}\|v_i-y\|
		\leq
		\sqrt{1-t^2}
		=
		\alpha_d.
	\]
	This shows the result.
	
	\subparagraph{Subcase 2b: $\lambda_0>1/2$.}

	If $\|y-v_0\|\leq\alpha_d$, $y$ is already close to a vertex and there is nothing to prove. 
	Suppose therefore that $\|y-v_0\|>\alpha_d$.
	In this case, $y$ must be close to the centroid $c$.
	We will define two points $w,x\in B^d$ such that
	\begin{equation}
		\label{eq:proof_centroid_case2b_1}
		\|y-c\|
		=
		\beta\|w-x\|,\qquad
		\|w\| \leq 1, \qquad	
		\|x\| \leq t.
	\end{equation}
	This will imply that
	\begin{equation}
		\label{eq:proof_centroid_case2b_2}
		\|y-c\|
		<
		\beta(1+t)
		=
		\alpha_d,
	\end{equation}
	which yields the result.
	
	First, define the new barycentric coordinates
	\[
		\theta_i=\frac{\lambda_i}{1-\lambda_0},
		\quad i\in\{1,\dots,d\}.
	\]
	They are nonnegative and sum to 1.
	Consider the point $z=\sum_{i=1}^{d}\theta_i v_i$.
	It satisfies
	\begin{equation}
		\label{eq:proof_centroid_case2b_4}
		y=\lambda_0 v_0+(1-\lambda_0)z.
	\end{equation}	
	Since $z\in B^d$, $\|v_0-z\|\leq2$.
	Combined with the assumption $\|y-v_0\|>\alpha_d$, this implies
	\[
		\alpha_d
		<
		\|y-v_0\|
		=
		(1-\lambda_0)\|z-v_0\|
		\leq
		2(1-\lambda_0).
	\]
	Therefore,
	\[
	\lambda_0<1-\frac{\alpha_d}{2}
	=
	\frac{d^2+3}{2(d^2+1)}.
	\]
	Since $d\geq3$, a computation shows that
	\begin{equation}
		\label{eq:proof_centroid_case2b_3}
		\lambda_0
		<
		\frac{d^2+3}{2(d^2+1)}
		\leq
		\frac{d^2-d+1}{d^2+1}.
	\end{equation}	
	
	Next, we define the point
	\[
	w=
	\sum_{i=1}^d\frac{1-\theta_i}{d-1}v_i.
	\]
	Again, the coefficients are nonnegative and sum to $1$, so $w\in B^d$.
	Moreover, it satisfies
	\[
	z+(d-1)w=\sum_{i=1}^d v_i.
	\]
	
	Last, define the point
	\[
		x=Av_0+(1-A)z,\qquad
		A=\frac{(d+1)\lambda_0-1}{d-1}.
	\]
	It satisfies the first equality of \Cref{eq:proof_centroid_case2b_1}:
	\[
	y-c=\beta(x-w).
	\]
	Moreover, a direct computation shows that
	\begin{equation}
		\label{eq:proof_centroid_case2b_5}
		A(1-A)-(1-\lambda_0)^2
		=
		\frac{(2\lambda_0-1)
			\bigl(d^2-d+1-(d^2+1)\lambda_0\bigr)}
		{(d-1)^2}>0,
	\end{equation}	
	where the positivity is ensured by \Cref{eq:proof_centroid_case2b_3}.
	In particular, $A(1-A)>(1-\lambda_0)^2$, $A\in[0,1]$, and hence $x\in B^d$.
	
	To conclude, we use the identity
	\[
		\|Au+(1-A)v\|^2
		=
		A\|u\|^2+(1-A)\|v\|^2 - A(1-A)\|u-v\|^2,
	\]
	valid for all vectors $u,v\in\R^d$ and scalars $A\in[0,1]$.
	Applied to $x=Av_0+(1-A)z$, it yields
	\begin{align*}
		\|x\|^2
		&=
		A\|v_0\|^2+(1-A)\|z\|^2
		-A(1-A)\|v_0-z\|^2
		\tag*{(by the identity)}
		\\
		&\leq
		1-A(1-A)\|v_0-z\|^2
		\tag*{(since $v_0,z\in B^d$)}
		\\
		&<
		1-(1-\lambda_0)^2\|v_0-z\|^2
		\tag*{(by \Cref{eq:proof_centroid_case2b_5})}
		\\
		&=
		1-\|y-v_0\|^2
		\tag*{(by \Cref{eq:proof_centroid_case2b_4})}
		\\
		&<
		1-\alpha_d^2
		\tag*{(since $\|y-v_0\|>\alpha_d$ by assumption)}
		\\
		&=
		t^2.
		\tag*{(by definition of $t$)}
	\end{align*}
	We have obtained the inequality $\|x\|<t$ of \Cref{eq:proof_centroid_case2b_1}.
	We deduce \Cref{eq:proof_centroid_case2b_2}: $y$ is within distance $\alpha_d$ of the centroid.
	This last step completes the proof.
\end{proof}

\begin{remark}
	\label{rem:extremal_configuration_centroid}
	The constant $\alpha_d$ in \Cref{prop:steiner_covering_centroid} is sharp.
	It is attained in degenerate configurations.
	For instance, in dimension $d=2$, one can take the vertices $v_0=e_1$ and $v_1=v_2=-e_1$. Their centroid is $-e_1/3$.
	The point $y=e_1/3$ is at distance $2/3$ from $v_0$ and the centroid:
	\[
		\|y-e_1\| = \|y-c\| = \frac{2}{3}.
	\]
	In dimension $d\geq3$, we can take $v_0=\alpha_d e_1+te_2$, $v_1=-\alpha_d e_1+te_2$, and $v_2=\dots=v_d=-e_2$ with $t=2d/(d^2+1)$ as in the proof.
	It satisfies $\alpha_d^2+t^2=1$.
	The centroid of $v_0,\dots,v_d$ is 
	\[
		c = \frac{2t-(d-1)}{d+1} e_2
		=\frac{-d^2+2d+1}{d^2+1} e_2.
	\]
	Take the midpoint $y=(v_0+v_1)/2=te_2$. 
	It is at equal distance from $v_0,v_1$ and the centroid:
	\[
		\|y-v_0\| = \|y-v_1\| = \frac{1}{2}\|v_0-v_1\| = \alpha_d,\qquad
		\|y-c\| = \left|\frac{2d}{d^2+1}-\frac{-d^2+2d+1}{d^2+1}\right| = \alpha_d.
	\]

	By contrast, for the regular simplex on the unit sphere, the centroid coincides with its minicenter, hence we can use the computation in \Cref{rem:extremal_configuration_minicenter}: the worst points lie on the edges, and are at distance $\sqrt{d/(2(d+1))}$ from the vertices and the minicenter.
\end{remark}

\subsection{A variance estimate from pivotal sampling}

We aim to prove \Cref{thm:steiner_covering} for barycentric refinement.
The proof requires a technical lemma, which we give in this subsection.
We obtain the main result in the following subsection.

Given a point in $\conv(\sigma)$ whose barycentric coordinates are bounded by $1/r$, where $r\in\N_{>0}$, the lemma selects $r$ \textit{distinct} vertices whose barycenter is close to the point.
It is an application of \emph{pivotal sampling} \cite{deville1998unequal,Chauvet2012} based on a result of Chauvet and Ruiz-Gazen~\cite{ChauvetRuizGazen2017}.

\begin{lemma}
	\label{lem:pivotal_sampling}
	Consider a simplex in the unit ball $\sigma=\{v_0,\dots,v_{d}\}\subset B^d\subset\R^{d}$.
	Let $y\in\conv(\sigma)$ be a point with barycentric coordinates at most $1/r$, where $r$ is an integer in $\{1,\dots,d+1\}$.
	Then there exists a random $r$-element subset $R\subset\{0,\dots,d\}$ such that
	\[
	\E\left[
	\frac1r\sum_{i\in R}v_i
	\right]
	= y,
	\qquad
	\E\left\|
	\frac1r\sum_{i\in R}v_i-y
	\right\|^2
	\leq
	\frac1r\big(1-\|y\|^2\big).
	\]
	In particular, there exists at least one realization of $R$ for which this inequality holds.
\end{lemma}

\begin{proof}
	Denote by $\lambda_0,\dots,\lambda_d$ the barycentric coordinates of $y$.
	By assumption, they all satisfy $\lambda_i\leq1/r$.
	Define $\pi_i=r\lambda_i$ for all $i\in \{0,\dots,d\}$.
	They satisfy $0\leq \pi_i\leq 1$ and sum to $r$. 
	We consider ordered pivotal sampling with inclusion probabilities $\pi=(\pi_0,\dots,\pi_d)$.
	It produces a \textit{random} $r$-element subset $R\subset\{0,\dots,d\}$ such that
	\[
	\PP(i\in R)=\pi_i, \quad i \in\{0,\dots,d\}.
	\]
	We use the comparison of ordered pivotal sampling with multinomial sampling in \cite[Theorem~1]{ChauvetRuizGazen2017}.
	Let $A_1,\dots,A_r$ be independent random variables with $\PP(A_j=a_i)=\lambda_i$ for all $j\in\{1,\dots,r\}$ and $i\in\{0,\dots,d\}$.
	The theorem says that, for all real numbers $a_0,\dots,a_d$,
	\begin{equation}
		\label{eq:variance_comparison_pivotal}
		\operatorname{Var}\left(\sum_{i\in R}a_i\right)
		\leq
		\operatorname{Var}\left(\sum_{j=1}^r A_j\right).
	\end{equation}
	
	We now apply \Cref{eq:variance_comparison_pivotal} coordinate by coordinate in $\R^d$.
	Let $W_1,\dots,W_r\in\R^d$ be independent random vectors satisfying
	\[
	\PP(W_j=v_i)=\lambda_i.
	\]
	Let $e_1,\dots,e_m$ be an orthonormal basis of the span of the vectors
	$v_0,\dots,v_d$, and set
	\[
	a_i^{(\ell)}=\langle v_i,e_\ell\rangle,
	\quad i\in\{0,\dots,d\},\quad \ell\in\{1,\dots,m\}.
	\]
	For every $\ell\in\{1,\dots,m\}$, the scalar random variables $\langle W_j,e_\ell\rangle$, $1\leq j\leq r$, are independent and take the value $a_i^{(\ell)}$ with probability $\lambda_i$.
	Consequently, applying \Cref{eq:variance_comparison_pivotal} to $a_0^{(\ell)},\dots,a_d^{(\ell)}$ yields
	\[
	\operatorname{Var}\left(\sum_{i\in R}\langle v_i,e_\ell\rangle\right)
	\leq
	\operatorname{Var}\left(\sum_{j=1}^r \big\langle W_j,e_\ell\big\rangle\right).
	\]	
	Summing over $\ell=1,\dots,m$ gives
	\begin{equation*}
		\sum_{\ell=1}^m
		\operatorname{Var}\left(\sum_{i\in R}\langle v_i,e_\ell\rangle\right)
		\leq
		\sum_{\ell=1}^m
		\operatorname{Var}\left(\sum_{j=1}^r\langle W_j,e_\ell\rangle\right).
	\end{equation*}
	This can be rephrased as
	\begin{equation}
		\label{eq:mean_comparison_vector}
		\E\left\|
		\sum_{i\in R}v_i-r y
		\right\|^2
		\leq
		\E\left\|
		\sum_{j=1}^r W_j-r y
		\right\|^2.
	\end{equation}	
	Indeed, the left-hand side satisfies 
	\[
		\E\left\|
		\sum_{i\in R}v_i-r y
		\right\|^2
		=
		\sum_{\ell=1}^m
		\operatorname{Var}\left(\sum_{i\in R}\langle v_i,e_\ell\rangle\right)
		\qquad
		\text{since}
		\qquad
		\E\left[\sum_{i\in R}v_i\right]
		=
		\sum_{i=0}^d \pi_i v_i
		=
		r y,
	\]
	and the right-hand side satisfies 
	\[
		\E\left\|
		\sum_{j=1}^r W_j-r y
		\right\|^2
		=
		\sum_{\ell=1}^m
		\operatorname{Var}\left(\sum_{j=1}^r\langle W_j,e_\ell\rangle\right)
		\qquad
		\text{since}
		\qquad
		\E\left[\sum_{j=1}^r W_j\right]
		=
		r\sum_{i=0}^d \lambda_i v_i
		=
		r y.
	\]
	
	Since the random vectors $W_1,\dots,W_r$ are independent and identically distributed,
	\[
	\E\left\|
	\sum_{j=1}^r W_j-r y
	\right\|^2
	=
	\sum_{j=1}^r
	\E\left\|W_j-y\right\|^2 
	=
	r\left(\sum_{i=0}^d\lambda_i\|v_i\|^2-\|y\|^2\right).
	\]
	As the $v_i$ have norm at most 1, the right-hand side is bounded above by $r(1-\|y\|^2)$.
	Therefore, \Cref{eq:mean_comparison_vector} yields
	\[
	\E\left\|
	\frac1r\sum_{i\in R}v_i-y
	\right\|^2
	\leq
	\frac{1}{r}(1-\|y\|^2).\qedhere
	\]
\end{proof}

\subsection{Covering bound for barycentric refinement}

\begin{proposition}
	\label{prop:steiner_covering_barycentric}
	Consider a simplex in the unit ball $\sigma=\{v_0,\dots,v_{d}\}\subset B^d\subset\R^{d}$.
	For every $y\in\conv(\sigma)$, there exists an integer $k\in\{1,\dots,d+1\}$ and a $k$-element subset $\{v_{i_1},\dots,v_{i_k}\}\subset\sigma$ (without repetitions) whose barycenter $x=\frac{1}{k}\sum_{\ell=1}^k v_{i_\ell}$ satisfies $\|x-y\| \leq 1/2$.
\end{proposition}

The proof is organized according to the size of the largest barycentric coordinate of the point $y=\sum_{i=0}^d \lambda_i v_i$.
When one coordinate is large, $y$ is close to a vertex or an edge midpoint. 
When all coordinates are small, \Cref{lem:pivotal_sampling}, presented in the previous subsection, gives a barycenter of four distinct vertices within the required distance. 
The intermediate regime includes double and triple barycenters, which we analyze through probabilistic arguments.

Thus, although the barycenters of all nonempty subsets of $\sigma$ are allowed, the proof only uses barycenters of at most four vertices. 
This is sharp in a sense explained in \Cref{rem:extremal_configuration_barycenter}.

\begin{proof}
	Write $y=\sum_{i=0}^d \lambda_i v_i$ in barycentric coordinates, and reorder the vertices so that $\lambda_0\geq \lambda_1\geq \cdots\geq \lambda_d$.
	We distinguish three regimes depending on $\lambda_0$.
	
	\paragraph{Case 1: $\lambda_0\geq 1/2$.}
	We show that $y$ is close to $v_0$ or to the midpoint of an edge $\{v_0,v_i\}$.
	Let $\mu_0=2\lambda_0-1$ and $\mu_i=2\lambda_i$ for $i\in\{1,\dots,d\}$. 
	The $(\mu_i)_{i=0}^d$ are nonnegative, sum to $1$, and
	\[
	y=\mu_0 v_0+\sum_{i=1}^d \mu_i\frac{v_0+v_i}{2}.
	\]
	The points $v_0$ and $(v_0+v_i)/2$, $i\in\{1,\dots,d\}$, all lie in the ball centered at $v_0/2$ and of radius $1/2$.
	As in the proof of \Cref{prop:steiner_covering_minicenter}, expanding the norms yields
	\[
	\mu_0\|v_0-y\|^2
	+
	\sum_{i=1}^d \mu_i
	\left\|\frac{v_0+v_i}{2}-y\right\|^2
	\leq
	\frac14-\left\|y-\frac{v_0}{2}\right\|^2
	\leq \frac14.
	\]
	Therefore at least one of these points is at distance at most $1/2$ from $y$.
	
	\paragraph{Case 2: $\lambda_0\leq 1/4$.}
	By \Cref{lem:pivotal_sampling} with $r=4$, there is a subset $T$ of four elements such that
	\[
	\left\|\frac14\sum_{i\in T}v_i-y\right\|^2
	\leq
	\frac14(1-\|y\|^2)
	\leq
	\frac14.
	\]
	Thus the barycenter of $T$ is at distance at most $1/2$ from $y$.
	
	\paragraph{Case 3: $1/4<\lambda_0<1/2$.}
	We define the point
	\[
	w=\sum_{i=1}^d \frac{\lambda_i}{1-\lambda_0} v_i.
	\]
	It satisfies the relation 
	\[
	y=\lambda_0 v_0+(1-\lambda_0)w.
	\]
	We split this case according to whether the new weights $\mu_i=\lambda_i/(1-\lambda_0)$ are all at most $1/2$.
	
	\subparagraph{Subcase 3a: $\forall i\in\{1,\dots,d\},\, \mu_i\leq 1/2$.} 
	We show that $y$ must be close to a barycenter of two or three vertices, depending on the values of $\lambda_0$ and $\|w\|$ (see \Cref{eq:proof_barycenter_gamma_1} below).
	First, we sample a random index $I$ with distribution $\PP(I=i)=\mu_i$ and define
	\[
	E=\frac{v_0+v_I}{2}.
	\]
	Second, we apply \Cref{lem:pivotal_sampling} to $w$ with $r=2$ and coordinates $\mu_1,\dots,\mu_d\leq1/2$.
	We obtain a random pair $R=\{J,K\}\subset\{1,\dots,d\}$ with midpoint 
	\[
	p=\frac{v_J+v_K}{2}
	\]
	such that
	\[
	\E [p]=w,
	\qquad
	\E\|p-w\|^2\leq \frac12(1-\|w\|^2).
	\]
	We build the (random) triple barycenter
	\[
	B=\frac{v_0+v_J+v_K}{3}=\frac13 v_0+\frac23 p.
	\]

	For the edge midpoint $E$, observe that
	\[
	E-y = \frac12(v_I-w)+\left(\frac12-\lambda_0\right)(v_0-w).
	\]
	Since $\E [v_I]=w$, the expectation $\E \langle v_I-w, v_0-w \rangle$ vanishes, hence 
	\begin{align}
		\E\|E-y\|^2
		&=
		\frac14\,\E\|v_I-w\|^2
		+
		\left(\frac12-\lambda_0\right)^2\|v_0-w\|^2 \nonumber \\
		&\leq
		\frac14(1-\|w\|^2)
		+
		\left(\frac12-\lambda_0\right)^2\|v_0-w\|^2 \nonumber 
		\tag*{(since $\E [v_I]=w$)}
		\\
		&\leq
		\frac14(1-\|w\|^2)
		+
		\left(\frac12-\lambda_0\right)^2(1+\|w\|)^2, \label{eq:barycenter_estimate_edge}
	\end{align}
	where we used $\|v_0\|\leq1$ on the last line.
	Similarly, for the triple barycenter, we use
	\[
	B-y = \frac23(p-w)+\left(\frac13-\lambda_0\right)(v_0-w).
	\]
	The relation $\E [p]=w$ implies that $\E \langle p-w, v_0-w \rangle$ is zero, hence
	\begin{align}
		\E\|B-y\|^2
		&=
		\frac49\,\E\|p-w\|^2
		+
		\left(\frac13-\lambda_0\right)^2\|v_0-w\|^2 \nonumber \\
		&\leq
		\frac29(1-\|w\|^2)
		+
		\left(\frac13-\lambda_0\right)^2(1+\|w\|)^2.
		\label{eq:barycenter_estimate_triple}
	\end{align}
	
	We now show that at least one of these two expectations is at most $1/4$.
	If
	\begin{equation}
		\label{eq:proof_barycenter_gamma_1}
		\frac12-\lambda_0\leq \frac{\|w\|}{2(1+\|w\|)},
	\end{equation}
	then \Cref{eq:barycenter_estimate_edge} gives
	\[
	\E\|E-y\|^2
	\leq
	\frac14(1-\|w\|^2)
	+
	\frac{\|w\|^2}{4}
	=
	\frac14.
	\]
	Otherwise,
	\begin{equation}
		\label{eq:proof_barycenter_gamma_2}
		\frac12-\lambda_0> \frac{\|w\|}{2(1+\|w\|)}.	
	\end{equation}
	If $1/4\leq\lambda_0\leq 1/3$, then $|1/3-\lambda_0|\leq1/12$, hence \Cref{eq:barycenter_estimate_triple} gives
	\[
	\E\|B-y\|^2
	\leq
	\frac29(1-\|w\|^2)+\frac1{144}(1+\|w\|)^2
	=
	\frac{33+2\|w\|-31\|w\|^2}{144}
	\leq
	\frac14.
	\]
	Thus we may assume $\lambda_0>1/3$.
	Define
	\[
	\gamma=\lambda_0-\frac13>0.
	\]
	From \Cref{eq:proof_barycenter_gamma_2}, we deduce
	\[
	\gamma
	=
	\frac16-\left(\frac12-\lambda_0\right)
	<
	\frac16-\frac{\|w\|}{2(1+\|w\|)}
	=
	\frac{1-2\|w\|}{6(1+\|w\|)}.
	\]
	Equivalently,
	\[
	\gamma^2(1+\|w\|)^2
	<
	\frac{(1-2\|w\|)^2}{36}.
	\]
	Substituting this inequality in \Cref{eq:barycenter_estimate_triple} yields
	\begin{align*}
		\E\|B-y\|^2
		&<
		\frac29(1-\|w\|^2)
		+
		\frac{(1-2\|w\|)^2}{36}
		=
		\frac{9-4\|w\|-4\|w\|^2}{36}.
	\end{align*}
	This last expression is at most $1/4$.
	Thus either $\E\|E-y\|^2\leq 1/4$ or $\E\|B-y\|^2\leq 1/4$.
	Consequently, at least one realization of $E$ or $B$ is at distance at most $1/2$ from $y$.
	
	\subparagraph{Subcase 3b: $\exists i\in\{1,\dots,d\},\, \mu_i> 1/2$.}
	As in the previous subcase, we show that $y$ must be close to a barycenter of two or three vertices, depending on the value of $1-\lambda_0-\lambda_1$.	
	Since the weights $\mu_i=\lambda_i/(1-\lambda_0)$ are ordered, the assumption becomes $\mu_1>1/2$.
	Thus
	\[
	\lambda_1>\frac{1-\lambda_0}{2}.
	\]
	Define the quantity 
	\[
	\nu=1-\lambda_0-\lambda_1 \geq 0.
	\] 
	Note that $\nu>0$, since $1/2>\lambda_0\geq\lambda_1$ here.
	We distinguish the cases $\nu\leq1/4$ and $\nu>1/4$.

	If $0<\nu\leq1/4$, then $y$ is close to the edge midpoint 
	\[
	E=(v_0+v_1)/2.
	\]
	Indeed, define
	\begin{equation}
		\label{eq:proof_barycentric_z}
		z=\frac1\nu\sum_{i=2}^d \lambda_i v_i.
	\end{equation}
	It satisfies $y=\lambda_0 v_0+\lambda_1 v_1+\nu z$.
	Consequently, one has
	\begin{align*}
		y-E
		&=
		\left(\lambda_0-\frac12\right)v_0
		+
		\left(\lambda_1-\frac12\right)v_1
		+
		\nu z .
	\end{align*}
	Since $\lambda_0,\lambda_1\leq 1/2$, and since
	\[
	\left(\frac12-\lambda_0\right)+\left(\frac12-\lambda_1\right)
	=
	1-\lambda_0-\lambda_1
	=
	\nu,
	\]
	we obtain
	\begin{align*}
		\|y-E\|
		&\leq
		\left(\frac12-\lambda_0\right)
		+
		\left(\frac12-\lambda_1\right)
		+
		\nu
		=
		2\nu
		\leq
		\frac12.
	\end{align*}
	
	It remains to consider $\nu>1/4$. 
	In this case, $\lambda_0$, $\lambda_1$, and $\nu$ lie in $(1/4,1/2)$, since $\lambda_0+\lambda_1+\nu=1$.
	Let $I\in\{2,\dots,d\}$ be a random index with probability
	\[
	\PP(I=i)=\frac{\lambda_i}{\nu},
	\]
	and define the triple barycenter and its mean
	\[
		B=\frac{v_0+v_1+v_I}{3},
		\qquad
		b= \E[B].
	\]
	We show that
	\[
		\E\|B-y\|^2
		<
		\frac14,
	\]
	which will imply the result.
	The expected value of $v_I$ is $z$, defined in \Cref{eq:proof_barycentric_z}.
	We have
	\begin{align}
		\E\|B-b\|^2
		&=
		\frac19\,\E\|v_I-z\|^2
		\tag*{(since $B-b=(v_I-z)/3$)}\nolinenumbers\\
		&\leq
		\frac19(1-\|z\|^2)
		\tag*{(since $\E [v_I]=z$)}\nolinenumbers\\
		&\leq
		\frac19. \label{eq:proof_barycenter_nu_1}
	\end{align} 
	We now estimate $\|y-b\|$.
	Since $y=\lambda_0v_0+\lambda_1v_1+\nu z$, we have
	\[
	y-b
	=
	\left(\lambda_0-\frac13\right)v_0
	+
	\left(\lambda_1-\frac13\right)v_1
	+
	\left(\nu-\frac13\right)z.
	\]
	Moreover, from $\lambda_0,\lambda_1,\nu\in(1/4,1/2)$ and
	$\lambda_0+\lambda_1+\nu=1$ we deduce that 
	\[
	\left|\lambda_0-\frac13\right|
	+
	\left|\lambda_1-\frac13\right|
	+
	\left|\nu-\frac13\right|
	\leq
	\frac13.
	\]
	As $v_0$, $v_1$, and $z$ all have norm at most $1$, it follows that
	\begin{equation}
		\label{eq:proof_barycenter_nu_2}
		\|y-b\|\leq \frac13.
	\end{equation}
	Finally, we combine $\E[B-b]=0$ with \Cref{eq:proof_barycenter_nu_1,eq:proof_barycenter_nu_2} to get
	\begin{align*}
		\E\|B-y\|^2
		&=
		\|b-y\|^2+\E\|B-b\|^2 \\
		&\leq
		\frac19+\frac19
		<
		\frac14.
	\end{align*}
	Consequently, $B$ has a realization within distance $1/2$ of $y$.
	The proof is complete.
\end{proof}

\begin{remark}
	\label{rem:extremal_configuration_barycenter}
	The proof of \Cref{prop:steiner_covering_barycentric} only uses barycenters of at most four vertices. 
	In particular, the proposition remains true if $k$ is restricted to take values in $\{1,2,3,4\}$.
	
	This phenomenon is already visible for the regular simplex. 
	Let $\sigma=\{v_0,\dots,v_d\}\subset S^{d-1}$ be regular, let $Y_{\leq k}$ denote the barycenters of faces of cardinality at most $k$, and put $K=\min\{k,d+1\}$. 
	Then the covering radius of $Y_{\leq k}$ inside $\conv(\sigma)$ is
	\[
	\sqrt{
		\max\left\{
		\frac14,\,
		\frac{d+1-K}{Kd}
		\right\}
	}.
	\]
	The second term is the squared distance from the center $0$ to the closest allowed face barycenters, namely those of cardinality $K$. 
	The term $1/4$ is a persistent obstruction: for every vertex $v_i$, the point $v_i/2$ belongs to $\conv(\sigma)$.
	Moreover, if $b_\tau$ denotes the barycenter of a face $\tau$, then
	\[
	\left\|
	\frac12 v_i-b_\tau
	\right\|^2
	=
	\left\{
	\begin{array}{ll}
		\frac14 &\text{ if }  v_i\in \tau,\\[2mm]
		\frac14+\frac{d+1}{|\tau|d} &\text{ if } v_i\notin \tau.
	\end{array}
	\right.
	\]
	Thus the point $v_i/2$ is always at distance $1/2$ from the allowed barycenters.
\end{remark}

\section{Sharp approximate Carathéodory theorem}
\label{sec:sharp_bounds_kfold}

This section treats the last case of \Cref{thm:steiner_covering}, namely, the $k$-edgewise refinement.
Our main result can be viewed as an enhancement of the classical approximate Carathéodory theorem, recalled in \Cref{subsubsec:approximate_caratheodory}.
It is stated in \Cref{subsec:statement_sharp_bounds_kfold}, and the proof covers \Cref{subsec:optimality_conditions,subsec:estimate_without_repetitions,subsec:proof_covering_kfold}.
\Cref{rem:extremal_configuration_kfold} shows that our result is sharp and is attained by the regular simplex.
As an application, we deduce in \Cref{subsec:link_with_conjecture} some cases of a
conjecture of Bárány and Füredi.

\subsection{Statement and proof overview}
\label{subsec:statement_sharp_bounds_kfold}

\begin{theorem}[proof p.~\pageref{proof:steiner_covering_kfold}]
	\label{thm:steiner_covering_kfold}
	Consider a simplex in the unit ball $\sigma=\{v_0,\dots,v_{d}\}\subset B^d\subset\R^{d}$ and an integer $k\geq2$.
	For every $y\in\conv(\sigma)$, there exists a $k$-tuple $(v_{i_1},\dots,v_{i_k})$ of vertices of $\sigma$ (allowing repetitions) whose barycenter $x=\frac1k\sum_{\ell=1}^k v_{i_\ell}$ satisfies $\|x-y\| \leq R_{d,k}$, where
	\[
	R_{d,k}
	=
	\sqrt{
		\frac{s(d+1-s)}{dk^2}
	},
	\qquad
	s=\min\left\{k,\left\lfloor\frac{d+1}{2}\right\rfloor\right\}.
	\]
\end{theorem}

\begin{remark}
	\label{rem:extremal_configuration_kfold}
	Explicitly, depending on the size of $k$ and the parity of $d$, the bound is
	\[
		R_{d,k}
		=
		\left\{
		\begin{array}{ll}
			\sqrt{\frac{d+1-k}{dk}} &\text{ if } k\leq \left\lfloor\frac{d+1}{2}\right\rfloor,\\[2mm]
			\frac{d+1}{2k\sqrt{d}} &\text{ if } k> \left\lfloor\frac{d+1}{2}\right\rfloor \text{ and } d \text{ odd},\\[2mm]
			\frac{\sqrt{d+2}}{2k} &\text{ if } k> \left\lfloor\frac{d+1}{2}\right\rfloor \text{ and } d \text{ even}.
		\end{array}
		\right.
	\]
	It is strictly smaller than the bound $1/\sqrt{k}$ given by the approximate Carathéodory theorem for $k\geq2$ (see \Cref{eq:classical_approximate_caratheodory} in \Cref{subsubsec:approximate_caratheodory}).
	It matches the value obtained by Bomze, Gollowitzer, and Yıldırım in \cite[Theorem~9]{Bomze_2013} for the regular simplex, after accounting for the factor $\sqrt{(d+1)/d}$ relating the two formulations.
	Thus the constant is the best possible.
\end{remark}

Our proof involves computations in barycentric coordinates.
The admissible points $x$ in the theorem are the $k$-fold barycenters, whose barycentric coordinates belong to
\begin{equation}
	A_{d,k}
	=
	\left\{
	\alpha\in\R^{d+1} \mid
	\alpha_i=\frac{n_i}{k},\ 
	n_i\in\N,\
	\sum_{i=0}^d n_i=k
	\right\}.
	\label{eq:lambda_d_k}
\end{equation}
To find such a barycenter near a given point $y$, we use the following strategy.
Let $(\lambda_0,\dots,\lambda_d)$ be the barycentric coordinates of $y$.
We round them down to the nearest multiple of $1/k$ via
\[
	\lambda^-_i = \frac{\lfloor k\lambda_i\rfloor}{k},
	\quad
	i\in\{0,\dots,d\}.
\]
After normalization, the coordinates $(\lambda^-_0,\dots,\lambda^-_d)$ define an $\ell$-fold barycenter, with $\ell\leq k$.
Let $r=k-\ell$.
To complete it to a $k$-fold barycenter, we consider the remaining coordinates
\[
	\lambda^+_i = \lambda_i - \lambda^-_i,
	\quad
	i\in\{0,\dots,d\}.
\]
After normalization, the tuple $(\lambda^+_0,\dots,\lambda^+_d)$ defines a \emph{central point} of $\conv(\sigma)$, by which we mean a point whose barycentric coordinates are uniformly small.
We show in \Cref{lem:covering_without_repetitions} that such a point is well approximated by an $r$-fold barycenter (more precisely, by a barycenter without repetitions).
Combining it with the rounded-down part yields the desired point.

Proving \Cref{lem:covering_without_repetitions} is the delicate step.
To show that central points are close to $r$-fold barycenters, we argue by optimizing over all spherical configurations $\sigma\in (S^{d-1})^{d+1}$.
Arrange the vertices in a $d\times(d+1)$ matrix $V=[v_0\ \cdots\ v_d]$, and consider the \emph{Gram matrix}
\[
G=V^\top V
= \big(\langle v_i,v_j\rangle\big)_{i,j = 0}^d.
\]
If $\alpha$ denotes the barycentric coordinates of a point $x\in\conv(\sigma)$, the matrix product reads
\[
V\alpha = x.
\]
In particular, if $y$ is another point with barycentric coordinates $\lambda$, then
\[
\|y-x\|^2
= 
(\lambda-\alpha)^\top G(\lambda-\alpha).
\]
To prove the result, we consider a maximizer $(V,\lambda)$ of
\begin{equation}
	\label{eq:minimization_problem}
	\Phi(V,\lambda)
	=
	\min_{\alpha\in A^*_{d,r}}
	(\lambda-\alpha)^\top G(\lambda-\alpha),
	\qquad
	G=V^\top V,
\end{equation}
where $A^*_{d,r}$ denotes the barycentric coordinates of the $r$-fold barycenters without repetitions:
\begin{equation*}
	A^*_{d,r}
	=
	\left\{
	\alpha\in\R^{d+1} \mid
	\alpha_i=\frac{1}{r}\text{ or } 0,\ 
	\sum_{i=0}^d \alpha_i=1
	\right\}.
\end{equation*}

\Cref{subsec:optimality_conditions} derives optimality conditions for this problem: \Cref{lem:gram_variations} characterizes which matrices are obtained as first-order deformations of $G$, and \Cref{lem:optimality_conditions} applies Gordan's lemma to the active constraints.
These conditions are used in \Cref{subsec:estimate_without_repetitions}
to prove \Cref{lem:covering_without_repetitions}, which shows the estimate for central points.
The full theorem is derived in \Cref{subsec:proof_covering_kfold}.

\subsection{Optimality conditions for the Gram matrix}
\label{subsec:optimality_conditions}

We first derive optimality conditions from \Cref{eq:minimization_problem}.
\Cref{lem:gram_variations} characterizes which first-order variations of the Gram matrix are compatible with vertices on the unit sphere.
Then, in \Cref{lem:optimality_conditions}, we apply Gordan's lemma to the nearest barycenters at a maximizer.
It yields weights on the active barycenters, whose first moment is centered at the maximizing point, and whose second moment depends on the linear dependence among the vertices.

Given unit vectors $v_0,\dots,v_d$ in $\R^d$, we recall the notation
\[
	V=[v_0\ \cdots\ v_d],
	\qquad
	G=V^\top V.
\]
If $V$ is affinely full-dimensional, then $\dim\ker G=1$.
Let $\tau_G$ be the unique vector satisfying
\[
	V\tau_G=0,
	\qquad
	\langle e, \tau_G\rangle=1,
\]
where $e = e_1+\dots+e_{d+1}$ is the sum of the standard basis of $\R^{d+1}$.
Such a vector exists since the relations $\sum_{i=0}^d \tau_iv_i=0$ and $\sum_{i=0}^d \tau_i=0$ would imply an affine dependence of the $(v_i)_{i=0}^d$.

\begin{lemma}
	\label{lem:gram_variations}
	Let $\sigma=\{v_0,\dots,v_{d}\}\subset\R^{d}$ be a subset of unit vectors such that $V$ is affinely full-dimensional.
	A symmetric matrix $H\in\R^{(d+1)\times(d+1)}$ occurs as the derivative of the Gram matrix along a smooth deformation of the vertices on the unit sphere if and only if
	\begin{equation}
		\label{eq:gram_variations}
		\diag (H)=0,
		\qquad
		\tau_G^\top H\tau_G=0.
	\end{equation}
\end{lemma}

\begin{proof}
	We first show that every infinitesimal deformation of the vertices on the unit sphere satisfies these two conditions, and then prove the converse.
	
	\paragraph{Step 1: Direct implication.}
	
	Consider a smooth curve $t\mapsto \mathcal{V}(t)$ such that $\mathcal{V}(t)\in (S^{d-1})^{d+1}$ and $\mathcal{V}(0)=V$.
	Consider the associated Gram matrices
	\[
		\mathcal{G}(t) = \mathcal{V}(t)^\top \mathcal{V}(t).
	\]
	In particular, $\mathcal{G}(0) = G$.
	We must show that the derivative $H = \mathcal{G}'(0)$ satisfies \Cref{eq:gram_variations}.
	
	We first observe a relation that will be useful below: if we define $W = \mathcal{V}'(0)$, then
	\begin{equation}
		H=\mathcal{G}'(0)=V^\top W+W^\top V.
		\label{eq:proof_first_variation1}
	\end{equation}
	
	Next, the diagonal entries of $\mathcal{G}(t)$ are all 1. Thus, differentiating at $t=0$ gives
	\begin{equation*}
		H_{i,i} = 0
		\quad \forall i\in\{0,\dots,d\}.
	\end{equation*}
	Moreover, for $t$ small, the matrices satisfy $\mathcal{G}(t)\tau_{\mathcal{G}(t)}=0$.
	Differentiating at $t=0$ gives
	\[
		H\tau_G+G\tau'_G=0.
	\]
	Multiplying on the left by $\tau_G^\top$, and using $\tau_G^\top G=0$, we obtain
	\begin{equation*}
		\tau_G^\top H \tau_G=0.
	\end{equation*}
	This shows \Cref{eq:gram_variations}, as wanted.
	
	\paragraph{Step 2: Characterization of the infinitesimal deformations.}
	
	Conversely, we must show that every symmetric matrix $H$ satisfying \Cref{eq:gram_variations} is obtained as the derivative of the Gram matrix of a smooth deformation of the vertices $v_0,\dots,v_d$.
	We fix such a matrix $H$.
	We first show that there exists a matrix
	$W=[w_0\ \cdots\ w_d]$ such that
	\[
	\langle v_i,w_i\rangle=0
	\quad \forall i\in\{0,\dots,d\},
	\qquad
	H=V^\top W+W^\top V.
	\]
	Motivated by \Cref{eq:proof_first_variation1}, we consider the linear map
	\begin{align*}
		\Psi\colon \mathcal{W} &\longrightarrow \Sym_{d+1}(\R)\\
		W &\longmapsto V^\top W+W^\top V,
	\end{align*}
	where $\Sym_{d+1}(\R)$ is the space of $(d+1)\times(d+1)$ symmetric matrices, and
	\[
	\mathcal{W} = \big\{W=[w_0\ \cdots\ w_d] \in\R^{d\times(d+1)} \mid \langle v_i,w_i\rangle=0\ \forall i\in\{0,\dots,d\}\big\}.
	\]
	We show that $\Psi(\mathcal{W}) = \mathcal{H}$, where
	\begin{align*}
		\mathcal{H} &= \{H \in \operatorname{Sym}_{d+1}(\R) \mid \diag (H)=0,\ \tau_G^\top H\tau_G=0\}.
	\end{align*}
	
	First, observe that for every $W\in\mathcal W$, we have
	\[
		\Psi(W)_{i,i}
		= \langle v_i,w_i\rangle + \langle w_i,v_i\rangle
		= 0,
	\]
	and, since $V\tau_G=0$,
	\[
	\tau_G^\top\Psi(W)\tau_G
	=
	\langle W\tau_G,V\tau_G\rangle+\langle V\tau_G,W\tau_G\rangle
	=
	0.
	\]
	We deduce that $\Psi(\mathcal W)\subseteq\mathcal H$.
	We prove equality by comparing the dimensions.
			
	First, we show that the kernel of $\Psi$ is
	\begin{equation}
		\ker\Psi = \{ A V \mid A \in \mathfrak{so}(d)\},	
		\label{eq:proof_first_variation4}
	\end{equation}
	where $\mathfrak{so}(d)$ denotes the space of $d\times d$ skew-symmetric matrices.
	On the one hand, the inclusion $\ker\Psi \supset \{ A V \mid A \in \mathfrak{so}(d)\}$ is obvious.
	Conversely, if $W\in\ker\Psi$, then
	\begin{equation}
		V^\top W+W^\top V=0.		
		\label{eq:proof_first_variation5}
	\end{equation}
	Multiplying on the right by $\tau_G$ and using $V\tau_G=0$ gives
	\[
		V^\top W\tau_G=0.
	\]
	Since the columns of $V$ span $\R^d$, the map $V^\top\colon\R^d\to\R^{d+1}$ is injective.
	Hence $W\tau_G=0$.

	Observe that the assignment $v_i\mapsto w_i$ extends to a linear map $A:\R^d\to\R^d$.
	Indeed, the only linear relation satisfied by $V$ is $V\tau_G=0$, which is satisfied by $W$.	
	In other words, 
	\[
		W=AV.
	\]	 
	\Cref{eq:proof_first_variation5} then becomes
	\[
		V^\top(A+A^\top)V=0.
	\]
	Since the columns of $V$ span $\R^d$, we deduce that $A+A^\top=0$, thus $A$ is skew-symmetric.
	Consequently, we have that $\ker\Psi \subset \{ A V \mid A \in \mathfrak{so}(d)\}$.
	We deduce \Cref{eq:proof_first_variation4}. 	

	We deduce that the dimension of the kernel is:
	\[
		\dim\ker\Psi = \dim \mathfrak{so}(d) = \frac{d(d-1)}{2}.
	\]
	Moreover, it is clear that $\dim \mathcal{W} = (d+1)(d-1)$.
	Thus the dimension of $\Psi(\mathcal{W})$ is
	\begin{align*}
		\dim\Psi(\mathcal{W})
		&= \dim\mathcal{W} - \dim\ker\Psi\\
		&=(d+1)(d-1)-\frac{d(d-1)}2
		= \frac{d(d+1)}2-1.
	\end{align*}
	
	On the other hand, the space of symmetric $(d+1)\times(d+1)$ matrices with zero diagonal has dimension $d(d+1)/2$.
	The condition $\tau_G^\top H\tau_G=0$ adds an independent linear relation.
	Indeed, $\tau_G$ has at least two nonzero coordinates.
	We deduce that
	\[
		\dim\mathcal H
		=
		\frac{d(d+1)}2-1.
	\]
	This matches $\dim\Psi(\mathcal{W})$.
	Since $\Psi(\mathcal W)\subseteq\mathcal H$, we conclude that the two spaces are equal.
		
	\paragraph{Step 3: Integration into a smooth deformation.}

	Since the fixed matrix $H$ belongs to $\mathcal H$, we deduce that there exists $W=[w_0\ \cdots\ w_d]\in\mathcal W$ such that
	\[
		H=\Psi(W)=V^\top W+W^\top V.
	\]
	We must integrate this infinitesimal deformation into a smooth deformation on the sphere.
	We simply take 
	\[
		\mathcal V(t)
		=
		[v_0(t)\ \cdots\ v_d(t)],
		\qquad
		v_i(t)
		=
		\frac{v_i+t w_i}{\|v_i+t w_i\|},
		\quad
		i \in \{0,\dots,d\}.
	\]
	For $t$ close to $0$, this is a smooth curve in $(S^{d-1})^{d+1}$.
	It satisfies $\mathcal V(0)=V$. 
	Moreover,
	\[
		v_i'(0)
		=
		w_i-\langle v_i,w_i\rangle v_i
		=
		w_i
	\]
	since $W\in\mathcal W$. 
	Thus, $\mathcal V'(0)=W$.
	Besides, if we denote $\mathcal G(t)=\mathcal V(t)^\top\mathcal V(t)$, then
	\[
		\mathcal G'(0)
		= \mathcal{V}(0)^\top \mathcal{V}'(0)+\mathcal{V}'(0)^\top \mathcal{V}(0)
		= V^\top W+W^\top V
		= H.
	\]
	Thus $H$ is the derivative of the Gram matrix along a smooth deformation of the vertices.
\end{proof}

We now derive optimality conditions for our maximization problem.
We formulate the lemma for any finite $A\subset\Delta_d$, where $\Delta_d$ is the set of barycentric coordinates, defined as
\[
	\Delta_d=\left\{\lambda\in\R_{\geq 0}^{d+1} 
	~\big|~
	\sum_{i=0}^d\lambda_i=1\right\}.
\]
However, in the sequel, we will only use it in the proof of \Cref{lem:covering_without_repetitions} with $A=A^*_{d,r}$.

\begin{lemma}
	\label{lem:optimality_conditions}
	Let $A\subset\Delta_d$ be finite.
	Suppose that $(V,\lambda)$ is a global maximizer of
	\[
		\Phi(V,\lambda)
		=
		\min_{\alpha\in A}
		(\lambda-\alpha)^\top G(\lambda-\alpha),
		\qquad
		G=V^\top V,
	\]
	over all spherical configurations $V$ in $(S^{d-1})^{d+1}$ and over an open neighborhood of $\lambda$ in $\Delta_d$.
	Suppose that $V$ is affinely full-dimensional.	
	Define the maximum value and the active points
	\begin{align*}
		\Phi^*&=\Phi(V,\lambda)>0,\\
		A_0 &=\left\{\alpha\in A\mid (\lambda-\alpha)^\top G(\lambda-\alpha)=\Phi^*\right\}.		
	\end{align*}
	Denote by $\tau_G$ the unique vector satisfying 
	\[
	V\tau_G=0,
	\qquad
	\langle e, \tau_G\rangle=1,
	\]
	where $e = e_1+\dots+e_{d+1}$.
	Then there exist nonnegative weights $(w_\alpha)_{\alpha\in A_0}$ such that
	\begin{align}
		\sum_{\alpha\in A_0}w_\alpha &=1, 
		\label{eq:lem_optimality_1}\\	
		\sum_{\alpha\in A_0}w_\alpha (\lambda-\alpha) &=0,
		\label{eq:lem_optimality_2}\\	
		\sum_{\alpha\in A_0}w_\alpha (\lambda-\alpha) (\lambda-\alpha)^\top
		&= \Phi^*\bigl(\diag(\tau_G)-\tau_G \tau_G^\top\bigr). 
		\label{eq:lem_optimality_3}
	\end{align}
	Moreover, $\tau_G$ has nonnegative coordinates.
\end{lemma}

\begin{proof}
	We split the proof into four steps, following the four statements in the lemma.
	
	\paragraph*{Step 1: First-order variations.}
	
	Define the objective function associated with $\alpha\in A$ as
	\[
		\Phi_\alpha(G,\lambda)
		=
		(\lambda-\alpha)^\top G(\lambda-\alpha).
	\]
	Besides, let $H$ be an admissible first-order variation of $G$, and let $u$ be a first-order variation of $\lambda$. 
	That is, we consider
	\[
		G(t)=G+tH+o(t),
		\qquad
		\lambda(t)=\lambda+tu+o(t).
	\]
	In addition, we suppose that $\langle e,u\rangle=0$, so that $\lambda(t)$ lies in the affine hyperplane $\sum_{i=0}^d\lambda_i=1$, up to an error term $o(t)$.
	We call $(H,u)$ an \textit{admissible pair}.
	For an active point $\alpha\in A_0$, set
	\[
		r_\alpha=\lambda-\alpha.
	\]
	Then
	\[
		\lambda(t)-\alpha=r_\alpha+tu+o(t).
	\]
	Therefore,
	\begin{align*}		
		\Phi_\alpha(G(t),\lambda(t))
		&=
		(r_\alpha+tu)^\top(G+tH)(r_\alpha+tu)+o(t)  \\
		&=
		r_\alpha^\top G r_\alpha
		+
		t\left(
		r_\alpha^\top H r_\alpha+2u^\top G r_\alpha
		\right)
		+o(t).
	\end{align*}
	Thus, the first-order variation of the active constraint is
	\begin{equation}
		\label{eq:proof_lem_optimality_1}
		L_\alpha(H,u)
		=
		r_\alpha^\top H r_\alpha
		+
		2u^\top G r_\alpha.
	\end{equation}
	
	In preparation for applying Gordan's lemma, let us observe that there is no admissible pair $(H,u)$ such that, for all active $\alpha\in A_0$,
	\[
		L_\alpha(H,u)>0.
	\]
	Indeed, if such a pair existed, then the deformation $(H,u)$ would increase all active objective functions $\Phi_\alpha$, $\alpha\in A_0$.
	Since the inactive constraints are separated from the minimum 
	\[
		\Phi(V,\lambda)
		=
		\min_{\alpha\in A}
		\Phi_\alpha(V,\lambda)
	\]
	by a positive gap, they would remain inactive for small $t$.
	Thus, following the deformation $(H,u)$ would contradict the maximality of $(V,\lambda)$.

	Thus, by Gordan's lemma applied to $(L_\alpha)_{\alpha\in A_0}$ and the vector space of admissible pairs $(H,u)$, there exist nonnegative weights $(w_\alpha)_{\alpha\in A_0}$ such that
	\begin{equation}
		\label{eq:proof_lem_optimality_2}
		\sum_{\alpha\in A_0} w_\alpha L_\alpha(H,u)=0
	\end{equation}
	for every admissible $(H,u)$.  
	After normalization, we may assume that
	\[
	\sum_{\alpha\in A_0}w_\alpha=1,
	\]
	as in \Cref{eq:lem_optimality_1}.	
	Define the first and second moments
	\begin{equation}
		\label{eq:proof_lem_optimality_moments}
		m_1=\sum_{\alpha\in A_0}w_\alpha r_\alpha,
		\qquad
		M_2=\sum_{\alpha\in A_0}w_\alpha r_\alpha r_\alpha^\top.
	\end{equation}
	
	\paragraph*{Step 2: Constraints on the first moment.}

	By the definition of $L_\alpha$ in \Cref{eq:proof_lem_optimality_1}, the linear relation in \Cref{eq:proof_lem_optimality_2}, and the definition of $m_1$ and $M_2$, we obtain
	\begin{equation}
		\label{eq:proof_lem_optimality_3}
		0 = \sum_{\alpha\in A_0} w_\alpha L_\alpha(H,u)
		= \tr(HM_2)+2u^\top G m_1
	\end{equation}
	for every admissible pair $(H,u)$, i.e., such that $\langle e,u\rangle=0$.
	
	Take $H=0$ in \Cref{eq:proof_lem_optimality_3}.
	It gives
	\[
		u^\top G m_1=0.
	\]
	Since this holds for every $u\in e^\perp$, we deduce that $G m_1$ belongs to the span of $e$.

	Moreover, $\ker(G)$ is spanned by $\tau_G$.
	Since $G$ is symmetric, this means that $\im(G)$ is perpendicular to $\tau_G$.
	Together with $Gm_1 \in \Span(e)$, we deduce that
	\[
		Gm_1\in \Span(e)\cap \tau_G^\perp.
	\]	
	Since $\langle e,\tau_G\rangle=1$, the right-hand side is $\{0\}$, hence $Gm_1 = 0$.
	In other words, one has
	\begin{equation}
		\label{eq:proof_lem_optimality_4}
		m_1 \in \ker(G) = \Span(\tau_G).
	\end{equation}
	
	On the other hand, every $r_\alpha=\lambda-\alpha$ in the definition of $m_1$ has coordinates that sum to 0, since $\lambda$ and $\alpha$ do.
	Consequently, the coordinates of $m_1$ sum to 0:
	\[
		\langle e,m_1\rangle
		= \sum_{\alpha\in A_0}w_\alpha \langle e,r_\alpha\rangle
		=0.
	\]
	Combined with \Cref{eq:proof_lem_optimality_4} and $\langle e,\tau_G\rangle=1$, we obtain \Cref{eq:lem_optimality_2}:
	\[
		m_1=0.
	\]
	
	\paragraph*{Step 3: Constraints on the second moment.}
	
	Now, take $u=0$ in \Cref{eq:proof_lem_optimality_3}.
	It reads
	\[
		\tr(HM_2)=0
	\]
	for every infinitesimal variation $H$ of the Gram matrix, i.e., by \Cref{lem:gram_variations}, for matrices in
	\begin{align*}
		\mathcal{H} &= \{H \in\Sym_{d+1}(\R) \mid \diag (H)=0,\ \tau_G^\top H\tau_G=0\}.
	\end{align*}
	The equality above can be rephrased as
	\[
		M_2 \in \mathcal{H}^\perp,
	\]
	where we use the dot product $(A,B)\mapsto\tr(AB)$ on the symmetric $(d+1)\times(d+1)$ matrices.

	Observe that we have
	\[
		\mathcal{H}^\perp = 
		\left\{
		D+\mu\,\tau_G\tau_G^\top
		\mid
		D \text{ diagonal},\ \mu\in\R
		\right\},
	\]
	where we take the orthogonal complement in the space of symmetric matrices.
	Indeed, the orthogonal complement of $\{H\mid \diag(H)=0\}$ coincides with the diagonal matrices.
	Besides, the space $\{H\mid \tau_G^\top H\tau_G=0\}$ has $\tau_G\tau_G^\top$ as a normal direction, since
	\[
	\tau_G^\top H\tau_G
	=
	\tr\bigl((\tau_G\tau_G^\top)H\bigr).
	\]
	Since the orthogonal complement of an intersection is the sum of the
	orthogonal complements, we obtain the expression for $\mathcal{H}^\perp$.

	Consequently, there exist a diagonal matrix $D\in\R^{(d+1)\times(d+1)}$ and a $\mu\in\R$ such that
	\begin{equation}
		\label{eq:proof_lem_optimality_5}
		M_2=D+\mu\,\tau_G\tau_G^\top.
	\end{equation}
	On the other hand, we have observed previously that every $r_\alpha=\lambda-\alpha$ in the definition of $M_2$ is perpendicular to $e$.
	That is, they satisfy $r_\alpha^\top e = 0$. 
	Therefore,
	\[
		M_2 e 
		= \sum_{\alpha\in A_0}w_\alpha r_\alpha r_\alpha^\top e  
		= 0.
	\]
	Substituting $M_2$ with its expression in \Cref{eq:proof_lem_optimality_5}, we get
	\[
		D e+\mu\,\tau_G\langle \tau_G,e\rangle=0.
	\]
	Using $\langle \tau_G,e\rangle=1$, this becomes
	\[
	D e+\mu\,\tau_G=0.
	\]
	Since $D$ is diagonal, this is equivalent to
	\[
	D=-\mu\,\diag(\tau_G).
	\]
	Using \Cref{eq:proof_lem_optimality_5} again, we deduce that
	\begin{equation}
		\label{eq:proof_lem_optimality_6}
		M_2
		=
		-\mu \big(\diag(\tau_G)-\tau_G\tau_G^\top\big).	
	\end{equation}

	To close this step, we identify the scalar $\mu$.
	First, observe that each active point $\alpha\in A_0$ realizes the maximum $\Phi^*$:
	\[
		\Phi^* = r_\alpha^\top G r_\alpha.
	\]
	This still holds for their weighted sum:
	\[
		\Phi^* = \sum_{\alpha\in A_0} w_\alpha r_\alpha^\top G r_\alpha.
	\]
	But the right-hand side coincides with the trace of the product $GM_2$, since
	\[
		\tr(GM_2)
		=
		\tr\left( G \sum_{\alpha\in A_0}w_\alpha r_\alpha r_\alpha^\top\right)
		=
		\sum_{\alpha\in A_0} w_\alpha r_\alpha^\top G r_\alpha.
	\]
	Using $M_2=-\mu \bigl(\diag(\tau_G)-\tau_G\tau_G^\top\bigr)$ in \Cref{eq:proof_lem_optimality_6}, we deduce that
	\begin{equation}
		\label{eq:proof_lem_optimality_7}
		\Phi^* 
		= \tr(GM_2)
		= -\mu \tr \bigg(G\big(\diag(\tau_G)-\tau_G\tau_G^\top\big)\bigg).
	\end{equation}
	Moreover, since $\langle e,\tau_G\rangle=1$ by definition of $\tau_G$, and since the diagonal entries of the Gram matrix $G$ are 1, we have
	\[
		\tr\big(G\diag(\tau_G)\big)=\langle e,\tau_G\rangle=1.
	\]
	On the other hand, it follows from $G\tau_G=0$ that
	\[
		\tr \left(G\tau_G\tau_G^\top\right)
		=
		0.
	\]
	This allows us to identify $\mu$ in \Cref{eq:proof_lem_optimality_7} via
	\[
		\Phi^* = -\mu (1-0) = -\mu.
	\]
	Together with \Cref{eq:proof_lem_optimality_6}, we obtain \Cref{eq:lem_optimality_3} in the statement of the lemma:
	\[
		M_2 = \Phi^*\big(\diag(\tau_G)-\tau_G\tau_G^\top\big).
	\]
	
	\paragraph*{Step 4: Semidefiniteness.}
	
	The matrix $M_2$ is positive semidefinite, since it is a weighted sum of the $r_\alpha r_\alpha^\top$, $\alpha\in A_0$, with nonnegative coefficients; see \Cref{eq:proof_lem_optimality_moments}.

	Since $\Phi^*>0$, we deduce from $M_2 = \Phi^*\big(\diag(\tau_G)-\tau_G\tau_G^\top\big)$ that the matrix
	\[
	\diag(\tau_G)-\tau_G\tau_G^\top
	\]
	is also positive semidefinite.
	Its $(i,i)$-entry is
	\[
		(\tau_G)_i - (\tau_G)_i^2 = (\tau_G)_i(1-(\tau_G)_i).
	\]
	If some coordinate $(\tau_G)_i$ were negative, then the entry $(\tau_G)_i(1-(\tau_G)_i)$ would be negative, which is impossible for a positive semidefinite matrix.
	Thus, $\tau_G$ has nonnegative coordinates.
	This shows the fourth and last statement of the lemma.
\end{proof}

\subsection{Estimate without repetitions for central points}
\label{subsec:estimate_without_repetitions}

We now prove an approximation result for points whose barycentric coordinates
are at most $1/r$, $r\in\N_{>0}$.
The goal is to approximate such a point by the barycenter of $r$ \textit{distinct} vertices.

We point out that the statement shares some similarities with \Cref{lem:pivotal_sampling}.
Whereas \Cref{lem:pivotal_sampling} gives a bound depending on the norm of the point $y$, the estimate below is a worst-case bound depending only on $d$ and $r$.
We justify in \Cref{rem:covering_without_repetitions} that the bound is sharp and is attained by the regular simplex.

\begin{lemma}
	\label{lem:covering_without_repetitions}	
	Consider a simplex in the unit ball $\sigma=\{v_0,\dots,v_{d}\}\subset B^d\subset\R^{d}$.
	Let $y\in\conv(\sigma)$ be a point with barycentric coordinates at most $1/r$, where $r$ is an integer in $\{1,\dots,d+1\}$.
	Then there exists an $r$-element subset $R\subset\{0,\dots,d\}$ such that
	\[
		\left\|
		\frac1r\sum_{i\in R}v_i-y
		\right\|
		\leq
		S_{d,r},
		\qquad
		S_{d,r}
		=
		\sqrt{\frac{d+1-r}{rd}}.
	\]
	More generally, if $\sigma$ lies in a ball of radius $\delta\geq0$, then the right-hand side is multiplied by $\delta$.
\end{lemma}

\begin{proof}
	The statement for $\delta\geq0$ follows from the statement in the unit ball by translating to the center of the ball and scaling by $\delta$.
	Therefore we only prove the unit-ball statement.

	We argue by induction on $d$, assuming the result is already known in dimensions less than $d$ and for all values of $r$.
	We distinguish the cases $r\in\{1,d+1\}$, $\{d\}$, and $\{2,\dots,d-1\}$.

	\paragraph{Case 1: extremal cases.}
	The cases $r=1$ and $r=d+1$ are elementary.
	To see this, let $(\lambda_0,\dots,\lambda_d)$ denote the barycentric coordinates of $y\in\conv(\sigma)$.
	If $r=1$, then
	\begin{equation}
		\label{eq:proof_covering_without_repetitions_base_case}
		\sum_{i=0}^d\lambda_i\|y-v_i\|^2
		=
		\sum_{i=0}^d\lambda_i\|v_i\|^2-\|y\|^2
		\leq 1.
	\end{equation}
	Thus a vertex $v_i$ is within distance $1=S_{d,1}$ from $y$. 
	For $r=d+1$, the relations $\lambda_i\leq 1/r$ and $\sum_{i=0}^d\lambda_i=1$ force $y$ to be the barycenter of $\sigma$, which is at distance $0 = S_{d,d+1}$ from itself.
	
	\paragraph{Case 2: $r = d$.}
	In this case, the barycentric coordinates of $y$ satisfy $\lambda_i\leq 1/d$.
	Define
	\[
	\nu_i=1-d\lambda_i.
	\]
	These new coordinates are nonnegative, sum to 1, and satisfy
	\[
	y
	=
	\frac1dS-\frac1dV\nu
	\qquad\text{where}\qquad
	S=\sum_{i=0}^dv_i.
	\]
	Note that $S/(d+1)$ is the barycenter of $\sigma$.
	Let $T_i = \frac{1}{d}(S-v_i)$ denote the barycenter omitting $v_i$.
	By the relation above, one deduces that, for all $i\in\{0,\dots,d\}$,
	\[
	\|y-T_i\|^2
	=
	\frac1{d^2}\|V\nu-v_i\|^2.
	\]
	Now, by the same argument as in \Cref{eq:proof_covering_without_repetitions_base_case}, some $\|V\nu-v_i\|^2$ must be less than or equal to 1.
	This shows that $y$ is within distance $1/d=S_{d,d}$ from a barycenter of $d$ vertices.
	
	\paragraph{Case 3: $2\leq r\leq d-1$.}

	We treat this case by contradiction.
	Let $(V,\lambda)$ be a maximizer of
	\begin{equation*}
		\Phi(V,\lambda)
		=
		\min_{\alpha\in A^*_{d,r}}
		(\lambda-\alpha)^\top G(\lambda-\alpha),
		\qquad
		G=V^\top V,
		\qquad
		V=[v_0\ \cdots\ v_d],
	\end{equation*}
	where the maximum is taken over all $d\times(d+1)$-matrices $V$ representing a tuple of $d+1$ points in the unit ball and barycentric coordinates $\lambda$ satisfying $\lambda_i\leq1/r$ for all $i\in\{0,\dots,d\}$, and where $A^*_{d,r}$ denotes the coordinates of the $r$-fold barycenters of $\sigma$ without repetitions:
	\[
		A^*_{d,r}
		=
		\left\{
		\alpha\in\R^{d+1} \mid
		\alpha_i=\frac{1}{r}\text{ or } 0,\ 
		\sum_{i=0}^d \alpha_i=1
		\right\}.
	\]
	By \Cref{lem:sphere_to_ball}, stated after the proof, we can assume that the maximum is attained at a spherical configuration $V\in(S^{d-1})^{d+1}$.
	Now, suppose for contradiction that
	\begin{equation}
		\label{eq:proof_covering_without_repetitions_contradiction}
		\Phi(V,\lambda) > S_{d,r}^2.
	\end{equation}
	In the remainder of the proof, we denote by $y= V\lambda$ the maximizing point.
	We distinguish three cases: either some $\lambda_i=0$, some $\lambda_i=1/r$, or all coordinates satisfy $0<\lambda_i<1/r$.
	
	\subparagraph{Subcase 3a: $\exists i\in\{0,\dots,d\},\, \lambda_i=0$.}
	
	If $\lambda_i$ is zero, then $y$ lies in the opposite face spanned by the remaining $d$ vertices, which we identify with $\R^{d-1}$.
	These vertices lie in a ball of radius 1.
	By induction, $y$ is within distance $S_{d-1,r}$ from an $r$-element barycenter.
	Hence
	\[
		\Phi(V,\lambda) \leq S_{d-1,r}^2.
	\]
	Since $S_{d-1,r}\leq S_{d,r}$, this contradicts \Cref{eq:proof_covering_without_repetitions_contradiction}.
	
	\subparagraph{Subcase 3b: $\exists i\in\{0,\dots,d\},\, \lambda_i=1/r$.}
	
	Suppose that some $\lambda_i=1/r$.
	We write
	\[
	y = V\lambda
	=
	\frac1rv_i+\frac{r-1}{r}z,
	\]
	where $z$ is in the face spanned by the remaining $d$ vertices.
	The barycentric coordinates of $z$ in that face are at most $1/(r-1)$.
	By induction, $z$ is within distance $S_{d-1,r-1}$ from an $(r-1)$-element barycenter $b$.
	In other words,
	\[
	\|z-b\|^2
	\leq
	S_{d-1,r-1}^2
	=
	\frac{d-r+1}{(r-1)(d-1)}.
	\]
	Consider the $r$-element barycenter
	\[
	\tilde{b} = \frac1rv_i+\frac{r-1}{r}b.
	\]
	It satisfies
	\[
		y-\tilde{b}
		=
		\left(\frac1rv_i+\frac{r-1}{r}z\right)
		-
		\left( \frac1rv_i+\frac{r-1}{r}b\right)
		= \frac{r-1}{r}(z-b).
	\]
	We deduce the bound
	\begin{align*}
		\big\|y-\tilde{b}\big\|^2
		\leq
		\left(\frac{r-1}{r}\right)^2 S_{d-1,r-1}^2
		=
		\frac{(r-1)(d-r+1)}{r^2(d-1)}
		<
		\frac{d-r+1}{rd}
		=
		S_{d,r}^2,
	\end{align*}
	where the strict inequality follows from $r<d$.
	Again, this contradicts \Cref{eq:proof_covering_without_repetitions_contradiction}.
	
	\subparagraph{Subcase 3c: $\forall i\in\{0,\dots,d\},\, 0<\lambda_i<1/r$.}

	We aim to apply \Cref{lem:optimality_conditions}.
	First, we observe that the configuration $V$ is affinely full-dimensional.
	Suppose it is not.
	Then we can find a nonzero vector $\tau\in\R^{d+1}$ such that	
	\[
		V\tau=0,
		\qquad
		\langle e,\tau\rangle=0,
	\]
	where $e=e_1+\dots+e_{d+1}$ is the sum of the standard basis vectors of $\R^{d+1}$. 
	In this subcase, all coordinates of $\lambda$ lie strictly between $0$ and $1/r$. Hence we may move along the line $\lambda+t\tau$ until one coordinate reaches $0$ or $1/r$, while preserving the equality
	\[
		V(\lambda+t\tau)=V\lambda.
	\]
	Let $t^*$ be such a parameter, and let $\lambda^* = \lambda+t^*\tau$.
	Some $\lambda^*_i$ is $0$ or $1/r$.
	Moreover,
	\[
	\Phi(V,\lambda^*)=\Phi(V,\lambda)>S_{d,r}^2.
	\]
	This contradicts Subcase 3a (some coordinate is zero) or Subcase 3b (some coordinate is $1/r$). 
	Therefore $V$ is affinely full-dimensional.
	
	Hence we can apply \Cref{lem:optimality_conditions} with $A=A^*_{d,r}$.
	Let $\tau=\tau_G\in\R^{d+1}$, so that
	\[
		V\tau=0,
		\qquad
		\langle e, \tau\rangle=1.
	\]	
	Let $\A$ denote the set of subsets of $\{0,\dots,d\}$ of cardinality $r$.
	Each $R\in \A$ corresponds to the coordinates $\alpha_R\in A$ of a barycenter.
	By the lemma, there are weights $w_R\geq0$, with $w_R=0$ for inactive $R$, such that
	\begin{align}
		\sum_{R\in\A}w_R&=1,
		\label{eq:proof_covering_without_repetitions_weights0}\\
		\sum_{R\in\A} w_R(\lambda-\alpha_R)&=0, \label{eq:proof_covering_without_repetitions_weights1}\\
		\sum_{R\in\A} w_R(\lambda-\alpha_R)(\lambda-\alpha_R)^\top
		&=
		\Phi^* \bigl(\diag(\tau)-\tau\tau^\top\bigr),
		\label{eq:proof_covering_without_repetitions_weights2}
	\end{align}
	where $\Phi^* = \Phi(V,\lambda)$ is the maximum, and $w_R=0$ if $R$ is not active.
	From these relations, we will deduce the following consequences: for $i,j\in \{0,\dots,d\}$ distinct,
	\begin{align}
		\lambda_i\left(\frac1r-\lambda_i\right)
		&= \Phi^*\tau_i(1-\tau_i),
		\label{eq:proof_covering_without_repetitions_equalities2}
		\\
		\tau_i 
		&=\frac{t_i/r-\Phi^*}{t_i^2-\Phi^*},
		\label{eq:proof_covering_without_repetitions_equalities4}
		\\
		t_it_j
		&\geq \Phi^*
		\label{eq:proof_covering_without_repetitions_equalities5}
	\end{align}
	where we define
	\[
		t_i=\frac{\lambda_i}{\tau_i}.
	\]
	
	We first prove \Cref{eq:proof_covering_without_repetitions_equalities2}.
	Fix $i\in \{0,\dots,d\}$.
	Define the quantity
	\begin{align*}
		\sigma_i=\sum_{R\in\A_i}w_R
		&\quad\text{where}\quad 
		\A_i = \{R\in\A\mid i\in R\}.
	\end{align*}
		For every $i\in\{0,\dots,d\}$, the $(i+1)$st row in \Cref{eq:proof_covering_without_repetitions_weights1} gives
	\[
		0
		=
		\sum_{R\in\A}w_R\bigl(\lambda_i-(\alpha_R)_i\bigr).
	\]
	Since $\sum_{R\in\A}w_R=1$ by \Cref{eq:proof_covering_without_repetitions_weights0} and $(\alpha_R)_i=1/r$ exactly when $i\in R$, we deduce that
	\[
		0 = \lambda_i-\frac1r\sum_{R\in\A_i}w_R.
	\]
	That is to say, 
	\[
		\sigma_i=r\lambda_i.
	\]
	We obtain the relations
	\begin{align}
		\sum_{R\in\A}w_R(\alpha_R)_i
		&= \frac{1}{r}\sum_{R\in\A_i}w_R
		= \frac{1}{r}\sigma_i
		= \lambda_i,
		\label{eq:proof_covering_without_repetitions_equalities1_degree1}
		\\
		\sum_{R\in\A}w_R(\alpha_R)_i^2
		&= \frac{1}{r^2}\sum_{R\in\A_i}w_R
		= \frac{1}{r^2}\sigma_i
		= \frac{\lambda_i}{r}.
		\label{eq:proof_covering_without_repetitions_equalities1_degree2}
	\end{align}
	
	Now, we take the $(i,i)$-diagonal entry of \Cref{eq:proof_covering_without_repetitions_weights2}:
	\begin{equation}
		\label{eq:proof_covering_without_repetitions_weights2_diagonal}
		\sum_{R\in\A}w_R(\lambda_i-(\alpha_R)_i)^2
		=
		\Phi^* \bigl(\tau_i-\tau_i^2).
	\end{equation}
	We expand the left-hand side as
	\[
		\sum_{R\in\A}w_R(\lambda_i-(\alpha_R)_i)^2
		= 
		\sum_{R\in\A}w_R\lambda_i^2
		+\sum_{R\in\A}w_R(\alpha_R)_i^2
		-2\sum_{R\in\A}w_R\lambda_i(\alpha_R)_i.
	\]
	We identify each term:
	\begin{align}
		\sum_{R\in\A}w_R\lambda_i^2
		&= \lambda_i^2\sum_{R\in\A}w_R = \lambda_i^2,
		\tag*{(by \Cref{eq:proof_covering_without_repetitions_weights0})}
		\\
		\sum_{R\in\A}w_R(\alpha_R)_i^2
		&= \frac{\lambda_i}{r},
		\tag*{(by \Cref{eq:proof_covering_without_repetitions_equalities1_degree2})}
		\\
		\sum_{R\in\A}w_R\lambda_i(\alpha_R)_i
		&= \lambda_i\sum_{R\in\A}w_R(\alpha_R)_i
		= \lambda_i^2.
		\tag*{(by \Cref{eq:proof_covering_without_repetitions_equalities1_degree1})}
	\end{align}
	Their sum is
	\[
		\sum_{R\in\A}w_R(\lambda_i-(\alpha_R)_i)^2
		= \lambda_i^2+\frac{\lambda_i}{r}-2\lambda_i^2
		= \lambda_i\left(\frac1r-\lambda_i\right).
	\]
	Substituting this expression in \Cref{eq:proof_covering_without_repetitions_weights2_diagonal} yields \Cref{eq:proof_covering_without_repetitions_equalities2}, as wanted:
	\[
		\lambda_i\left(\frac1r-\lambda_i\right)
		=
		\Phi^*\tau_i(1-\tau_i).
	\]
	
	We continue with \Cref{eq:proof_covering_without_repetitions_equalities4}.
	Since $0<\lambda_i<1/r$, \Cref{eq:proof_covering_without_repetitions_equalities2} implies that $0<\tau_i<1$.
	Recall that we defined $t_i=\lambda_i/\tau_i$.
	Substituting $\lambda_i=\tau_it_i$ into \Cref{eq:proof_covering_without_repetitions_equalities2} gives
	\[
	\tau_it_i\left(\frac1r-\tau_it_i\right)
	=
	\Phi^*\tau_i(1-\tau_i).
	\]
	Solving this equation for $\tau_i$ yields \Cref{eq:proof_covering_without_repetitions_equalities4}:
	\begin{equation*}
		\tau_i=\frac{t_i/r-\Phi^*}{t_i^2-\Phi^*}.
	\end{equation*}	

	Only \Cref{eq:proof_covering_without_repetitions_equalities5} is left to prove.
	For $i,j\in \{0,\dots,d\}$ distinct, define the quantity
	\[
		p_{i,j}=\sum_{R\in\A_{i,j}}w_R
		\quad\text{where}\quad 
		\A_{i,j} = \{R\in\A\mid \{i,j\}\subset R\}.
	\]
	The entry $(i,j)$ of \Cref{eq:proof_covering_without_repetitions_weights2} gives the relation
	\[
		\sum_{R\in\A}w_R(\lambda_i-(\alpha_R)_i)(\lambda_j-(\alpha_R)_j)
		=
		-\Phi^*\tau_i\tau_j.
	\]
	We expand the products on the left-hand side as
	\[
		w_R(\lambda_i-(\alpha_R)_i)(\lambda_j-(\alpha_R)_j)
		=
		w_R\lambda_i\lambda_j
		-
		w_R\lambda_i(\alpha_R)_j
		-
		w_R\lambda_j(\alpha_R)_i
		+
		w_R(\alpha_R)_i(\alpha_R)_j.
	\]
	As before, we identify each sum:
	\begin{align*}
		\sum_{R\in\A}w_R\lambda_i\lambda_j
		&=\lambda_i\lambda_j\sum_{R\in\A}w_R
		= \lambda_i\lambda_j,
		\tag*{(by \Cref{eq:proof_covering_without_repetitions_weights0})}
		\\
		\sum_{R\in\A}w_R\lambda_i(\alpha_R)_j
		&= \lambda_i\sum_{R\in\A}w_R(\alpha_R)_j
		= \lambda_i\lambda_j,
		\tag*{(by \Cref{eq:proof_covering_without_repetitions_equalities1_degree1})}
		\\
		\sum_{R\in\A}w_R\lambda_j(\alpha_R)_i
		&= \lambda_j\sum_{R\in\A}w_R(\alpha_R)_i
		= \lambda_j\lambda_i,
		\tag*{(by \Cref{eq:proof_covering_without_repetitions_equalities1_degree1})}
		\\
		\sum_{R\in\A}w_R(\alpha_R)_i(\alpha_R)_j
		&= \frac{1}{r^2}\sum_{R\in\A_{i,j}}w_R = \frac{1}{r^2}p_{i,j}.
		\tag*{(by definition of $p_{i,j}$)}
	\end{align*}
	Therefore
	\[
		\frac{1}{r^2}p_{i,j}-\lambda_i\lambda_j=-\Phi^*\tau_i\tau_j.
	\]
	Equivalently,
	\[
		p_{i,j}=r^2(\lambda_i\lambda_j-\Phi^*\tau_i\tau_j).
	\]
	Since $p_{i,j}\geq0$, we deduce \Cref{eq:proof_covering_without_repetitions_equalities5}: $t_it_j\geq \Phi^*$.
	
	We have proved \Cref{eq:proof_covering_without_repetitions_equalities2,eq:proof_covering_without_repetitions_equalities4,eq:proof_covering_without_repetitions_equalities5}.
	We deduce a last relation regarding the possible values of $t_i=\lambda_i/\tau_i$.
	On the one hand, observe that 
	\[
		r\Phi^*>\sqrt{\Phi^*}.
	\]
	Indeed, since $2\leq r\leq d-1$ in this Case 3, we have
	\begin{equation}
		\label{eq:proof_covering_without_repetitions_lastcase3}
		S_{d,r}^2-\frac1{r^2}
		=
		\frac{(r-1)(d-r)}{r^2d}
		>
		0,
	\end{equation}
	hence $\Phi^*>S_{d,r}^2>1/r^2$.
	We deduce $r\Phi^*>\sqrt{\Phi^*}$, as announced.
	Together with \Cref{eq:proof_covering_without_repetitions_equalities4} and $\tau_i\geq0$, it follows that each $t_i$ belongs to one
	of the two intervals 
	\[
		(0,\sqrt{\Phi^*})
		\quad\text{or}\quad
		(r\Phi^*,\infty).		
	\]
	On the other hand, by \Cref{eq:proof_covering_without_repetitions_equalities5},
	at most one $t_i$ can be less than $\sqrt{\Phi^*}$.
	We conclude the proof by distinguishing cases according to this condition.
	
	\subparagraph{Subcase 3c(i): $\forall i\in\{0,\dots,d\},\,t_i>\sqrt{\Phi^*}$.}
	In this case, as we have seen above, each $t_i$ belongs to the interval $(r\Phi^*,\infty)$. 
	In particular, $r\Phi^*<1$, since
	\[
	r\Phi^* = r\Phi^*\sum_{i=0}^d\tau_i < \sum_{i=0}^d\tau_it_i=\sum_{i=0}^d\lambda_i=1.
	\]
	For $t>r\Phi^*$, define the function
	\[
	h(t)=\frac{r(t^2-\Phi^*)}{t-r\Phi^*}.
	\]
	Its second derivative is
	\[
	h''(t)
	=
	\frac{2r\Phi^*(r^2\Phi^*-1)}{(t-r\Phi^*)^3},
	\]
	which is positive on $(r\Phi^*,\infty)$.
	That is, $h$ is convex. 

	Moreover, observe that $h(t_i)=1/\tau_i$ by \Cref{eq:proof_covering_without_repetitions_equalities4}, and that
	\[
		\sum_{i=0}^d\tau_i=1,
		\qquad
		\sum_{i=0}^d\tau_it_i=1,
		\qquad
		\sum_{i=0}^d\tau_ih(t_i)=d+1.
	\]
	We apply Jensen's inequality at the points $(t_0,\dots,t_d)$ with weights $(\tau_0,\dots,\tau_d)$:
	\[
		\sum_{i=0}^d\tau_i\, h(t_i)
		\geq
		h\left(\sum_{i=0}^d\tau_it_i\right).
	\]
	The left-hand side is equal to $d+1$, and the right-hand side is $h(1)$.
	We deduce that	
	\[
		d+1
		\geq
		h(1)
		=
		\frac{r(1-\Phi^*)}{1-r\Phi^*}.
	\]
	Since $1-r\Phi^*>0$, solving this inequality for $\Phi^*$ gives
	\[
	\Phi^*\leq \frac{d+1-r}{rd}=S_{d,r}^2,
	\]
	contradicting the assumption $\Phi^*>S_{d,r}^2$ in \Cref{eq:proof_covering_without_repetitions_contradiction}.
	
	\subparagraph{Subcase 3c(ii): $\exists i\in\{0,\dots,d\},\,t_i<\sqrt{\Phi^*}$.}
	We reorder the indices so that $t_0<\sqrt{\Phi^*}$. 
	The inequalities $0<\tau_0<1$, $t_0<\sqrt{\Phi^*}$, and \Cref{eq:proof_covering_without_repetitions_equalities4} applied to $\tau_0$ yield
	\begin{equation}
		\label{eq:proof_covering_without_repetitions_lastcase4}
		0<t_0<\frac1r.
	\end{equation}
	Moreover, for $j\in\{1,\dots,d\}$, \Cref{eq:proof_covering_without_repetitions_equalities5} gives
	\[
		t_j\geq \frac{\Phi^*}{t_0}.
	\]

	Equality cannot occur.
	Indeed, if $t_j=\Phi^*/t_0$, then \Cref{eq:proof_covering_without_repetitions_equalities4} on $\tau_0$ and $\tau_j$ gives
	\[
	\tau_0
	=
	\frac{t_0/r-\Phi^*}{t_0^2-\Phi^*}
	\qquad\text{and}\qquad
	\tau_j
	=
	\frac{(\Phi^*/t_0)/r-\Phi^*}{(\Phi^*/t_0)^2-\Phi^*}
	=
	\frac{t_0/r-t_0^2}{\Phi^*-t_0^2}.
	\]
	Hence
	\[
	\tau_0+\tau_j= \frac{t_0/r-\Phi^*}{t_0^2-\Phi^*} + \frac{t_0/r-t_0^2}{\Phi^*-t_0^2}  
	= 1.
	\]
	This would force the remaining $\tau_i$ to be zero.
	However, \Cref{eq:proof_covering_without_repetitions_equalities2} and the assumption $0<\lambda_i<1/r$ of this Subcase 3c force them to be positive.
	This is a contradiction.
	Therefore,
	\[
	t_j>\frac{\Phi^*}{t_0}.
	\]
	
	To conclude, we consider, for $t\in(\sqrt{\Phi^*},\infty)$, the function
	\[
	\phi(t)=\frac{t/r-\Phi^*}{t^2-\Phi^*}.
	\]
	Note that $\tau_i=\phi(t_i)$ by \Cref{eq:proof_covering_without_repetitions_equalities4}, and $\lambda_i=t_i\tau_i$ by definition of $t_i$.
	We use the identity
	\[
		t\phi(t)
		+
		\frac{t_0-r\Phi^*}{rt_0-1}\phi(t)
		-
		\frac1r
		=
		\frac{(t_0t-\Phi^*)(r^2\Phi^*-1)}
		{r(t^2-\Phi^*)(1-rt_0)}.
	\]
	For $t\geq \Phi^*/t_0$, the right-hand side is nonnegative, since
	\begin{align*}
		t_0&<\sqrt{\Phi^*}, 
		\tag*{(by assumption of this Subcase 3c(ii))}
		\\
		r^2\Phi^*&>1, 
		\tag*{(by \Cref{eq:proof_covering_without_repetitions_contradiction,eq:proof_covering_without_repetitions_lastcase3})}
		\\
		1-rt_0&>0.
		\tag*{(by \Cref{eq:proof_covering_without_repetitions_lastcase4})}
	\end{align*}
	Thus,
	\[
		t\phi(t)
		+
		\frac{t_0-r\Phi^*}{rt_0-1}\phi(t)
		\geq \frac1r.
	\]
	We apply this inequality to $t_1,\dots,t_d$ and sum over $i=1,\dots,d$:
	\begin{equation}
		\sum_{i=1}^d\left(t_i\phi(t_i)
		+
		\frac{t_0-r\Phi^*}{rt_0-1}\phi(t_i)\right)
		=
		\sum_{i=1}^d\lambda_i
		+
		\frac{t_0-r\Phi^*}{rt_0-1}
		\sum_{i=1}^d \tau_i
		\geq
		\frac dr.
		\label{eq:proof_covering_without_repetitions_lastcase1}
	\end{equation}
	
	On the other hand, we have the relation
	\begin{equation}
		1-t_0\tau_0
		+
		\frac{t_0-r\Phi^*}{rt_0-1}(1-\tau_0)
		=
		1.
		\label{eq:proof_covering_without_repetitions_lastcase2}
	\end{equation}	
	Indeed, $\tau_0=(t_0/r-\Phi^*)/(t_0^2-\Phi^*)$ from \Cref{eq:proof_covering_without_repetitions_equalities4} yields
	\[
		1-\tau_0 
		= 1 - \frac{t_0/r-\Phi^*}{t_0^2-\Phi^*}
		= \frac{t_0^2-\Phi^* - t_0/r+\Phi^*}{t_0^2-\Phi^*}
		= \frac{t_0^2- t_0/r}{t_0^2-\Phi^*}
		= \frac{t_0}{r}\cdot\frac{rt_0-1}{t_0^2-\Phi^*}.
	\]	
	Consequently,
	\[
		\frac{t_0-r\Phi^*}{rt_0-1}(1-\tau_0)
		=
		\frac{t_0-r\Phi^*}{rt_0-1}
		\cdot
		\frac{t_0}{r}\cdot\frac{rt_0-1}{t_0^2-\Phi^*}
		=
		\frac{t_0}{r}\cdot\frac{t_0-r\Phi^*}{t_0^2-\Phi^*}
		=
		t_0\frac{t_0/r-\Phi^*}{t_0^2-\Phi^*}
		=
		t_0\tau_0.
	\]
	We deduce \Cref{eq:proof_covering_without_repetitions_lastcase2}:
	\[
	1-t_0\tau_0+\frac{t_0-r\Phi^*}{rt_0-1}(1-\tau_0)
	=
	1-t_0\tau_0+t_0\tau_0
	=
	1.
	\]

	The left-hand sides of \Cref{eq:proof_covering_without_repetitions_lastcase1,eq:proof_covering_without_repetitions_lastcase2} coincide, since
	\[
		\sum_{i=1}^d\lambda_i = 1-\lambda_0 = 1-t_0\tau_0
		\qquad\text{and}\qquad
		\sum_{i=1}^d \tau_i = 1-\tau_0.
	\]
	We deduce that $1\geq d/r$, so $d\leq r$, which contradicts $r\leq d-1$.
	This proves the lemma.
\end{proof}

\begin{remark}
	\label{rem:covering_without_repetitions}
	The constant $S_{d,r} = \sqrt{(d+1-r)/(rd)}$ in \Cref{lem:covering_without_repetitions} is sharp and is attained by a regular simplex with vertices in $S^{d-1}$.
	Indeed, in this case, the vertices satisfy $\langle v_i,v_j\rangle = -1/d$ for all $i\neq j$.
	Hence, for every $r$-subset $R\subset\{0,\dots,d\}$, one has
	\[
	\left\|
	\frac1r\sum_{i\in R}v_i
	\right\|^2
	=
	\frac1{r^2}
	\left(
	r-\frac{r(r-1)}{d}
	\right)
	=
	\frac{d+1-r}{rd}
	=
	S_{d,r}^2.
	\]
	That is to say, the origin is at distance $S_{d,r}$ from the $r$-fold barycenters without repetitions.
\end{remark}

Finally, we state the lemma used in Case 3 of the proof of \Cref{lem:covering_without_repetitions}.

\begin{lemma}
	\label{lem:sphere_to_ball}
	Let $A\subset\Delta_d$ be finite. 
	For $V=[v_0\ \cdots\ v_d]\in(B^d)^{d+1}$ and $\lambda\in\Delta_d$, define
	\[
		\Phi(V,\lambda)
		=
		\min_{\alpha\in A}
		\|V(\lambda-\alpha)\|^2.
	\]
	Then, for every such $V$ and $\lambda$, there exists a spherical configuration $U\in(S^{d-1})^{d+1}$ with
	\[
		\Phi(U,\lambda)\geq \Phi(V,\lambda).
	\]
\end{lemma}

\begin{proof}
	Let $V=[v_0\ \cdots\ v_d]\in(B^d)^{d+1}$. 
	For each $i\in\{0,\dots,d\}$, lift $v_i\in B^d$ to a point
	\[
	\widehat v_i=(v_i,t_i)\in \R^{d+1}
	\quad
	\text{where}
	\quad
	t_i=\sqrt{1-\|v_i\|^2}.
	\]
	Then the lifted points lie on the unit sphere $S^{d}\subset\R^{d+1}$.
	Consider the affine hull
	\[
		\Pi=\aff\{\widehat v_0,\dots,\widehat v_d\},
	\]
	and let $c$ be the orthogonal projection of the origin onto $\Pi$. 
	By orthogonality, we have
	\[
		1
		=
		\|\widehat v_i\|^2
		=
		\|c\|^2+\|\widehat v_i-c\|^2.
	\]
	Thus all the lifted points $\widehat v_i$ lie on the sphere in $\Pi$ centered at $c$, with radius
	\[
		\delta=\sqrt{1-\|c\|^2}\leq 1.
	\]
	
	If $\delta=0$, then all $\widehat v_i$ coincide and the result is trivial.
	Assume now that $\delta>0$. 
	Let
	\[
		u_i=\frac{\widehat v_i-c}{\delta}.
	\]
	Since $\Pi$ has dimension at most $d$, we can identify $\Pi$ with a subspace of $\R^d$. 
	Then each $u_i$ is a unit vector in $\Pi$.
	Thus, the matrix
	\[
		U=[u_0\ \cdots\ u_d]
	\]
	can be seen as an element of $(S^{d-1})^{d+1}$.
	
	To conclude, consider a $\beta\in\R^{d+1}$ whose coordinates sum to 0.
	The center $c$ cancels:
	\[
		U\beta
		=
		\sum_{i=0}^d\beta_i u_i
		=
		\frac1\delta\sum_{i=0}^d\beta_i\widehat v_i.
	\]
	We expand the norm using $\widehat v_i=(v_i,t_i)$:
	\[
		\|U\beta\|^2
		=
		\frac1{\delta^2}
		\left\|
		\sum_{i=0}^d\beta_i\widehat v_i
		\right\|^2
		=
		\frac1{\delta^2}
		\left(
		\left\|\sum_{i=0}^d\beta_i v_i\right\|^2
		+
		\left(\sum_{i=0}^d\beta_i t_i\right)^2
		\right).
	\]
	Since $\delta\leq1$, it follows that
	\[
		\|U\beta\|^2
		\geq
		\left\|\sum_{i=0}^d\beta_i v_i\right\|^2
		=
		\|V\beta\|^2.
	\]
	Applying this inequality to $\beta=\lambda-\alpha$ yields
	\[
		\|U(\lambda-\alpha)\|^2
		\geq
		\|V(\lambda-\alpha)\|^2.
	\]
	Minimizing over $\alpha\in A$ yields the result.
\end{proof}

\subsection{Proof of the covering theorem for edgewise refinement}
\label{subsec:proof_covering_kfold}

We now prove the main result of this section, following the
strategy outlined in \Cref{subsec:statement_sharp_bounds_kfold}.

\begin{proof}[Proof of \Cref{thm:steiner_covering_kfold}]
	\label{proof:steiner_covering_kfold}
	Fix $y\in\conv(\sigma)$ and denote its barycentric coordinates by $(\lambda_0,\dots,\lambda_d)$.
		
	\paragraph*{Step 1: Rounding the barycentric coordinates.}
	
	For $i\in\{0,\dots,d\}$, define the quantities
	\[
		a_i=\lfloor k\lambda_i\rfloor,
		\qquad
		b_i = k\lambda_i - a_i.
	\]
	Note that $a_i$ is an integer and that $b_i$ belongs to $[0,1)$.
	Moreover, one has the relation 
	\[
		\lambda_i = \frac{a_i}{k} + \frac{b_i}{k}.
	\]
	Define $r = \sum_{i=0}^d b_i$.
	It is an integer, since
	\[
		\sum_{i=0}^d b_i = k \sum_{i=0}^d \lambda_i - \sum_{i=0}^d a_i
		= k - \sum_{i=0}^d a_i.
	\]
	Moreover, $r\leq k$ since $\sum_{i=0}^d a_i\geq0$.
	Last, $r\leq d$ since each $b_i$ is less than 1.
	
	If $r=0$, then $\lambda$ belongs to $A_{d,k}$, the coordinates of the $k$-barycenters, defined in \Cref{eq:lambda_d_k}, hence there is nothing to prove.
	Assume now that $r>0$.
	Define
	\[
	\mu_i=\frac{b_i}{r}, \qquad i\in\{0,\dots,d\}.
	\]
	Then the $\mu_i$ sum to 1 and are all smaller than $1/r$.
	Hence we can apply \Cref{lem:covering_without_repetitions}: there exists an $r$-element subset $R\subset\{0,\dots,d\}$ such that
	\begin{equation}
		\left\|
		\frac1r\sum_{i\in R}v_i
		- V\mu
		\right\|
		\leq
		\sqrt{\frac{d+1-r}{rd}},
		\qquad
		V=[v_0\ \cdots\ v_d].
		\label{eq:proof_covering_kfold_rounding}
	\end{equation}
	We update the integers $(a_0,\dots,a_d)$ as follows:
	\[
		\widetilde{a}_i =
			\left\{
				\begin{array}{ll}
					a_i+1 &\text{ if } i\in R,\\
					a_i &\text{ otherwise}.
				\end{array}
			\right.
	\]	
	The $\widetilde{a}_i$ sum to $k$, since
	\[
		\sum_{i=0}^d \widetilde{a}_i
		=
		\sum_{i=0}^d a_i+r
		=
		k.
	\]
	Accordingly, we define new barycentric coordinates $\widetilde{\lambda}_0,\dots,\widetilde{\lambda}_d$ by
	\[
		\widetilde{\lambda}_i = \frac{\widetilde{a}_i}{k}.
	\]	
	The associated point $x = \sum_{i=0}^d \widetilde{\lambda}_iv_i$ is a $k$-fold barycenter, since its coordinates are multiples of $1/k$.
	We show that $x$ is close to $y$.
	
	\paragraph*{Step 2: Distance estimate.}

	Since $\lambda_i = a_i/k + (r/k)\mu_i$ and $\widetilde{\lambda}_i = \widetilde{a}_i/k$, their difference is
	\[
		\lambda_i - \widetilde{\lambda}_i =
			\left\{
			\begin{array}{ll}
				\displaystyle\frac{r}{k}\left(\mu_i-\frac{1}{r}\right) &\text{ if } i\in R,\\[.5cm]
				\displaystyle\frac{r}{k}\mu_i &\text{ otherwise}.
			\end{array}
			\right.
	\]
	Consequently,
	\[
		y - x
		= \sum_{i=0}^d\lambda_iv_i-\sum_{i=0}^d \widetilde{\lambda}_iv_i
		= \frac{r}{k}\left(\sum_{i=0}^d\mu_iv_i - \frac{1}{r}\sum_{i\in R}v_i\right).
	\]
	Combined with \Cref{eq:proof_covering_kfold_rounding}, we get
	\[
		\|x-y\|
		\leq
		\frac{r}{k}
		\left\|
		\sum_{i=0}^d\mu_iv_i - \frac{1}{r}\sum_{i\in R}v_i
		\right\|
		\leq
		\frac{r}{k}\sqrt{\frac{d+1-r}{rd}}
		= 
		\sqrt{\frac{r(d+1-r)}{dk^2}}.
	\]
	
	To conclude, we show the bound 
	\begin{equation}
		r(d+1-r) \leq s(d+1-s),
		\qquad s=\min\left\{k,\left\lfloor\frac{d+1}{2}\right\rfloor\right\}.
		\label{eq:proof_covering_kfold_rounding2}
	\end{equation}
	First, since $r\leq k$ and $r\leq d$, we have
	\[
		\min\{r,d+1-r\} \leq s.
	\]
	On the other hand, the map $x\mapsto x(d+1-x)$ is invariant under the transformation $x\mapsto d+1-x$ and is increasing on $[0,(d+1)/2]$.
	We deduce \Cref{eq:proof_covering_kfold_rounding2} and the result.	
\end{proof}

\subsection{Relation to a conjecture on cylinder coverings}
\label{subsec:link_with_conjecture}

Bárány and Füredi conjectured in \cite[Lemma~2]{BaranyFuredi1988} that
for every simplex $\sigma=\{v_0,\dots,v_d\}\subset B^d$,
\begin{equation}
	\label{eq:barany_furedi}
	\conv(\sigma)
	\subset
	\bigcup_{\substack{\tau\subset \sigma\\ |\tau|=k}}
	\conv(\tau)^{\zeta_{d,k}},
	\qquad
	\zeta_{d,k}=\sqrt{\frac{d-k+1}{dk}},
\end{equation}
where the union is taken over all $k$-element subsets $\tau\subset\sigma$, and where $\conv(\tau)^{\delta}$ denotes the \emph{$\delta$-cylinder} over $\conv(\tau)$.
This is the set of points $y\in\R^d$ whose orthogonal projection onto the affine hull $\aff(\tau)$ is at distance at most $\delta$ from $y$ and lies in the convex hull $\conv(\tau)$.
In this case, we say that $y$ \textit{projects orthogonally} onto $\conv(\tau)$.

\Cref{thm:steiner_covering_kfold} implies this conjecture in the range $1\leq k\leq \lfloor(d+1)/2\rfloor$.
In that range, our constant $R_{d,k}$ coincides with $\zeta_{d,k}$. 
For $k>\lfloor(d+1)/2\rfloor$, our bound has a different form.

\begin{corollary}
	\label{cor:barany_furedi}
	\Cref{eq:barany_furedi} holds for $1\leq k\leq \left\lfloor \frac{d+1}{2}\right\rfloor$.
\end{corollary}

\begin{proof}
	For $k=1$ the proof is immediate.
	Therefore we suppose $k>1$.
	Let $y\in\conv(\sigma)$.
	Since $k\leq \left\lfloor \frac{d+1}{2}\right\rfloor$, the constant $R_{d,k}$ in \Cref{thm:steiner_covering_kfold} is equal to $\zeta_{d,k}$.
	Thus, there exists a $k$-tuple $(v_{i_1},\dots,v_{i_k})$ of vertices of $\sigma$, repetitions allowed, whose barycenter $x=\frac1k\sum_{\ell=1}^k v_{i_\ell}$ satisfies
	\[
		\|x-y\|\leq \zeta_{d,k}.
	\]
	Let $F_0=\conv(\tau)$ be the face of $\sigma$ spanned by the distinct vertices among $v_{i_1},\dots,v_{i_k}$.
	Since $x\in F_0$, we have
	\[
		\dist(y,F_0)\leq \|y-x\|\leq \zeta_{d,k}.
	\]
	To prove the result, we must pass from $F_0$ to another face spanned by exactly $k$ vertices and onto which $y$ projects orthogonally.

	Let $q_1$ be the nearest point of $F_0$ to $y$. 
	It satisfies
	\[
		\|y-q_1\|
		=
		\dist(y,F_0)
		\leq
		\zeta_{d,k}.
	\]
	Let $F_1$ be the minimal face of $F_0$ containing $q_1$.
	Then $q_1$ lies in the relative interior of $F_1$, and $y-q_1$ is orthogonal to the affine space spanned by $F_1$.
	Thus, by definition of the cylinder,
	\[
		y\in (F_1)^{\zeta_{d,k}}.
	\]
	
	If $F_1$ is spanned by $k$ vertices, we are done. 
	Otherwise, we enlarge $F_1$ while preserving the orthogonality condition.  
	Suppose that $F_1$ is spanned by fewer than $k$ vertices, and that $y\notin F_1$.
	Let $\lambda_0,\dots,\lambda_d$ be the barycentric coordinates of $y$ in $\sigma$.
	Define the normal vector 
	\[
		n=y-q_1.
	\]
	For every vertex $v_i$ of $F_1$, we have $\langle n,v_i-q_1\rangle=0$.
	It follows that
	\[
		\|n\|^2
		=
		\langle n,y-q_1\rangle
		=
		\sum_{i=0}^d\lambda_i\langle n,v_i-q_1\rangle
		=
		\sum_{\substack{i=0,\dots,d\\v_i\notin F_1}} \lambda_i\langle n,v_i-q_1\rangle. 
	\]
	Since $\|n\|>0$, there exists a vertex $v\notin F_1$ such that $\langle n,v-q_1\rangle>0$.
	Consider the new face
	\begin{equation}
		\label{eq:proof_barany_furedi_1}
		F_1^+=\conv(F_1\cup\{v\}).
	\end{equation}
	We claim that 
	\begin{equation}
		\label{eq:proof_barany_furedi_2}
		\dist(y,F_1^+)<\dist(y,F_1)\leq \zeta_{d,k}.		
	\end{equation}
	Indeed, for small $\epsilon>0$, the point
	\[
		q_1^\epsilon=(1-\epsilon)q_1+\epsilon v
	\]
	belongs to $F_1^+$, and
	\[
		\|y-q_1^\epsilon\|^2
		=
		\|n-\epsilon(v-q_1)\|^2
		=
		\|n\|^2
		+ \epsilon^2\|v-q_1\|^2
		-2\epsilon\langle n,v-q_1\rangle.
	\]
	Since $\langle n,v-q_1\rangle>0$, the right-hand side is less than $\|n\|^2$ for small $\epsilon$.
	We deduce \Cref{eq:proof_barany_furedi_2}.
	However, $y$ may not project orthogonally onto $F_1^+$, so we modify it.
	
	Let $q_2$ be the nearest point of $F_1^+$ to $y$, and let $F_2$ be the minimal face of $F_1^+$ containing $q_2$.
	Then $q_2$ lies in the relative interior of $F_2$, and $y-q_2$ is orthogonal to the affine space spanned by $F_2$.
	Therefore, $y$ belongs to the cylinder over $F_2$, i.e., $y\in (F_2)^{\zeta_{d,k}}$.
	
	We repeat this procedure.
	It produces a sequence of faces
	\[
		F_1,\;F_2,\;F_3,\;\dots
	\]
	At each step, the distance from $y$ to the current face strictly decreases by \Cref{eq:proof_barany_furedi_2}, so the process cannot visit the same face twice.
	Since $\sigma$ has finitely many faces, the process terminates.  
	It can terminate only when \textbf{(1)} the current face is spanned by exactly $k$ vertices, or \textbf{(2)} $y$ belongs to the current face.
	Otherwise, the construction in \Cref{eq:proof_barany_furedi_1} would allow us to continue with a better face.
	In case \textbf{(1)}, the result follows immediately.
	In case \textbf{(2)}, we enlarge the current face by adding arbitrary vertices until it is spanned by exactly $k$ vertices.
	In all cases, there exists a $k$-element subset $\tau\subset\sigma$ such that $y\in\conv(\tau)^{\zeta_{d,k}}$.
\end{proof}

\section{Numerical experiments on asymptotic mesh size}
\label{sec:experiments}

The contraction factors in \Cref{table:constant_alpha} are worst-case estimates, and they do not by themselves determine which refinement scheme performs best in practice.
Indeed, the refinements insert different numbers of vertices.
In this section, we compare the resulting meshes experimentally, using normalized quantities that account for the number of vertices. 
The goal is to illustrate how recomputing the Delaunay complex can improve mesh size and quality under iteration.

\subsection{Asymptotic reference constants for mesh size}
\label{subsec:normalized_mesh_size}

To compare different methods fairly, we normalize the mesh size by its natural asymptotic scale. 
For background on covering densities and asymptotic covering problems, we refer to \cite[Chapter~2]{handbookgeometry2017} and \cite[Chapter~6]{B_r_czky_Jr_2004}.
The proofs of this section are postponed to \Cref{subsec:additional_proofs_covering_density}.

Given a finite set $X\subset S^d$, let $\covrad(X)$ denote its covering radius, as in \Cref{eq:def_covering_radius}.
It is known that optimal coverings of $S^d$ have covering radius of order $|X|^{-1/d}$ as $|X|\to\infty$, where $|X|$ is the cardinality of $X$.
Therefore we define the \emph{normalized covering radius} by
\[
\overline{\covrad}(X)
=
|X|^{1/d}\,\covrad(X).
\]
Similarly, if $K$ is an embedded triangulation of $S^d$ with vertex set $V(K)$, we define the \emph{normalized maximal circumradius} and \emph{normalized maximal diameter} by
\begin{align*}
	\overline{\maxcirc}(K)
	&=
	|V(K)|^{1/d}\,\maxcirc(K),\\
	\overline{\maxdiam}(K)
	&=
	|V(K)|^{1/d}\,\maxdiam(K).
\end{align*}
These quantities measure the mesh size after correcting for the number of vertices.

As a reference value for the normalized covering radius, define the asymptotic constant
\[
	\overline{R}_d
	=
	\liminf_{n\to\infty}
	\inf_{\substack{X\subset S^d\\ |X|=n}}
	\overline{\covrad}(X),
\]
where the infimum is taken over all subsets of $S^d$ of cardinality $n$.
It is expressed in terms of the \emph{covering density} $\vartheta(B^d)$ of $\R^d$ \cite[Chapter~2.1]{handbookgeometry2017} by the following formula.

\begin{lemma}[proof p.~\pageref{proof:covering_density}]
	\label{lem:covering_density}
	For every $d\geq1$,
	\[
		\overline{R}_d
		=
		\left(
		\frac{\vol(S^d)}{\vol(B^d)}\,\vartheta(B^d)
		\right)^{1/d}.
	\]
\end{lemma}

Exact values of $\vartheta(B^d)$ are known only in dimensions $1$ and $2$.
In dimension $3$, an upper bound is given by the optimal \textit{lattice covering density} \cite[Table~2.1.3]{handbookgeometry2017}.
The values are
\[
	\vartheta(B^1)=1,
	\qquad
	\vartheta(B^2)=\frac{2\pi}{\sqrt{27}},
	\qquad
	\vartheta(B^3)\leq \frac{5\sqrt{5}\pi}{24}.
\]
Together with \Cref{lem:covering_density}, we get the following estimates for normalized covering radii on $S^d$:
\begin{equation}
	\label{eq:minimal_covering_radius}
	\overline{R}_2
	=
	\sqrt{\frac{8\pi}{\sqrt{27}}}
	\approx 2.199,
	\qquad
	\overline{R}_3
	\leq
	\sqrt[3]{\frac{5\sqrt{5}\pi^2}{16}}
	\approx 1.904.
\end{equation}

We also introduce the analogous asymptotic constants for triangulations:
\begin{align*}
	\overline{C}_d
	&=
	\liminf_{n\to\infty}
	\inf_{\substack{K\text{ triangulation}\\[.15em] |V(K)|=n}}
	\overline{\maxcirc}(K),\\
	\overline{D}_d
	&=
	\liminf_{n\to\infty}
	\inf_{\substack{K\text{ triangulation}\\[.15em] |V(K)|=n}}
	\overline{\maxdiam}(K),
\end{align*}
where the infima are taken over embedded triangulations of $S^d$ with $n$ vertices.
The next lemma shows that the constants $\overline{C}_d$ and $\overline{D}_d$ are related to the optimal covering radius $\overline{R}_d$.

\begin{lemma}[proof p.~\pageref{proof:circumradius_diameter_density}]
	\label{lem:circumradius_diameter_density}
	For every $d\geq1$,
	\[
		\overline{C}_d= \overline{R}_d
		\qquad\text{and}\qquad
		\frac{1}{J_d}\,\overline{R}_d
		\leq
		\overline{D}_d
		\leq
		2\,\overline{R}_d,
		\quad
		J_d=\sqrt{\frac{d}{2(d+1)}}.	
	\]
\end{lemma}

In dimension $2$, the lower bound $\overline{R}_2/J_2\leq \overline{D}_2$ is sharp, and is attained asymptotically by the equilateral triangular lattice. 
However, for $d\geq3$ the exact values of $\overline{D}_d$ do not seem to be known.
We use the upper bound $\overline{D}_3\leq 2\overline{R}_3$. 
With \Cref{eq:minimal_covering_radius}, we obtain
\[
	\overline{D}_2
	=
	\sqrt{\frac{8\pi}{\sqrt{3}}}
	\approx 3.809,
	\qquad
	\overline{D}_3
	\leq
	2\sqrt[3]{\frac{5\sqrt{5}\pi^2}{16}}
	\approx 3.807.
\]
These constants will serve as reference values for the experiments below.

\subsection{Experimental comparison of iterated refinements and subdivisions}

\subsubsection{Description of the experiments}

We compare the Delaunay refinement schemes introduced in \Cref{subsec:global_delaunay_refinements}: minicenter, centroid, barycentric, and $k$-edgewise refinements for $k\in\{2,3,5\}$. 
We run the experiments on $S^2$ and $S^3$, starting either from the regular simplex or from random initial configurations. 
For each method, we iterate the refinement a certain number of times and report the final normalized covering radius, maximal circumradius, and maximal diameter, as defined in \Cref{subsec:normalized_mesh_size}.

We also compare Delaunay refinements with their classical counterparts when these exist, namely barycentric and edgewise subdivision.
For the random experiments on $S^2$ and $S^3$, we generate $100$ random initial samples, with $10$ points per sample on $S^2$ and $30$ points per sample on $S^3$. 
We use the same samples for the Delaunay and classical procedures, and report the mean and standard deviation.
The initial triangulation for the classical subdivisions is the Delaunay triangulation of the same point set.
Non-admissible point sets, in the sense of \Cref{subsec:delaunay_triangulations}, are discarded and resampled.

For reproducibility, we specify how the classical $k$-edgewise subdivision is implemented, since it depends on an ordering of the vertices.
We first fix a global ordering by vertex indices. 
Each facet $\{v_0,\dots,v_d\}$ is written with its vertices in increasing order, and the edgewise subdivision is applied with respect to this order. 
A new vertex $(v_{i_1},\dots,v_{i_k})$ is indexed by the sorted $k$-tuple of old vertices (repetitions allowed) whose barycenter it represents.
The original vertices, corresponding to tuples $(v_i,\dots,v_i)$, keep their original indices.
The new vertices are inserted in lexicographic order of these $k$-tuples. 
This construction makes the subdivision independent of the order in which the facets are processed.

\subsubsection{Results}

The results are shown in \Cref{tab:mesh_size}. 
The four panels correspond respectively to the regular and random configurations on $S^2$ (\Cref{tab:subfloat:2sphere_regular,tab:subfloat:2sphere_random}) and on $S^3$ (\Cref{tab:subfloat:3sphere_regular,tab:subfloat:3sphere_random}).

For regular initial configurations (\Cref{tab:subfloat:2sphere_regular,tab:subfloat:3sphere_regular}), the relative performance of minicenter and centroid refinement depends on the dimension. 
On $S^2$, centroid refinement gives smaller normalized quantities than minicenter refinement, whereas the opposite is true on $S^3$.
This is consistent with the constants in \Cref{thm:steiner_covering}: the centroid constant is $2/3$ in dimension $2$ and $4/5$ in dimension $3$, while the minicenter constant is $1/\sqrt{2}$ in both dimensions. 
However, for random initial configurations (\Cref{tab:subfloat:2sphere_random,tab:subfloat:3sphere_random}), minicenter refinement performs better in both dimensions. 
This suggests that the worst-case bounds in \Cref{thm:steiner_covering} are not representative of the typical random configurations considered here.

The comparison between Delaunay barycentric refinement and classical barycentric subdivision shows a clear effect in all four panels.
The normalized covering radii are similar, but the normalized maximal circumradii and diameters are much smaller for the Delaunay version. 
This reflects the fact that barycentric subdivision may create poorly shaped simplices, while the Delaunay triangulation can change the simplices after each refinement step.

The results on edgewise refinement are more dimension-dependent.
On $S^2$ with the regular initial configuration (\Cref{tab:subfloat:2sphere_regular}), the Delaunay and classical edgewise constructions give the same values.
Indeed, the small triangles produced by edgewise subdivision are already Delaunay, so recomputing the triangulation does not change it.
For random configurations on $S^2$ (\Cref{tab:subfloat:2sphere_random}), the Delaunay version gives slightly better covering radii and circumradii.

The effect is clearer on $S^3$.
For random initial configurations (\Cref{tab:subfloat:3sphere_random}), the Delaunay edgewise refinements consistently improve the normalized maximal circumradius and diameter compared with classical edgewise subdivision.
For the regular initial configuration (\Cref{tab:subfloat:3sphere_regular}), the comparison is mixed: Delaunay $2$-edgewise refinement is clearly better than $2$-edgewise subdivision, the Delaunay $3$-edgewise refinement gives a worse covering radius but better maximal circumradius and diameter, and the $5$-edgewise values coincide.

Overall, the experiments support the idea that Delaunay refinements can substantially improve mesh quality as measured by normalized circumradii and diameters.
Moreover, the contraction constants of \Cref{thm:steiner_covering} capture only worst-case shrinkage; the practical performance depends on the case.
The $3$-edgewise refinement performs best for the regular configuration on $S^2$, minicenter refinement performs best for random configurations on $S^2$, and $2$-edgewise refinement performs best for both regular and random configurations on $S^3$.

\begin{table}[!htbp]
	\centering
	\captionsetup[subfloat]{font=normalsize}
	\setlength{\tabcolsep}{5pt}
	\renewcommand{\arraystretch}{1.25}
	\subfloat[
		Results on $S^2$ starting from the regular simplex.
	]{%
		\label{tab:subfloat:2sphere_regular}
		\begin{tabular}{|c|cccccc|}
			\hline
			Method
			& Minicenter
			& Centroid
			& Barycentric
			& $2$-Edgewise
			& $3$-Edgewise
			& $5$-Edgewise
			\\\hline
			\multicolumn{7}{|c|}{Delaunay refinements}
			\\\hline
			$\overline{\covrad}$ & 4.65 & 3.06 & 5.03 & 4.00 & \textbf{2.68} & 4.00 \\
			$\overline{\maxcirc}$ & 4.65 & 3.06 & 5.03 & 4.00 & \textbf{2.68} & 4.00 \\
			$\overline{\maxdiam}$ & 9.31 & 5.82 & 9.94 & 6.92 & \textbf{4.85} & 6.93 \\
			Iterations & 6 & 6 & 5 & 6 & 5 & 3 \\
			\hline
			\multicolumn{7}{|c|}{Classical subdivisions}
			\\\hline
			$\overline{\covrad}$ & --- & --- & 5.03 & 4.00 & \textbf{2.68} & 4.00 \\
			$\overline{\maxcirc}$ & --- & --- & 51.77 & 4.00 & \textbf{2.68} & 4.00 \\
			$\overline{\maxdiam}$ & --- & --- & 19.24 & 6.92 & \textbf{4.85} & 6.93 \\
			Iterations & --- & --- & 5 & 6 & 5 & 3 \\
			\hline
		\end{tabular}
		}
\end{table}

\begin{table}[!htbp]
	\ContinuedFloat
	\centering
	\captionsetup[subfloat]{font=normalsize}
	\setlength{\tabcolsep}{4pt}
	\renewcommand{\arraystretch}{1.25}
	\subfloat[
			Results on $S^2$ starting from $100$ random samples of $10$ points.
	]{%
		\label{tab:subfloat:2sphere_random}
		\begin{tabular}{|c|cccccc|}
			\hline
			Method
			& Minicenter
			& Centroid
			& Barycentric
			& $2$-Edgewise
			& $3$-Edgewise
			& $5$-Edgewise
			\\\hline
			\multicolumn{7}{|c|}{Delaunay refinements}
			\\\hline
			$\overline{\covrad}$ & $\bm{5.2 \pm 0.7}$ & $7.5 \pm 5.2$ & $13.6 \pm 5.0$ & $8.9 \pm 7.7$ & $11.7 \pm 9.9$ & $11.9 \pm 9.7$ \\
			$\overline{\maxcirc}$ & $\bm{5.2 \pm 0.7}$ & $7.5 \pm 5.2$ & $13.6 \pm 5.0$ & $8.9 \pm 7.7$ & $11.7 \pm 9.9$ & $11.9 \pm 9.7$ \\
			$\overline{\maxdiam}$ & $\bm{10.4 \pm 1.4}$ & $14.1 \pm 9.5$ & $26.7 \pm 9.8$ & $15.9 \pm 13.1$ & $22.4 \pm 19.1$ & $22.6 \pm 18.6$ \\
			Iterations & 6 & 6 & 5 & 6 & 4 & 3 \\
			\hline
			\multicolumn{7}{|c|}{Classical subdivisions}
			\\\hline
			$\overline{\covrad}$ & --- & --- & $14.7 \pm 6.1$ & $10.0 \pm 8.4$ & $11.7 \pm 9.9$ & $11.9 \pm 9.7$ \\
			$\overline{\maxcirc}$ & --- & --- & $337.4 \pm 32.6$ & $11.0 \pm 10.2$ & $11.8 \pm 10.0$ & $12.2 \pm 10.2$ \\
			$\overline{\maxdiam}$ & --- & --- & $61.4 \pm 16.1$ & $17.5 \pm 14.5$ & $22.3 \pm 19.1$ & $22.6 \pm 18.6$ \\
			Iterations & --- & --- & 5 & 6 & 4 & 3 \\
			\hline
		\end{tabular}
	}
\end{table}

\begin{table}[!htbp]
	\ContinuedFloat
	\centering
	\captionsetup[subfloat]{font=normalsize}
	\setlength{\tabcolsep}{5pt}
	\renewcommand{\arraystretch}{1.25}
	\subfloat[
			Results on $S^3$ starting from the regular simplex.
		]{%
		\label{tab:subfloat:3sphere_regular}
		\begin{tabular}{|c|cccccc|}
			\hline
			Method
			& Minicenter
			& Centroid
			& Barycentric
			& $2$-Edgewise
			& $3$-Edgewise
			& $5$-Edgewise
			\\\hline
			\multicolumn{7}{|c|}{Delaunay refinements}
			\\\hline
			$\overline{\covrad}$ & 3.19 & 4.20 & 6.07 & \textbf{2.79} & 4.29 & 4.11 \\
			$\overline{\maxcirc}$ & 3.19 & 4.20 & 6.07 & \textbf{2.79} & 4.29 & 4.11 \\
			$\overline{\maxdiam}$ & 6.30 & 8.20 & 11.37 & \textbf{5.58} & 8.58 & 8.22 \\
			Iterations & 5 & 5 & 4 & 5 & 4 & 2 \\
			\hline
			\multicolumn{7}{|c|}{Classical subdivisions}
			\\\hline
			$\overline{\covrad}$ & --- & --- & 5.90 & 2.85 & 4.20 & 4.11 \\
			$\overline{\maxcirc}$ & --- & --- & 61.86 & 9.72 & 6.61 & 4.11 \\
			$\overline{\maxdiam}$ & --- & --- & 24.43 & 8.36 & 9.19 & 8.22 \\
			Iterations & --- & --- & 4 & 5 & 4 & 2 \\
			\hline
		\end{tabular}
	}
\end{table}

\begin{table}[!htbp]
	\ContinuedFloat
	\centering
	\captionsetup[subfloat]{font=normalsize}
	\setlength{\tabcolsep}{4.25pt}
	\renewcommand{\arraystretch}{1.25}
	\subfloat[
			Results on $S^3$ starting from $100$ random samples of $30$ points.
	]{%
		\label{tab:subfloat:3sphere_random}
		\begin{tabular}{|c|cccccc|}
			\hline
			Method
			& Minicenter
			& Centroid
			& Barycentric
			& $2$-Edgewise
			& $3$-Edgewise
			& $5$-Edgewise
			\\\hline
			\multicolumn{7}{|c|}{Delaunay refinements}
			\\\hline
			$\overline{\covrad}$ & $5.3 \pm 0.4$ & $7.5 \pm 1.0$ & $10.3 \pm 1.2$ & $\bm{5.1 \pm 1.2}$ & $5.5 \pm 1.3$ & $5.6 \pm 1.2$ \\
			$\overline{\maxcirc}$ & $5.3 \pm 0.4$ & $7.5 \pm 1.0$ & $10.3 \pm 1.2$ & $\bm{5.1 \pm 1.2}$ & $5.5 \pm 1.3$ & $5.6 \pm 1.2$ \\
			$\overline{\maxdiam}$ & $10.4 \pm 0.8$ & $14.7 \pm 2.1$ & $19.8 \pm 2.4$ & $\bm{9.7 \pm 2.1}$ & $10.5 \pm 2.5$ & $10.8 \pm 2.4$ \\
			Iterations & 4 & 4 & 3 & 4 & 3 & 2 \\
			\hline
			\multicolumn{7}{|c|}{Classical subdivisions}
			\\\hline
			$\overline{\covrad}$ & --- & --- & $10.4 \pm 1.2$ & $5.2 \pm 1.2$ & $5.5 \pm 1.3$ & $5.6 \pm 1.4$ \\
			$\overline{\maxcirc}$ & --- & --- & $107.2 \pm 2.0$ & $52.8 \pm 9.1$ & $64.4 \pm 20.8$ & $45.1 \pm 18.6$ \\
			$\overline{\maxdiam}$ & --- & --- & $48.4 \pm 5.2$ & $14.4 \pm 2.7$ & $15.4 \pm 3.5$ & $15.2 \pm 4.3$ \\
			Iterations & --- & --- & 3 & 4 & 3 & 2 \\
			\hline
		\end{tabular}
	}
	\caption{
		Normalized mesh quantities after iterated refinements and subdivisions on $S^2$ and $S^3$, starting from regular or random configurations.
		The number of iterations depends on the method.
		We report the normalized covering radius ($\overline{\covrad}$), the normalized maximal circumradius ($\overline{\maxcirc}$), and the normalized maximal diameter ($\overline{\maxdiam}$), as defined in \Cref{subsec:normalized_mesh_size}, after the prescribed number of iterations.
		For comparison, the asymptotic reference bounds for $(\overline{\covrad},\overline{\maxcirc},\overline{\maxdiam})$ are $(2.20,2.20,3.81)$ on $S^2$ and $(1.90,1.90,3.81)$ on $S^3$.
		For each experiment, the best---i.e., lowest---values are indicated in bold.
		For the Delaunay refinements, $\overline{\covrad}$ and $\overline{\maxcirc}$ coincide by \Cref{lem:equality_covering_circumradius}.
		}
	\label{tab:mesh_size}
\end{table}

\FloatBarrier

\appendix

\section{Notation}
\label{sec:notation}

\begingroup
\renewcommand{\arraystretch}{1.15}
\begin{longtable}{
		>{\raggedleft\arraybackslash}p{0.18\linewidth}
		@{\hspace{1em}}
		p{0.76\linewidth}
	}
	\hline
	\textbf{Symbol}  & \textbf{Meaning}	\\
	\hline
	
	$\R$, $\R_{\geq 0}$, $\N$, $\N_{>0}$ & Real numbers, nonnegative part, nonnegative integers, positive integers\\
	
	$\R^d$, $B^d$, $S^{d-1}$ & Euclidean space of dimension $d$, its unit ball, its unit sphere \\
	
	$\|x\|$, $\langle x,y\rangle$ & Euclidean norm, Euclidean inner product \\
	
	$\d(x,y)$ & Geodesic distance on the sphere \\
		
	$\conv(X)$, $\aff(X)$ & Convex and affine hulls of a set $X\subset\R^d$ \\
	
	$K$, $\geomreal{K}$, $V(K)$ & Geometric simplicial complex, its geometric realization, its vertices \\
	
	$\sigma$, $\conv(\sigma)$ & Geometric simplex, its convex hull \\
	
	$\diam(\sigma)$ & Euclidean or spherical diameter of a geometric simplex \\
	
	$\maxdiam(K)$ & Maximal diameter of the simplices of a simplicial complex \\
	
	$\circum(\sigma)$ & Spherical circumradius of a simplex \\
	
	$\maxcirc(K)$ & Maximal circumradius of the facets of a simplicial complex \\
	
	$\sub(K)$ & Subdivision of a complex (barycentric or edgewise subdivision) \\
	
	$\Del{X}$ & Spherical Delaunay complex on $X\subset S^{d}$ \\
	
	$\covrad(X)$ & Covering radius of $X\subset S^{d}$ (see \Cref{eq:def_covering_radius})\\
	
	$\PP$, $\E$  & Probability and expectation \\
	\hline
\end{longtable}
\endgroup

\section{Additional proofs}
\label{sec:additional_proofs}

\subsection{Contraction factor for edgewise subdivision}
\label{subsec:additional_proofs_shrinking_classical_edgewise}

We prove the bound stated in \Cref{eq:shrinking_classical_edgewise} in the introduction.

\begin{proposition}
	Consider a $d$-simplex $\sigma=\{v_0,\dots,v_d\}\subset \R^d$ and let $\sub_k(\sigma)$ denote its $k$-fold edgewise subdivision, $k\geq2$. 
	Then
	\[
		\maxdiam(\sub_k(\sigma)) \leq \min\bigg\{1, \frac{\lfloor\frac{d+1}{2}\rfloor}{k}\bigg\} \diam(\sigma).
	\]
	The constant is optimal over all $d$-simplices.
\end{proposition}

\begin{proof}
	Let $\tau=\{w_0,\dots,w_d\}$ be a facet of $\sub_k(\sigma)$ with vertices ordered.
	The inequality $\diam(\tau)\leq\diam(\sigma)$ is obvious.
	We show that $\diam(\tau)\leq\left(\lfloor\frac{d+1}{2}\rfloor/k\right)\diam(\sigma)$.
	Let $(V_i)_{i=1}^d$ denote the shape vectors of $\sigma$, i.e., 
	\[
		V_i = v_i - v_{i-1},\quad i\in\{1,\dots,d\}.
	\]
	Similarly, let $(W_i)_{i=1}^d$ be the shape vectors of $\tau$.
	We use the characterization in the proof of the Independence Lemma in \cite{edelsbrunner2000edgewise}: there exists a permutation $\pi$ of $\{1,\dots,d\}$ such that
	\[
		W_i = \frac{1}{k}V_{\pi(i)},\quad i\in\{1,\dots,d\}.
	\]
	Let $e = w_j - w_i$ be the vector associated with the edge $\{i,j\}$ of $\tau$, with $i<j$.
	We show that its length is at most $\left(\lfloor\frac{d+1}{2}\rfloor/k\right)\diam(\sigma)$.
	It can be written in terms of shape vectors of $\tau$:
	\[
		e = (w_{i+1} - w_i) + (w_{i+2} - w_{i+1}) + \dots + (w_j - w_{j-1}) = \sum_{\ell=i+1}^{j} W_\ell.
	\]
	Expressed in terms of the shape vectors of $\sigma$, it reads
	\[
		e = \frac{1}{k} \sum_{\ell=i+1}^{j} V_{\pi(\ell)}.
	\]
	Given integers $a\leq b$, let us denote the integer interval by $[a,b]=\{a,\dots,b\}$.
	We are interested in the permutation $\pi$ restricted to $[i+1,j]$.
	Its image can be partitioned into disjoint maximal intervals:
	\[
		\left\{ \pi(\ell) \mid \ell \in [i+1,j] \right\}
		= \bigcup_{\ell \in L} [a_\ell, b_\ell],
	\]
	where $L\subset[1,d]$ and
	\begin{align*}
		a_\ell \leq b_\ell \quad&\text{for } \ell \in L,\\
		b_\ell + 1 < a_m \quad&\text{for } \ell, m \in L,~\ell<m.
	\end{align*}
	
	We claim that the cardinality of $L$ is at most $\lfloor(d+1)/2\rfloor$, that is, there are at most $\lfloor(d+1)/2\rfloor$ intervals.
	Indeed, if a subset of $[1,d]$ has $N$ maximal consecutive blocks, then at least $N-1$ integers separate these blocks.
	The number of intervals plus separators is
	\[
		N + (N-1).
	\]
	Since $N + (N-1) \leq d$, we deduce the claim.
	
	Now, each interval $[a_\ell, b_\ell]$ corresponds to an edge of $\sigma$, more precisely, to
	\[
		\sum_{i=a_\ell}^{b_\ell} V_i = v_{b_\ell}-v_{a_\ell-1}.
	\]
	We write the edge $e$ as a combination of these new edges:
	\[
		e = \frac{1}{k}\sum_{\ell=i+1}^{j} V_{\pi(\ell)} = \frac{1}{k} \sum_{\ell\in L} (v_{b_\ell}-v_{a_\ell-1}).
	\]
	Each edge $v_{b_\ell}-v_{a_\ell-1}$ has length at most $\diam(\sigma)$, since it belongs to $\sigma$.
	We deduce that 
	\[
		\|e\| \leq \sum_{\ell\in L} \frac{1}{k} \|v_{b_\ell}-v_{a_\ell-1}\|
		\leq \frac{|L|}{k}\diam(\sigma).
	\]
	Since $|L|\leq \lfloor(d+1)/2\rfloor$, we deduce the result.
	
	To show that the constant is sharp, we build an example.
	Put the points $v_0,\dots,v_d$ on the real line, alternating between 0 and 1.
	Let 
	\[
		\beta = \min\left\{k,\left\lfloor\frac{d+1}{2}\right\rfloor\right\}.
	\]
	For each $1\leq i\leq \beta$, we have
	\[
	v_{2i-1}-v_{2i-2} = 1.
	\]
	We claim that $\sub_k(\sigma)$ contains the edge
	\begin{equation}
		\label{eq:proof_shrinking_classical_edgewise}
		\frac1k
		\sum_{i=1}^\beta
		\left(v_{2i-1}-v_{2i-2}\right).
	\end{equation}
	Indeed, in the color-scheme description \cite{edelsbrunner2000edgewise}, we can take the two columns
	\[
		\big(0,2,4,\dots,2\beta-2,d,\dots,d\big)
		\quad\text{and}\quad
		\big(1,3,5,\dots,2\beta-1,d,\dots,d\big),
	\]
	where we pad with $d$'s if $\beta<k$. 
	These two columns define an edge of $\sub_k(\sigma)$.
	It is precisely the one in \Cref{eq:proof_shrinking_classical_edgewise}.
	Its length is 
	\[
		\left\|
		\frac1k
		\sum_{i=1}^\beta
		\left(v_{2i-1}-v_{2i-2}\right)
		\right\|
		=
		\frac{\beta}{k}.
	\]
	This shows that, in this case,
	\[
		\maxdiam(\sub_k(\sigma)) 
		=
		\min\bigg\{1,\frac{\lfloor\frac{d+1}{2}\rfloor}{k}\bigg\} \diam(\sigma).
		\qedhere
	\]
\end{proof}

\subsection{Normalized mesh size and covering density}
\label{subsec:additional_proofs_covering_density}

We prove \Cref{lem:covering_density,lem:circumradius_diameter_density} stated in \Cref{subsec:normalized_mesh_size}.

\begin{proof}[Proof of \Cref{lem:covering_density}]
	\label{proof:covering_density}
	We use a characterization of the covering density $\vartheta(B^d)$ that follows from \cite[Lemma~2.1]{B_r_czky_2001}.
	For every closed bounded Jordan-measurable subset $J\subset\R^d$ with nonempty interior, if $N(J,r)$ denotes the smallest number of Euclidean balls of radius $r$ needed to cover $J$, then
	\begin{equation}
		\label{eq:asymptotic_euclidean_covering_number}
		\lim_{r\to0}N(J,r)\,r^d
		=
		\frac{\vol(J)}{\vol(B^d)}\vartheta(B^d).
	\end{equation}
	
	Now, let $N(S^d,r)$ denote the smallest number of spherical balls of radius $r$ needed to cover $S^d$.
	To transfer the equality above to the sphere $S^d$, we fix an $\epsilon>0$ and cover it by finitely many Jordan-measurable sets $J_1,\dots,J_n$, such that each $J_i$ is contained in a chart $\phi_i\colon U_i\to\R^d$.
	We can choose the charts to distort the volumes by at most a factor of $1+\epsilon$.
	Applying \Cref{eq:asymptotic_euclidean_covering_number} to each $\phi_i(J_i)$, with increasingly finer coverings $(J_i)_{i=1}^n$, and letting $\epsilon$ tend to zero gives
	\begin{equation}
		\label{eq:asymptotic_spherical_covering_number}
		\lim_{r\to0}N(S^d,r)\,r^d
		=
		\frac{\vol(S^d)}{\vol(B^d)}\vartheta(B^d).
	\end{equation}
	
	We now relate this quantity to the covering radius.
	For every $n\geq1$, let
	\[
		\rho_n = \inf_{\substack{X\subset S^d\\ |X|=n}}
		\covrad(X).
	\]
	The infimum is attained by compactness of $(S^d)^n$.
	Moreover, it satisfies $N(S^d,\rho_n) \leq n$.
	Let
	\[
	A_d=\frac{\vol(S^d)}{\vol(B^d)}\vartheta(B^d).
	\]
	By \Cref{eq:asymptotic_spherical_covering_number}, for every $\epsilon>0$ and all sufficiently small $r$,
	\[
	(A_d-\epsilon)r^{-d}
	\leq
	N(S^d,r)
	\leq
	(A_d+\epsilon)r^{-d}.
	\]
	Since $N(S^d,\rho_n) \leq n$, the first inequality gives
	\begin{equation}
		\label{eq:proof_covering_number1}
		A_d-\epsilon \leq n\rho_n^d.
	\end{equation}
	Conversely, taking
	\[
	r=\left(\frac{A_d+\epsilon}{n}\right)^{1/d}
	\]
	in the second inequality gives $N(S^d,r)\leq n$ for $n$ large enough.
	Hence, by definition of $\rho_n$,
	\begin{equation}
		\label{eq:proof_covering_number2}
		\rho_n\leq
		\left(\frac{A_d+\epsilon}{n}\right)^{1/d}.
	\end{equation}
	\Cref{eq:proof_covering_number1,eq:proof_covering_number2} imply
	\[
	A_d-\epsilon
	\leq
	n\rho_n^d
	\leq
	A_d+\epsilon.
	\]
	Taking the limit as $n\to\infty$ and then as $\epsilon\to0$, we get
	\[
	\lim_{n\to\infty}n^{1/d}\rho_n
	=
	A_d^{1/d}.
	\]
	That is, we have shown that
	\[
	\overline{R}_d
	=
	\liminf_{n\to\infty}
	n^{1/d}\rho_n
	=
	\left(
	\frac{\vol(S^d)}{\vol(B^d)}
	\,\vartheta(B^d)
	\right)^{1/d}.
	\qedhere
	\]
\end{proof}

\begin{proof}[Proof of \Cref{lem:circumradius_diameter_density}]
	\label{proof:circumradius_diameter_density}
	
	We prove the two relations separately.
	
	\paragraph{Equality for the maximal circumradius.}
	
	The first equality follows from the Delaunay construction.
	If $X\subset S^d$ is a sufficiently dense sample of cardinality $n$, then we can build its Delaunay triangulation $\Del{X}$.
	By \Cref{lem:equality_covering_circumradius}, the covering radius of $X$ coincides with the maximal circumradius of $\Del{X}$.
	Taking the limit as $n$ tends to infinity yields $\overline{C}_d\leq \overline{R}_d$.
	
	Conversely, suppose that $K$ is a triangulation of $S^d$.
	We claim that
	\begin{equation}
		\label{eq:circumradius_diameter_density_proof}
		\covrad(V(K))\leq\maxcirc(K).
	\end{equation}
	Indeed, every point $y\in S^d$ lies in a spherical simplex. 
	That is, there exists $\sigma=\{v_0,\dots,v_d\}\subset S^d$ and a scalar $r>0$ such that $ry\in\conv(\sigma)$. 
	Let $(\lambda_0,\dots,\lambda_d)$ be the (Euclidean) barycentric coordinates of $ry$ in $\sigma$.
	Let $w$ and $\Theta$ be the spherical circumcenter and circumradius of $\sigma$.
	They satisfy $\langle v_i,w\rangle=\cos(\Theta)$ for every vertex $v_i$.
	Therefore
	\begin{align*}
		\langle y,ry\rangle &= \sum_{i=0}^d \lambda_i \langle y, v_i\rangle,\\
		\langle w, ry\rangle &= \sum_{i=0}^d \lambda_i \langle w, v_i\rangle = \cos(\Theta).
	\end{align*}
	Moreover, $\langle y,ry\rangle\geq\langle w, ry\rangle$ since both $y$ and $w$ have norm 1.
	This shows that the terms $\langle y, v_i\rangle$ on the right-hand side of the first equation cannot all be less than $\cos(\Theta)$.
	That is, $y$ is within spherical distance $\Theta$ from a vertex $v_i$.
	Since this holds for every $y\in S^d$ with $\Theta\leq \maxcirc(K)$, we deduce \Cref{eq:circumradius_diameter_density_proof}.

	\Cref{eq:circumradius_diameter_density_proof} is valid for all triangulations with $n$ vertices.
	When $n$ goes to infinity, we obtain $\overline{C}_d\geq \overline{R}_d$.
	Combined with $\overline{C}_d\leq \overline{R}_d$, we deduce $\overline{C}_d=\overline{R}_d$.
	
	\paragraph{Inequality for the maximal diameter.}
	
	The first inequality $\overline{R}_d/J_d\leq\overline{D}_d$ comes from Jung's theorem: a Euclidean $d$-simplex of diameter $\delta$ has Euclidean miniradius at most $J_d\,\delta$, where by miniradius we mean the radius of the minimum enclosing ball.
	But we need a spherical version of this result.
	Let $\sigma=\{v_0,\dots,v_d\}\subset S^d$ be a sufficiently small spherical simplex, and let $\Delta$ and $\Theta$ denote its spherical diameter and miniradius.
	Its Euclidean diameter is
	\[
		\delta=2\sin\left(\frac{\Delta}{2}\right).
	\]
	Let $\theta$ be the radius of the Euclidean minimum enclosing ball of $\sigma$. 
	Jung's theorem gives
	\[
		\theta
		\leq J_d\delta
		= 2J_d\sin\left(\frac{\Delta}{2}\right).
	\]
	On the other hand, the Euclidean and spherical miniradii are related by
	\[
	\theta=\sin(\Theta).
	\]
	The spherical Jung's theorem follows:
	\[
		\Theta
		\leq
		\arcsin\left(
		2J_d\,\sin\left(\frac{\Delta}{2}\right)
		\right).
	\]
	For small simplices, we obtain
	\[
		\Theta \leq J_d\,\Delta+o(\Delta).
	\]
	
	Now, for a fine triangulation $K$ of $S^d$, we apply the inequality facet by facet.
	It gives
	\[
		\covrad(V(K)) \leq J_d\,\maxdiam(K) + o(\maxdiam(K)).
	\]
	Indeed, the covering radius is at most the maximal (spherical) miniradius of the facets.
	By taking the limit as the number of vertices goes to infinity, we obtain
	\[
	\overline{R}_d\leq J_d\,\overline{D}_d.
	\]
	
	For the second inequality, we use Delaunay complexes.
	Given a subset $X\subset S^d$ with covering radius $\covrad(X)$, we have seen in \Cref{lem:equality_covering_circumradius} that $\maxcirc(\Del{X})=\covrad(X)$.
	Thus its simplices have diameter at most
	\[
		\maxdiam(\Del{X}) \leq 2\maxcirc(\Del{X}) = 2\covrad(X).
	\]
	We deduce that $\overline{D}_d\leq2\,\overline{R}_d$.
\end{proof}

\bibliographystyle{plainurl}
\bibliography{delaunay_refinement}

\end{document}